\documentclass[pra,onecolumn,superscriptaddress,nofootinbib]{revtex4}
\usepackage[a4paper, left=0.8in, right=0.8in, top=0.9in, bottom=0.9in]{geometry}

\usepackage{cmap} 
\usepackage[utf8]{inputenc}
\usepackage[english]{babel}
\usepackage[T1]{fontenc}
\usepackage{amsmath}
\usepackage{amsfonts}
\usepackage{bm}
\usepackage{xcolor}
\definecolor{blueviolet}{rgb}{0.2, 0.2, 0.6}
\definecolor{webgreen}{rgb}{0,.5,0}
\definecolor{webbrown}{rgb}{.6,0,0}
\usepackage[pdftex,
	bookmarks=false,
	colorlinks=true, 
	urlcolor=webbrown, 
	linkcolor=blueviolet, 
	citecolor=webgreen,
	pdfstartpage=1,
	pdfstartview={FitH},  
	bookmarksopen=false
	]{hyperref}
\allowdisplaybreaks
\usepackage{tikz}
\usepackage{braket}
\usepackage{natbib}
\usepackage{amsthm}
\usepackage{dsfont}
\usepackage{listings}

\DeclareFixedFont{\ttb}{T1}{txtt}{bx}{n}{9} 
\DeclareFixedFont{\ttm}{T1}{txtt}{m}{n}{9}  

\usepackage{color}
\definecolor{deepblue}{rgb}{0,0,0.5}
\definecolor{deepred}{rgb}{0.6,0,0}
\definecolor{deepgreen}{rgb}{0,0.5,0}

\newcommand\pythonstyle{\lstset{
language=Python,
basicstyle=\ttm,
morekeywords={self},              
keywordstyle=\ttb\color{deepblue},
emph={MyClass,__init__},          
emphstyle=\ttb\color{deepred},    
stringstyle=\color{deepgreen},
frame=tb,                         
showstringspaces=false
}}

\lstnewenvironment{python}[1][]
{
\pythonstyle
\lstset{#1}
}
{}

\newcommand\pythoninline[1]{{\pythonstyle\lstinline!#1!}}

\newtheorem{theorem}{Theorem}
\newtheorem{corollary}{Corollary}

\newtheorem{lemma}{Lemma}

\newtheorem{assumption}{Assumption}

\newcommand{\vct}[1]{\boldsymbol{\mathbf{#1}}}
\newcommand{\mtx}[1]{\boldsymbol{\mathbf{#1}}}

\newcommand{\half}{\tfrac{1}{2}}

\newcommand{\uvec}{\hat{u}}

\newcommand{\indicator}{\mathds{1}}

\DeclareMathOperator*{\sign}{sign}

\newtheorem{proposition}{Proposition}
\DeclareMathOperator{\Tr}{tr}
\DeclareMathOperator*{\E}{{\mathbb{E}}}

\newcommand{\ketbra}[2]{\lvert #1 \rangle \! \langle #2 \rvert}
\newcommand{\norm}[1]{\left\lVert#1\right\rVert}

\usepackage[noend]{algpseudocode}
\usepackage{algorithm,algorithmicx}

\algrenewcommand\alglinenumber[1]{\sf\scriptsize\color{blue}{#1}}
\algrenewcommand\algorithmicrequire{\textbf{Input:}}
\algrenewcommand\algorithmicensure{\textbf{Output:}}

\begin{document}

\title{Provably efficient machine learning for quantum many-body problems}
\date{\today}
\author{Hsin-Yuan Huang}
\affiliation{Institute for Quantum Information and Matter and \\ Department of Computing and Mathematical Sciences, Caltech, Pasadena, CA, USA}
\author{Richard Kueng}
\affiliation{Institute for Integrated Circuits, Johannes Kepler University Linz, Austria}
\author{Giacomo Torlai}
\affiliation{AWS Center for Quantum Computing, Pasadena, CA, USA}
\author{Victor V. Albert}
\affiliation{Joint Center for Quantum Information and Computer Science, National Institute of Standards and Technology and University of Maryland, College Park, MD, USA}
\author{John Preskill}
\affiliation{Institute for Quantum Information and Matter and \\ Department of Computing and Mathematical Sciences, Caltech, Pasadena, CA, USA}
\affiliation{AWS Center for Quantum Computing, Pasadena, CA, USA}

\begin{abstract}
Classical machine learning (ML) provides a potentially powerful approach to solving challenging quantum many-body problems in physics and chemistry. However, the advantages of ML over more traditional methods have not been firmly established. 
In this work, we prove that classical ML algorithms can efficiently predict ground state properties of gapped Hamiltonians in finite spatial dimensions, after learning from data obtained by measuring other Hamiltonians in the same quantum phase of matter.
In contrast, under widely accepted complexity theory assumptions, classical algorithms that do not learn from data cannot achieve the same guarantee. We also prove that classical ML algorithms can efficiently classify a wide range of quantum phases of matter. Our arguments are based on the concept of a classical shadow, a succinct classical description of a many-body quantum state that can be constructed in feasible quantum experiments and be used to predict many properties of the state. Extensive numerical experiments corroborate our theoretical results in a variety of scenarios, including Rydberg atom systems, 2D random Heisenberg models, symmetry-protected topological phases, and topologically ordered phases.
\end{abstract}

\maketitle

\section{Introduction}

\begin{figure}[t]
    \centering
    \includegraphics[width=1.0\linewidth]{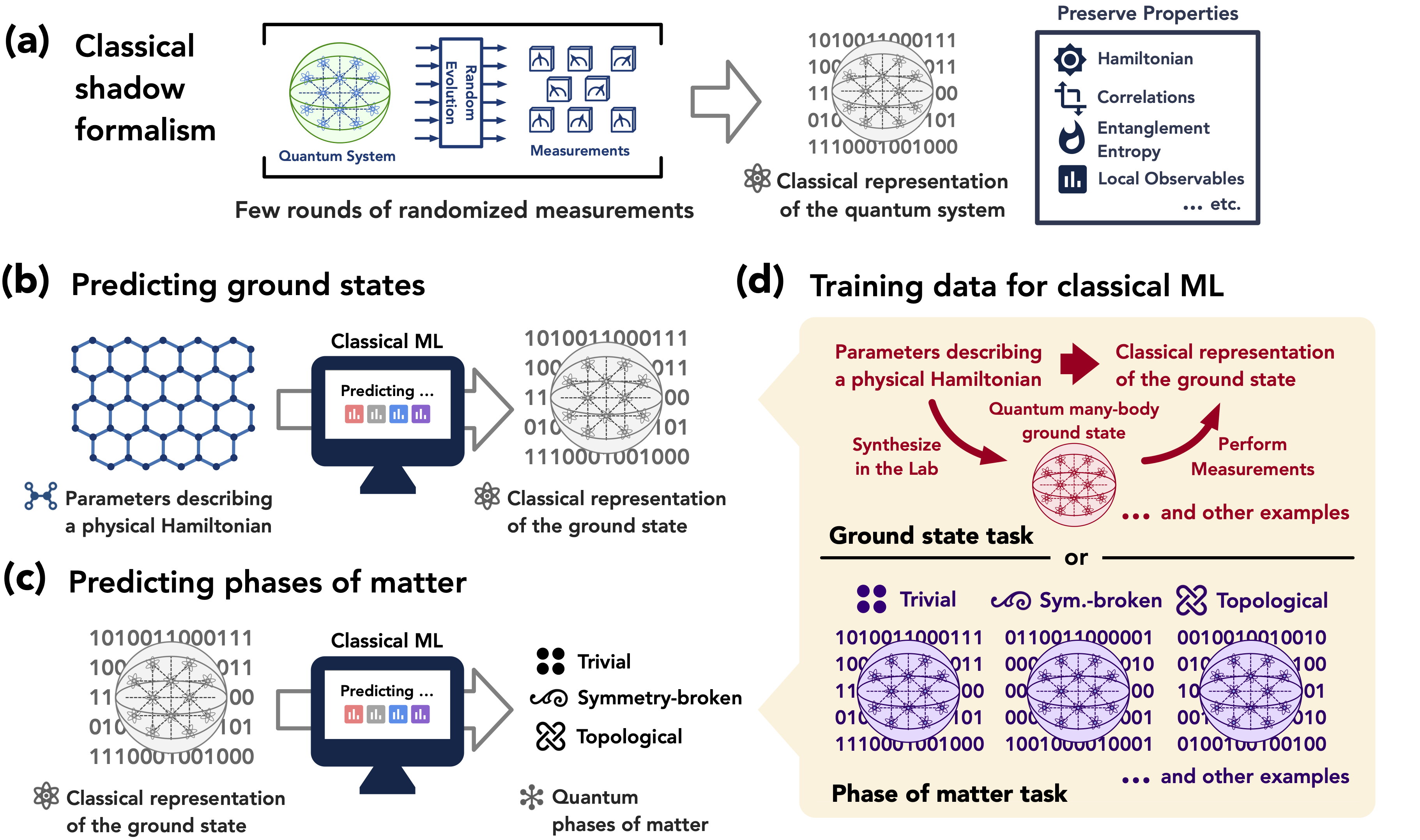}
    \caption{
    (a) \textsc{Efficient quantum-to-classical conversion.} The classical shadow of a quantum state, constructed by measuring very few copies of the state, can be used to predict many properties of the state with a rigorous performance guarantee.
    (b) \textsc{Predicting ground state properties.} 
    After training on data obtained in quantum experiments, a classical ML model predicts a classical representation of the ground state $\rho(x)$ of the Hamiltonian $H(x)$ for parameters $x$ spanning the entire phase. This representation yields estimates of the properties of $\rho(x)$, avoiding the need to run exhaustive classical computations or quantum experiments.
    (c) \textsc{Classifying quantum phases.}
    After training, a classical ML receives a classical representation of a quantum state, and predicts the phase from which the state was drawn.
    (d) \textsc{Training data.} 
    For predicting ground states, the classical ML receives a classical representation of $\rho(x)$ for each value of $x$ sampled during training. 
    For predicting quantum phases of matter, the training data consists of classical representations of quantum states accompanied by labels identifying the phase to which each state belongs.
    }
    \label{fig:main}
\end{figure}

Solving quantum many-body problems, such as finding ground states of quantum systems, has far-reaching consequences for physics, materials science, and chemistry. While classical computers have facilitated many profound advances in science and technology, they often struggle to solve such problems.
Powerful methods, such as density functional theory~\cite{HohenbergKohn, NobelKohn}, quantum Monte Carlo~\cite{CEPERLEY555,SandvikSSE,becca_sorella_2017} and density-matrix renormalization group~\cite{DMRG1,DMRG2}, have enabled solutions to certain restricted instances of many-body problems, but many general classes of problems remain outside the reach of even the most advanced classical algorithms.

Scalable fault-tolerant quantum computers will be able to solve a broad array of quantum problems, but are unlikely to be available for years to come. Meanwhile, how can we best exploit our powerful classical computers to advance our understanding of complex quantum systems?
Recently, classical machine learning (ML) techniques have been adapted to investigate problems in quantum many-body physics~\cite{CarleoRMP,APXReview}, with promising results~\cite{dassarma2017, carrasquilla2017nature,Carleo_2017,torlai_learning_2016,Nomura2017, evert2017nature,leiwang2016,gilmer2017neural,torlai_Tomo,vargas2018extrapolating,schutt2019unifying,Glasser2018,rodriguez2019identifying,qiao2020orbnet,choo_fermionicnqs2020,kawai2020predicting,moreno2020deep,Kottmann2021}.
So far these approaches are mostly heuristic, reflecting the general paucity of rigorous theory in ML. While shown to be effective in some intermediate-size experiments~\cite{Bohrdt2018,Rem2018,torlai_rydberg19}, these methods are generally not backed by convincing theoretical arguments to ensure good performance, particularly for problem instances where traditional classical algorithms falter.

In general, simulating quantum many-body physics is hard for classical computers, because accurately describing an $n$-qubit quantum system may require an amount of classical data that is exponential in $n$. 
In prior work, some of us addressed this bottleneck using  
\emph{classical shadows} --- succinct classical descriptions of quantum many-body states that can be used to accurately predict a wide range of properties with rigorous performance guarantees \cite{huang2020predicting, paini2019approximate}. 
Furthermore, this quantum-to-classical conversion technique can be readily implemented in various existing quantum experiments \cite{Struchalin2021Shadows,Elben2020Entanglement,choi2021emergent}. 
Classical shadows open new opportunities for addressing quantum problems using classical methods such as ML. In this paper, we build on the classical shadow formalism and devise {\color{black} polynomial-time} classical ML algorithms for quantum many-body problems which are supported by rigorous theory.

We consider two applications of classical ML, indicated in Figure~\ref{fig:main}. The first application we examine is learning to predict classical representations of quantum many-body ground states. We consider a family of Hamiltonians, where the Hamiltonian $H(x)$ depends smoothly on $m$ real parameters (denoted by $x$).
The ML algorithm is trained on a set of training data consisting of sampled values of $x$, each accompanied by the corresponding classical shadow for the ground state $\rho(x)$ of $H(x)$. This training data could be obtained from either classical simulations or quantum experiments.
During the prediction phase, the ML algorithm predicts a classical representation of $\rho(x)$ for new values of $x$ different from those in the training data. Ground state properties can then be estimated using the predicted classical representation.

This learning algorithm is efficient, provided that the ground state properties to be predicted do not vary too rapidly as a function of $x$. Indeed, sufficient upper bounds on the gradient can be derived for any family of gapped geometrically-local Hamiltonians in any finite spatial dimension, if the property of interest is the expectation value of a sum of few-body observables. The conclusion is that any such property can be predicted with a small average error, where the amount of training data and the classical computation time are polynomial in $m$ and linear in the system size. Furthermore, we show that classical algorithms that do not learn from data cannot provide the same rigorous guarantee without violating widely accepted complexity-theoretic conjectures.
This is a manifestation of the advantage of ML algorithms with data over those without data \cite{huang2020power}.

{\color{black} If the training data are obtained from quantum experiments, then one might choose to learn about properties of $\rho(x)$ for a new input $x$ by conducting new experiments rather than by using the classical ML to generalize from the training data. However, ML could be far more convenient in some cases, especially when changing some parameters may even require costly re-engineering of the entire experiment.
ML algorithms open up the possibility of efficiently and accurately predicting properties of quantum states that are extremely challenging to prepare and measure in the laboratory.
}

{\color{black}
Classical ML could be used to generalize from training data that are obtained from either quantum experiments or classical simulations; the same rigorous performance guarantees apply in either case. Even if the training data are generated classically, it could be more efficient and more accurate to use ML to predict properties for new values of the input $x$, rather than doing new simulations which could be computationally very demanding and of unverified reliability. Promising insights into quantum many-body physics are already being obtained using classical ML based on classical simulation data \cite{dassarma2017, Nomura2017,Carleo_2017,zhang2017machine,zhang2020interpreting, gilmer2017neural,vargas2018extrapolating,schutt2019unifying,qiao2020orbnet,choo_fermionicnqs2020,kawai2020predicting}. Our rigorous analysis identifies general conditions that guarantee the success of classical ML models, and elucidates the advantages of classical ML models over non-ML algorithms. These results enhance the prospects for interpretable ML techniques \cite{ribeiro2016should, arrieta2020explainable, zhang2020interpreting}  to further shed light on quantum many-body physics.
}

In the second application we examine, the goal is to classify quantum states of matter into phases \cite{Read2012} in a supervised learning scenario. Suppose that during training we are provided with sample quantum states which carry labels indicating whether each state belongs to phase $A$ or phase $B$. Our goal is to predict the phase label for new quantum states that were not encountered during training. We assume that, during both the learning and prediction stages, each quantum state is represented by its classical shadow, which could be obtained either from a classical computation or from an experiment on a quantum device.
The classical ML, then, trains on labeled classical shadows, and learns to predict labels for new classical shadows.

We assume that the $A$ and $B$ phases can be distinguished by a nonlinear function of marginal density operators of subsystems of constant size. This assumption is reasonable because we expect the phase to be revealed in subsystems which are larger than the correlation length, but independent of the total system size. We show that if such a function exists, a classical ML can learn to distinguish the phases using an amount of training data and classical processing which are polynomial in the system size. We do not need to know anything about this nonlinear function in advance, apart from its existence.

In what follows, we review the classical shadow formalism \cite{huang2020predicting}, and use this formalism to derive rigorous guarantees for ML algorithms in predicting ground state properties and classifying quantum phases of matter.
We also describe numerical experiments in a wide range of physical systems to support our theoretical results.

\section{Constructing efficient classical representations of quantum systems}

We begin with an overview of the randomized measurement toolbox \cite{Ohliger_2013,VanEnk2012,Elben2018,evans2019scalable,huang2020predicting,paini2019approximate}, relegating  further details to Appendix~\ref{sec:classical-shadows}.
We approximate an $n$-qubit quantum state $\rho$ by performing randomized single-qubit Pauli measurements on $T$ copies of $\rho$.
That is, we measure every qubit of the unknown quantum state $\rho$ in a random Pauli basis $X, Y$ or $Z$ to yield a measurement outcome of $\pm 1$.
Collapse of the wavefunction implies that this measurement procedure transforms $\rho$ into a random pure product state
$\ket{s^{(t)}} = \bigotimes_{i=1}^n \ket{s^{(t)}_i}$, where $\ket{s^{(t)}_i} \in \{\ket{0}, \ket{1}, \ket{+}, \ket{-}, \ket{\mathrm{i}+}, \ket{\mathrm{i}-}\}$ are eigenstates of the selected Pauli matrices.
Performing one randomized measurement grants us classical access to one such snapshot. 
Performing a total of $T$ randomized measurements grants us access to an entire collection 
$S_T (\rho) = \big\{ |s_i^{(t)} \rangle:\; i \in \{1,\ldots,n\},\; t \in \{1,\ldots,T\}\big\}$.
Each element is a highly structured single-qubit pure state,
and there are $nT$ of them in total. So, $3 n T$ bits suffice to store the entire collection in classical memory.
The randomized measurements can be performed in actual physical experiments or through classical simulations.
Resulting data can then be used to approximate the underlying $n$-qubit state $\rho$:
\begin{equation} \label{eq:sigma-T-shadow}
  \rho\approx\sigma_{T}(\rho)=\frac{1}{T}\sum_{t=1}^{T}\sigma_{1}^{(t)}\otimes\dots\otimes \sigma_{n}^{(t)}
  \quad \text{where} \quad
  \sigma_{i}^{(t)}=3\ketbra{s_{i}^{(t)}}{s_{i}^{(t)}}-\mathbb{I},
\end{equation}
and $\mathbb{I}$ denotes the $2 \times 2$ identity matrix.
This \emph{classical shadow} representation~\cite{huang2020predicting,paini2019approximate} exactly reproduces the global density matrix in the limit $T \to \infty$, but $T = \mathcal{O}(\mathrm{const}^r \log(n) / \epsilon^2)$ already provides an $\epsilon$-accurate approximation of \emph{all} reduced $r$-body density matrices (in trace distance).
This, in turn, implies that we can use $\sigma_T (\rho)$ to predict any function that depends on only reduced  density matrices, such as expectation values of (sums of) local observables and (sums of) entanglement entropies of small subsystems. 
Classical storage and postprocessing cost also remain  tractable in this regime. 
To summarize, the classical shadow formalism equips us with an efficient quantum-to-classical converter that allows classical machines to efficiently and reliably estimate subsystem properties of any quantum state $\rho$.

\section{Predicting ground states of quantum many-body systems}

We consider the task of predicting ground state representations of quantum many-body Hamiltonians in finite spatial dimensions. Suppose that a family of geometrically local, $n$-qubit Hamiltonians $\{H(x):\; x \in [-1,1]^m\}$ is parametrized by a classical variable $x$. That is, $H(x)$ smoothly maps a bounded $m$-dimensional vector $x$ (parametrization) to a Hermitian matrix of size $2^{n} \times 2^{n}$ ($n$-qubit Hamiltonian).
We do not impose any additional structure on this mapping; in particular, we do not assume knowledge about how the physical Hamiltonian depends on the parameterization. 
The goal is to learn a model $\hat{\sigma}(x)$ that can predict properties of the ground state $\rho (x)$ associated with Hamiltonian. 
This problem arises in many practical scenarios. Suppose diligent experimental effort has produced experimental data for ground state properties of various physical systems. 
We would like to use this data to train an ML model that predicts ground state representations of hitherto unexplored physical systems.

\subsection{An ML algorithm with rigorous guarantee}

We will prove that a classical ML algorithm can predict classical representations of ground states after training on data belonging to the same quantum phase of matter. Formally, we consider a smooth family of Hamiltonians $H(x)$ with a constant spectral gap.
During the training phase of the ML algorithm, many values of $x$ are randomly sampled, and for each sampled $x$, the classical shadow of the corresponding ground state $\rho(x)$ of $H(x)$ is provided, either by classical simulations or quantum experiments.
The full training data of size $N$ is given by $\big\{x_{\ell} \rightarrow \sigma_T(\rho(x_{\ell}))\big\}_{\ell = 1}^N$, where $T$ is the number of randomized measurements in the construction of the classical shadows at each value of $x_\ell$.

{\color{black} We train classical ML models using the size-$N$ training data, such that when given the input $x_\ell$, the ML model can produce a classical representation $\hat{\sigma}(x)$ that approximates $\sigma_T(\rho(x_{\ell}))$.
During prediction, the classical ML produces $\hat{\sigma}(x)$ for new values of $x$ different from those in the training data.
While $\hat{\sigma}(x)$ and $\sigma_T(\rho(x_{\ell}))$ classically represent exponentially large density matrices, the training and prediction can be done efficiently on a classical computer using various existing classical ML models, such as neural networks with large hidden layers \cite{jacot2018neural, li2019enhanced, du2019graph, neuraltangents2020} and kernel methods \cite{cortes1995support,CC01a}. In particular, the predicted output of the trained classical ML models can be written as the extrapolation of the training data using a learned metric $\kappa(x, x_{\ell}) \in \mathbb{R}$,
\begin{equation} \label{eq:sigma(x)}
    \hat{\sigma}(x) = \frac{1}{N} \sum_{\ell = 1}^N \kappa(x, x_{\ell}) \sigma_T(\rho(x_{\ell})).
\end{equation}
For example, prediction using a trained neural network with large hidden layers \cite{jacot2018neural} is equivalent to using the metric $\kappa(x, x_{\ell}) = \sum_{\ell' = 1}^N f^{\mathrm{(NTK)}}(x, x_{\ell'}) (F^{-1})_{\ell' \ell}$, where $f^{\mathrm{(NTK)}}(x, x')$ is the neural tangent kernel \cite{jacot2018neural} corresponding to the neural network and $F_{\ell' \ell} = f^{\mathrm{(NTK)}}(x_{\ell'}, x_{\ell})$; see Appendix~\ref{sec:neuralnetwork} for more discussion.}
The ground state properties are then estimated using these predicted classical representations $\hat{\sigma}(x)$.
Specifically, $f_O(x) = \mathrm{tr}\left( O \rho(x)\right)$ can be predicted efficiently whenever $O$ is a sum of few-body operators.

To derive a provable guarantee, we consider the simple metric $\kappa(x, x_{\ell}) = \sum_{k \in \mathbb{Z}^m, \norm{k}_2 \leq \Lambda} \cos(\pi k \cdot (x - x_{\ell}))$ with cutoff $\Lambda$, which we refer to as the $l_2$-Dirichlet kernel.
We prove that the prediction will be accurate and efficient if the function $f_O(x)$ does not vary too rapidly when $x$ changes in any direction.
Indeed, sufficient upper bounds on the gradient magnitude of $f_O(x)$ can be derived 
using quasi-adiabatic continuation \cite{hastings2005quasiadiabatic, bachmann2012automorphic}.

Under the $l_2$-Dirichlet kernel, the classical ML model is equivalent to learning a truncated Fourier series to approximate the function $f_O(x)$. The parameter $\Lambda$ is a cutoff for the wavenumber $k$ that depends on (upper bounds on) the gradient of $f_O (x)$.
Using statistical analysis, one can guarantee that $ \E_{x} |\Tr(O \hat{\sigma}(x)) - f_O(x)|^2 \leq \epsilon$ as long as the amount of training data obeys $N = m^{\mathcal{O}(1 / \epsilon)}$ in the $m \to\infty$ limit.
The conclusion is that any such $f_O(x)$ can be predicted with a small \emph{constant} average error, where the amount of training data and the classical computation time are polynomial in $m$ and at most linear in the system size $n$. Moreover, the training data need only contain a \textit{single} classical shadow snapshot at each point $x_\ell$ in the parameter space (i.e., $T=1$). An informal statement of the theorem is given below; we explain the proof strategy in Appendix~\ref{sec:proofideaGSUPP}, and provide more details in Appendix~\ref{sec:proofthmGSUPP}.
{\color{black}
We also discuss how one could generalize the proof to long-range interacting systems, electronic Hamiltonians, and other settings in Appendix~\ref{sec:generalize-GSUPP}. }

\begin{theorem}[Learning to predict ground state representations; informal] \label{thm:mainFourier}
For any smooth family of Hamiltonians $\{H(x):\; x\in[-1,1]^m\}$ in a finite spatial dimension with a constant spectral gap, the classical machine learning algorithm can learn to predict a classical representation of the ground state $\rho(x)$ of $H(x)$ that approximates few-body reduced density matrices up to a constant error $\epsilon$ when averaged over $x$. The required training data size $N$ and computation time are polynomial in $m$ and linear in the system size $n$.
\end{theorem}

Though formally ``efficient'' in the sense that $N$ scales polynomially with $m$ for any fixed approximation error $\epsilon$, the required amount of training data scales badly with $\epsilon$. This unfortunate scaling is not a shortcoming of the considered ML algorithm, but a necessary feature. In Appendix~\ref{app:proofthmlowerbound}, we show that the data size and time complexity cannot be improved further without making stronger assumptions about the class of gapped local Hamiltonians.
However, in cases of practical interest, the Hamiltonian may obey restrictions such as translational invariance or graph structure that can be exploited to obtain better results.
Incorporating these restrictions can be achieved by using a suitable $\kappa(x, x_\ell)$, such as one that corresponds to a large-width convolutional neural network \cite{li2019enhanced} or a graph neural network \cite{du2019graph}.
Rigorously establishing that neural-network-based ML algorithms can achieve improved prediction performance and efficiency for particular classes of Hamiltonians is a goal for future work. 

\subsection{Computational hardness for non-ML algorithms}

In the following proposition, we show that a classical algorithm that does not learn from data cannot achieve the same guarantee in estimating ground state properties without violating the widely believed conjecture that NP-complete problems cannot be solved in randomized polynomial time. This proposition is a corollary of standard complexity-theoretic results \cite{Lichtenstein1982PlanarFA, VALIANT198685}. See Appendix~\ref{sec:proofhardnonML} for the detailed statement and proof.

\begin{proposition}[Informal] \label{thm:hardnonML}
Consider a randomized polynomial-time classical algorithm $\mathcal{A}$ that does not learn from data. Suppose for any smooth family of two-dimensional Hamiltonians $\{H(x):\; x\in[-1,1]^m\}$ with a constant spectral gap, $\mathcal{A}$ can efficiently compute expectation values of one-body observables in the ground state $\rho(x)$ of $H(x)$ up to a constant error when averaged over $x$. Then there is a randomized classical algorithm that can solve NP-complete problems in polynomial time.
\end{proposition}

It is instructive to observe that a classical ML algorithm with access to data can perform tasks that cannot be achieved by classical algorithms which do not have access to data.
This phenomenon is studied in \cite{huang2020power}, where it is shown that the complexity class defined by classical algorithms that can learn from data is strictly larger than the class of classical algorithms that do not learn from data. (The data can be regarded as a restricted form of randomized advice string.)
We caution that obtaining the data to train the classical ML model could be challenging.
However, if we focus only on data that could be efficiently generated by quantum-mechanical processes, it is still possible that a classical ML that learns from data could be more powerful than classical computers.
In Appendix~\ref{sec:proofhardnonML} we present a contrived family of Hamiltonians that establishes this claim, based on the (classical) computational hardness of factoring.

\section{Classifying quantum phases of matter}

Classifying quantum phases of matter is another important application of machine learning to physics.
We will consider this classification problem in the case where quantum states are succinctly represented by their classical shadows. 
For simplicity, we consider the classification of two phases (denoted $A$ and $B$), but the analysis naturally generalizes to classifying any number of phases.

\subsection{ML algorithms}

We envision training a classical ML with classical shadows, where each classical shadow carries a label $y$ indicating whether it represents a quantum state $\rho$ from phase $A$ $(y(\rho) = 1)$ or phase $B$ $(y(\rho) = -1)$. We want to show that a suitably chosen classical ML can learn to efficiently predict the phase for new classical shadows beyond those encountered during training. Following a strategy which is standard in learning theory, we consider a classical ML that maps each classical shadow to a corresponding feature vector in a high-dimensional feature space, and then attempts to find a hyperplane that separates feature vectors in the $A$ phase from feature vectors in the $B$ phase. The learning is efficient if the geometry of the feature space is efficiently computable, and if the feature map is sufficiently expressive. Thus, our task is to construct a feature map with the desired properties. 

In the simpler task of classifying symmetry-breaking phases, there is typically a local order parameter $O = \sum_i O_i$ given as a sum of $r$-body observables for some $r > 0$ that satisfies
\begin{equation} \label{eq:observPhase}
    \Tr(O \rho) \geq 1, \forall \rho \in \mbox{phase $A$}, \quad\quad \Tr(O \rho) \leq -1, \forall \rho \in \mbox{phase $B$}.
\end{equation} 
Under this criterion, the classification function may be chosen to be $y(\rho) = \sign(\Tr(O \rho))$.
Hence, classifying symmetry-breaking phases can be achieved by finding a hyperplane that separates the two phases in the high-dimensional feature space 
that subsumes all $r$-body reduced density matrices of the quantum state $\rho$.
The feature vector consisting of all $r$-body reduced density matrices of the quantum state $\rho$ can be accurately reconstructed from the classical shadow representation $S_T(\rho)$ when $T$ is sufficiently large.

Finding a suitable choice of hyperplane in the feature space can be cast as a convex optimization problem known as the soft-margin support vector machine, discussed in more detail in Appendix~\ref{sec:trainSVM}.
With a sufficient amount of training data, the hyperplane found by the classical ML model will generalize so the phase $y(\rho)$ can be predicted accurately for a previously unseen quantum state $\rho$.
The classical ML is not merely a black box; it exhibits the order parameter (encoded by the hyperplane), guiding physicists toward a deeper understanding of the phase structure.

For more exotic quantum phases of matter, such as topologically ordered phases, the above classical ML model no longer suffices. The topological phase of a state is invariant under a constant-depth quantum circuit, and a phase containing the product state $\ket{0}^{\otimes n}$ is called the trivial phase.
Using these notions, we can prove that no observable --- not even one that acts on the entire system --- can be used to distinguish between two topological phases.
The proof, given in Appendix~\ref{sec:no-nonlocal-obs}, uses the observation that random single-qubit unitaries can confuse any global or local order parameter.
\begin{proposition} \label{prop:noobs}
Consider two distinct topological phases $A$ and $B$ (one of the phases could be the trivial phase). No observable $O$ exists such that
\begin{equation}
    \Tr(O \rho) > 0, \forall \rho \in \mbox{phase $A$}, \quad\quad \Tr(O \rho) \leq 0, \forall \rho \in \mbox{phase $B$}.
\end{equation}
\end{proposition}
\noindent While this proposition implies that no linear function $\Tr(O \rho)$ can be used to classify topologically ordered phases, it does not exclude nonlinear functions, such as quadratic functions $\Tr(O \rho \otimes \rho)$, degree-$d$ polynomials $\Tr(O \rho^{\otimes d})$ and more general analytic functions.
For example, it is known that the topological entanglement entropy \cite{kitaev2006topological,levin2006detecting}, a nonlinear function of $\rho$, can be used to classify a wide variety of topologically ordered phases. For this purpose, it suffices to consider a subsystem whose size is large compared to the correlation length of the state, but is independent of the total size of the system.
The correlation length in the ground state of a local Hamiltonian increases when the spectral gap between the ground state and the first excited state becomes smaller \cite{hastings2006spectral}.
On the other hand, a linear function on the full system will fail even with constant correlation length.

To learn nonlinear functions, we need a more expressive ML model. For this purpose we devise a powerful feature map that takes the classical shadow $S_T(\rho)$ of the quantum state $\rho$ to a feature vector that includes arbitrarily-large $r$-body reduced density matrices, as well as an arbitrarily-high-degree polynomial expansion,
\begin{equation}\label{eq:featureshadow}
    \phi^{\mathrm{(shadow)}}(S_T(\rho)) =
    \lim_{D,R \to \infty}
    \bigoplus_{d=0}^D \sqrt{\frac{\tau^d}{d!}} \left(\bigoplus_{r=0}^R \sqrt{\frac{1}{r!} \left(\frac{\gamma}{n}\right)^{r}} \bigoplus_{i_1 = 1}^n \ldots \bigoplus_{i_r = 1}^n \mathrm{vec}\left[ \frac{1}{T}\sum_{t=1}^T \bigotimes_{\ell=1}^r \sigma_{i_{\ell}}^{(t)} \right]\right)^{\otimes d},
\end{equation}
where $\tau, \gamma > 0$ are hyper-parameters.
The direct sum $\bigoplus_{r=0}^R$ is a concatenation of all $r$-body reduced density matrices, and the other direct sum $\bigoplus_{d=0}^D$ subsumes all degree-$d$ polynomial expansions.
The computational cost of finding a hyperplane 
in feature space that separates the training data into two classes is dominated by the cost of computing inner products between feature vectors.
The inner product $\langle \phi^{\mathrm{(shadow)}}(S_T(\rho)), \phi^{\mathrm{(shadow)}}(S_T(\tilde{\rho})) \rangle $ can be analytically computed by reorganizing the direct sums,
writing it as a double series, and wrapping both series into an exponential, which gives
\begin{equation} \label{eq:doubleinfkernel}
k^{\mathrm{(shadow)}}\left( S_T(\rho), S_T(\tilde{\rho}) \right)
=     \exp \left(\frac{\tau}{T^2} \sum_{t,t'=1}^T \exp\left( \frac{\gamma}{n} \sum_{i=1}^n \Tr\left( \sigma_{i}^{(t)}\tilde{\sigma}_{i}^{(t^{\prime})} \right) \right)\right),
\end{equation}
where $S_T (\rho)$ and $S_T (\tilde{\rho})$ are classical shadow representations of $\rho$ and $\tilde{\rho}$, respectively.
The computation time for the inner product is $\mathcal{O}(n T^2)$, linear in the system size $n$ and quadratic in $T$, the number of copies of each quantum state which are measured to construct the classical shadow.

\subsection{Rigorous guarantee}

By statistical analysis, we can establish a rigorous guarantee for the classical ML model $\langle \alpha, \phi^{\mathrm{(shadow)}}(S_T(\rho)) \rangle$, where $\alpha$ is the trainable vector defining the classifying hyperplane. The result is the following theorem proven in Appendix~\ref{sec:proofthmPHASEC}.

\begin{theorem}[Classifying quantum phases of matter; informal] \label{thm:phaseclassification}
If there is a nonlinear function of few-body reduced density matrices that classifies phases, then the classical algorithm can learn to classify these phases accurately. The required amount of training data and computation time scale polynomially in system size.
\end{theorem}

\noindent If there is an efficient procedure based on \emph{few-body reduced density matrices} for classifying phases, the proposed ML algorithm is guaranteed to find the procedure efficiently.
This includes local order parameters for classifying symmetry breaking phases, and topological entanglement entropy in a sufficiently large local region for partially classifying topological phases \cite{kitaev2006topological,levin2006detecting}.
We expect that, to classify topological phases accurately, the classical ML will need access to local regions that are sufficiently large compared to the correlation length, and
as we approach the phase boundary, the correlation length increases. As a result, the classifying function for topological phases may depend on $r$-body subsystems with a larger $r$, and the amount of training data and computation time required would increase accordingly.
Note that the classical ML not only classifies phases accurately, but also constructs a classifying function explicitly.

Our classical ML model may also be useful for classifying and understanding symmetry-protected topological (SPT) phases. SPT phases are characterized much like topological phases, but with the additional constraint that all structures involved (states, Hamiltonians, and quantum circuits) respect a particular symmetry. 
It is reasonable to expect that an SPT phase can be identified by examining reduced density matrices on constant-size regions \cite{Li2008,Pollmann2010,pollmann2012detection,Haegeman2012,Shapourian2017,Dehghani2021}, where the size of the region is large compared to the correlation length.
The existence of classifying functions based on reduced matrices have been rigorously established in some cases \cite{Kitaev2006,zhang2017quantum,Hastings2015,Kapustin2020,Bachmann2020,Bachmann2014a,tasakiPRL2018,Tasaki2020}.
In Appendix~\ref{sec:tasaki}, we prove that the ML algorithm is guaranteed to efficiently classify a class of gapped spin-$1$ chains in one dimension.
For more general SPT phases, the ML algorithm should be able to corroborate known classification schemes, determine new and potentially more compact classifiers, and shed light on interacting SPT phases in two or more dimensions for which complete classification schemes have not yet been firmly established.

{\color{black} The hypothesis of Theorem \ref{thm:phaseclassification}, stating that phases can be recognized by inspecting regions of constant size independent of the total system size, is particularly plausible for gapped phases, but might apply to some gapless phases as well. Our classical ML model would be able to efficiently classify such gapless phases. On the other hand, the contrapositive of Theorem \ref{thm:phaseclassification} asserts that if the classical ML is not able to distinguish between two distinct gapless phases, then nonlocal data is required to characterize at least one of those phases.
}

\section{Numerical experiments} \label{sec:numerics}

\begin{figure}[t]
    \centering
    \includegraphics[width=1.0\linewidth]{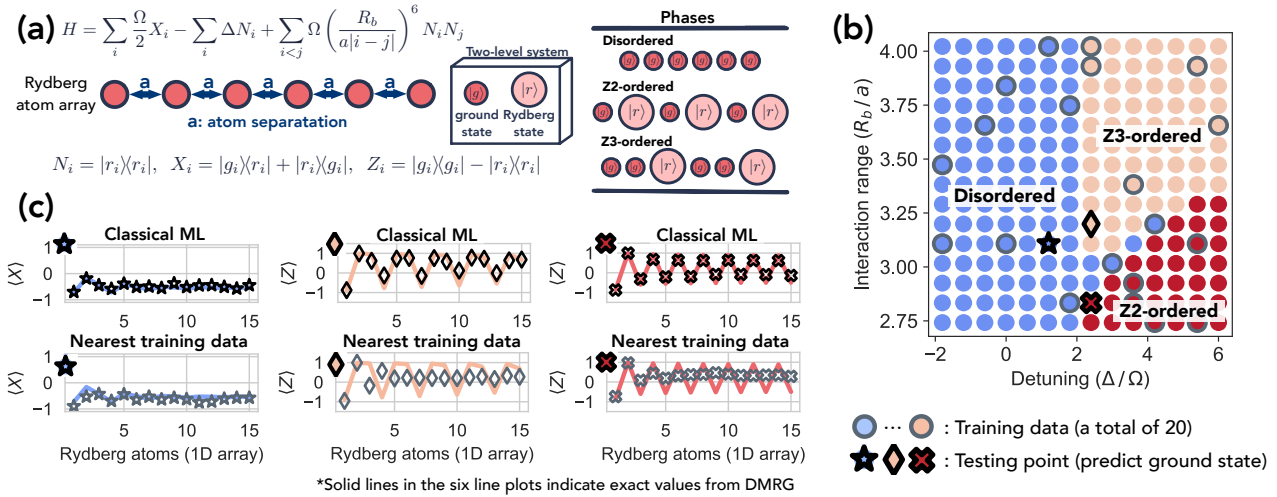}
    \caption{Numerical experiment for predicting ground-state properties in a 1D Rydberg atom system with 51 atoms. (a) \textsc{Hamiltonian.} Illustration of the Rydberg array geometry, Hamiltonian, and phases. (b) \textsc{Phase diagram.} The system's three distinct phases \cite{bernien2017probing} are characterized by two order parameters (for $Z_2$ and $Z_3$ orders). Training data are enclosed by gray circles, and three specific testing points are indicated by the star, diamond, and cross, respectively. (c) \textsc{Local expectation values.} We use classical ML (the best model is selected from a set of ML models) to predict the expectation values of Pauli operators $X_i$ and $Z_i$ for each atom at the three testing points. We compare with ``predictions'' obtained from the training data nearest to the testing points. The markers denote predicted values, while the solid lines denote exact values obtained from DMRG. Additional predictions are shown in Appendix~\ref{sec:more-numer}. }
    \label{fig:rydberg}
\end{figure}

We have conducted numerical experiments assessing the performance of classical ML algorithms in some practical settings. The results demonstrate that our theoretical claims carry over to practice, with the results sometimes turning out even better than our guarantees suggest. 

\subsection{Predicting ground state properties}

{\color{black} For predicting ground states, we consider classical ML models encompassed by Eq.~\eqref{eq:sigma(x)}. We examine various metrics $\kappa(x, x_{\ell})$ equivalent to training neural networks with large hidden layers \cite{jacot2018neural, neuraltangents2020} or training kernel methods \cite{cortes1995support, murphy2012machine}. We find the best ML model and the hyperparameters using a validation set to minimize root-mean-square error (RMSE) and report the predictions on a test set. The full details of the models and hyperparameters, as well as their comparison, are given in Appendix~\ref{sec:numdetail-groundstate}~and~\ref{sec:numdetail-groundstate2}. }

\emph{Rydberg atom chain} --- Our first example is a system of trapped Rydberg atoms~\cite{Fendley2004,browaeys_many-body_2020}, a programmable and highly-controlled platform for Ising-type quantum simulations~\cite{Schauss1455,Endres2016,bernien2017probing,Labuhn,Misha256,2020arXiv201212268S}. Following~\cite{bernien2017probing}, we consider a one-dimensional array of $n=51$ atoms, with each atom effectively described as a two-level system composed of a ground state $|g\rangle$ and a highly-excited Rydberg state $|r\rangle$. The atomic chain is characterized by a Hamiltonian $H(x)$ (given in Figure~\ref{fig:rydberg}(a)) whose parameters are the laser detuning $x_1 = \Delta / \Omega$ and the interaction range $x_2 = R_b / a$. The phase diagram (shown in Figure~\ref{fig:rydberg}(b)) features a 
disordered phase and several broken-symmetry phases,
stemming from the competition between the detuning and the Rydberg blockade (arising from the repulsive Van der Waals interactions).

We trained a classical ML model using $20$ randomly chosen values of the parameter $x=(x_1,x_2)$; these values are indicated by gray circles in Figure~\ref{fig:rydberg}(b). For each such $x$, an approximation to the exact ground state was found using DMRG~\cite{DMRG1} based on the formalism of matrix product states (MPS)~\cite{SCHOLLWOCK201196}. For each MPS, we performed $T=500$ randomized Pauli measurements to construct a classical shadow. The classical ML then predicted classical representations at the testing points in the parameter space, and these predicted classical representations were used to estimate expectation values of local observables at the testing points.

Predictions for expectation values of Pauli operators $Z_i$ and $X_i$ at the testing points are shown in Figure~\ref{fig:rydberg}(c), and found to agree well with exact values obtained from the DMRG computation of the ground state at the testing points. Additional predictions can be found in Appendix~\ref{sec:more-numer}. Also shown are results from a more naive procedure, in which properties are predicted using only the data at the point in the training set which is closest to the testing point. The naive procedure predicts poorly, illustrating that the considered classical ML model effectively leverages the data from multiple points in the training set.

This example corroborates our expectation that classical machines can learn to efficiently predict ground state representations. An important caveat is that the rigorous guarantee in Theorem \ref{thm:mainFourier} applies only when the training points and the testing points are sampled from the same phase, while in this example the training data includes values of $x$ from three different phases. Nevertheless, our numerics show that classical machines can still learn to predict well. 

\begin{figure}[t]
    \centering
    \includegraphics[width=1.0\linewidth]{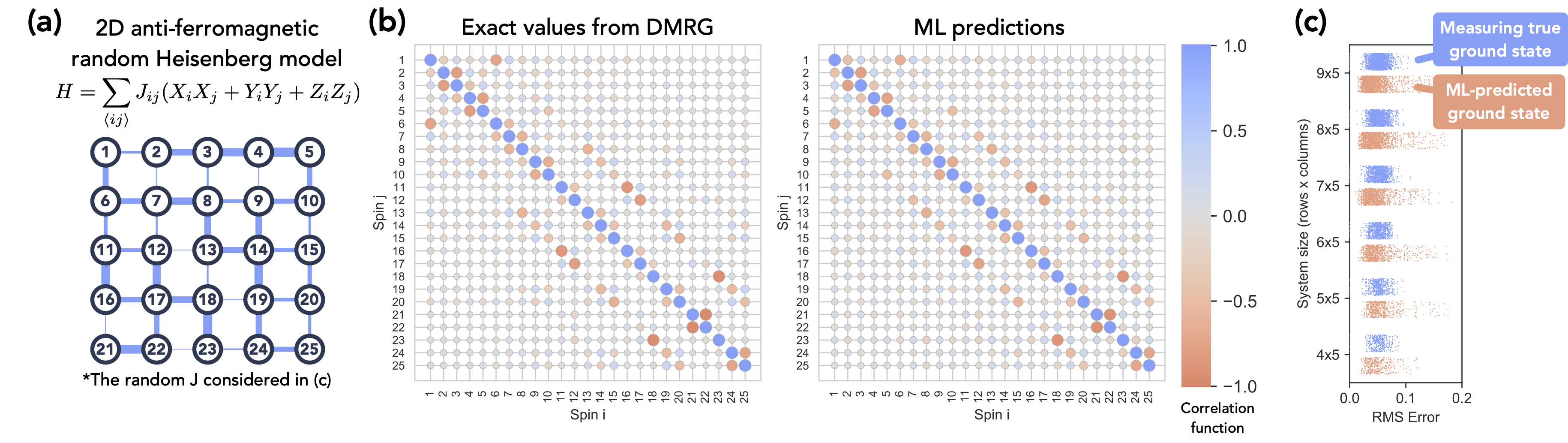}
    \caption{Numerical experiment for predicting ground state properties in the 2D antiferromagnetic Heisenberg model. (a) \textsc{Hamiltonian.} Illustration of the Heisenberg model geometry and Hamiltonian. We consider random couplings $J_{ij}$, sampled uniformly from $[0, 2]$. A particular instance is shown, with coupling strength indicated by the thickness of the edges connecting lattice points. (b) \textsc{Two-point correlator.} 
    Exact values and ML predictions of the expectation value of the correlation function $C_{ij}=\tfrac{1}{3} (X_i X_j + Y_i Y_j + Z_i Z_j)$ for all spin pairs $(ij)$ in the lattice, for the Hamiltonian instance shown in (a). The absolute value of $C_{ij}$ is represented by the size of each circle, while the circle's color indicates the actual value. (c) \textsc{Prediction error.} Each blue point indicates the root-mean-square error (averaged over Heisenberg model instances) of the correlation function for a particular pair $(ij)$, where the estimate of $C_{ij}$ is obtained using a classical shadow with $T=500$ randomized Pauli measurements of the true ground state. Red points indicate errors in ML predictions for $C_{ij}$.}
    \label{fig:heisenberg}
\end{figure}

\emph{2D antiferromagnetic Heisenberg model} --- Our next example is the two-dimensional antiferromagnetic Heisenberg model. Spin-$\tfrac{1}{2}$ particles (i.e. qubits) occupy sites on a square lattice, and for each pair $(ij)$ of neighboring sites the Hamiltonian contains a term $J_{ij}\left(X_iX_j+Y_iY_j+Z_iZ_j\right)$ where the couplings $\{J_{ij}\}$ are uniformly sampled from the unit interval $[0,2]$. The parameter $x$ is a list of all  $J_{ij}$ couplings; hence in this case the dimension of the parameter space is $m=O(n)$, where $n$ is the number of qubits. The Hamiltonian $H(x)$ on a $5\times 5$ lattice is shown in Figure~\ref{fig:heisenberg}(a).  

We trained a classical ML model using 90 randomly chosen values of the parameter $x=\{J_{ij}\}$. For each such $x$, the exact ground state was found using DMRG, and we simulated $T=500$ randomized Pauli measurements to construct a classical shadow. The classical ML predicted the classical representation at new values of $x$, and we used the predicted classical representation to estimate a two-body correlation function, the expectation value of $C_{ij}=\frac{1}{3}\left(X_iX_j + Y_iY_j+ Z_i Z_j\right)$, for each pair of qubits $(ij)$. In Figure~\ref{fig:heisenberg}(b), the predicted and actual values of the correlation function are displayed for a particular value of $x$, showing reasonable agreement.

Figure~\ref{fig:heisenberg}(c) shows the prediction performance for all pairs of spins and for variable system size. Each red point in the plot represents the RMSE in the correlation function estimated using our predicted classical representation, for a particular pair of spins and averaged over sampled values of $x$. For comparison, each blue point is the RMSE when the correlation function is predicted using the classical shadow obtained by measuring the actual ground state $T=500$ times.
For most correlation functions, the prediction error achieved by the best classical ML model is comparable to the error achieved by measuring the actual ground state.

\subsection{Classifying quantum phases of matter}\label{subsec:numerics-classification}

For classifying quantum phases of matter, we consider an unsupervised classical ML model that constructs an {\color{black} infinite-dimensional \emph{nonlinear}} feature vector for each quantum state $\rho$ by applying the map $\phi^{\mathrm{(shadow)}}$ in Eq.~\eqref{eq:featureshadow} with $\tau, \gamma = 1$ to the classical shadow $S_T(\rho)$ of the quantum state $\rho$. We then perform a principal component analysis (PCA) \cite{pearson1901liii} in the {\color{black} infinite-dimensional \emph{nonlinear}} feature space.
{\color{black} The low-dimensional subspace found by PCA in the nonlinear feature space corresponds to a nonlinear low-dimensional manifold in the original quantum state space.}
This method is efficient using the shadow kernel $k^{\mathrm{(shadow)}}$ given in Eq.~\eqref{eq:doubleinfkernel} and the kernel PCA procedure \cite{scholkopf1998nonlinear}. Details are given in Appendix~\ref{sec:numdetail-phases}~and~\ref{sec:numdetail-phases2}.

\textit{Bond-alternating XXZ model} --- We begin by considering the bond-alternating XXZ model with $n=300$ spins. The Hamiltonian is given in Figure~\ref{fig:XXZ}(a); it encompasses the bond-alternating Heisenberg model ($\delta=1$) and the bosonic version of the Su-Schrieffer-Heeger model  ($\delta=0$)~\cite{su1979solitons}.
The phase diagram in Figure~\ref{fig:XXZ}(b) is obtained by evaluating the partial reflection many-body topological invariant \cite{pollmann2012detection, elben2020many}.
There are three different phases: trivial, symmetry-protected topological, and symmetry broken.

For each value of $J$ and $\delta$ considered, we construct the exact ground state using DMRG, and find its classical shadow by performing randomized Pauli measurement $T=500$ times.
We then consider a two-dimensional principal subspace of the {\color{black} infinite-dimensional nonlinear} feature space found by the unsupervised ML based on the shadow kernel, which is visualized in Figure~\ref{fig:XXZ}(c, d).
We can clearly see that the different phases are well separated in the principal subspace. This shows that even without any phase labels on the training data, the ML model can already classify the phases accurately. Hence, when trained with only a small amount of labeled data, the ML model will be able to correctly classify the phases as guaranteed by Theorem~\ref{thm:phaseclassification}.

\begin{figure}[t]
    \centering
    \includegraphics[width=1.0\linewidth]{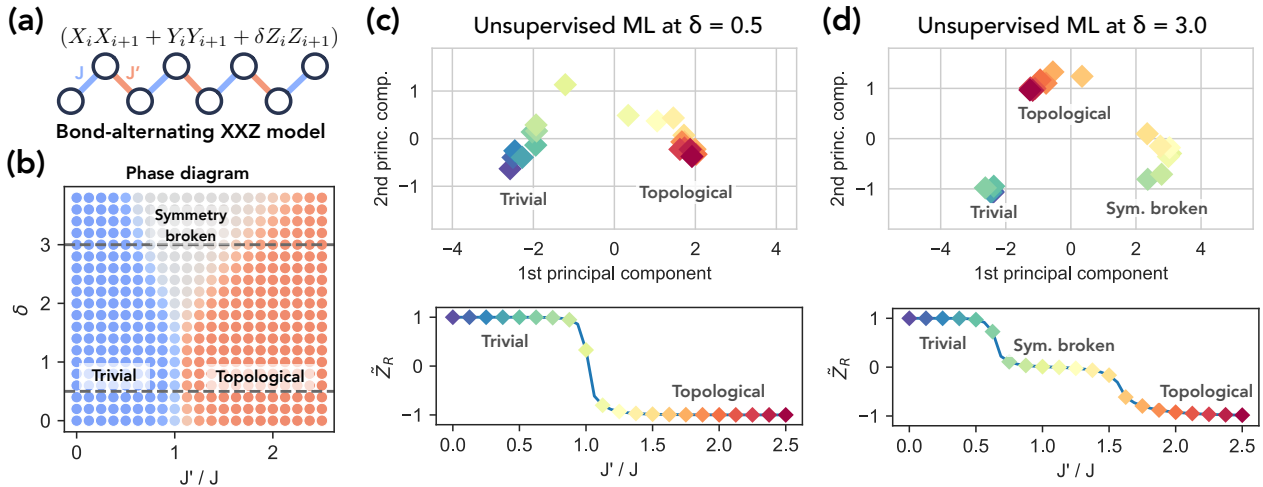}
    \caption{Numerical experiments for classifying quantum phases in the bond-alternating XXZ model. (a) \textsc{Hamiltonian.} Illustration of the model --- a one-dimensional qubit chain, where the coefficient of $(X_i X_{i+1} + Y_i Y_{i+1} + \delta Z_i Z_{i+1})$ alternates between $J$ and $J'$. (b)  \textsc{Phase diagram.} The system's three distinct phases are characterized by the many-body topological invariant $\tilde{Z}_R$ discussed in Refs.~\cite{pollmann2012detection, elben2020many}. Blue denotes $\tilde{Z}_R = 1$, red denotes $\tilde{Z}_R = -1$, and gray denotes $\tilde{Z}_R \approx 0$. (c,~d) \textsc{Unsupervised phase classification.} Bottom panels: $\tilde{Z}_R$ vs. $J'/J$ at cross sections (c) $\delta=0.5$ and (d) $\delta=3.0$ of the phase diagram. Top panels: visualization of the quantum states projected to two dimensions using the unsupervised ML ({\color{black} nonlinear} PCA with shadow kernel). In all panels, colors of the points indicate the value of $J' / J$, indicating that the two phases naturally cluster in the expressive feature space.}
    \label{fig:XXZ}
\end{figure}

\emph{Distinguishing a topological phase from a trivial phase} --- We consider the task of distinguishing the toric code topological phase from the trivial phase in a system of $n=200$ qubits. Figure~\ref{fig:topophase}(a) illustrates the sampled topological and trivial states. We generate representatives of the nontrival topological phase by applying low-depth geometrically local random quantum circuits to Kitaev's toric code state \cite{kitaev2003fault} with code distance 10, and we generate representatives of the trivial phase by applying random circuits to a product state.

Randomized Pauli measurements are performed $T=500$ times to convert the states to their classical shadows, and these classical shadows are mapped to feature vectors in the high-dimensional feature space using the feature map $\phi^{\mathrm{(shadow)}}$.
Figure~\ref{fig:topophase}(b) displays a one-dimensional projection of the feature space using the unsupervised classical ML for various values of the circuit depth, indicating that the phases become harder to distinguish as the circuit depth increases.
In Figure~\ref{fig:topophase}(c), we show the classification accuracy of the unsupervised classical ML model.
We also compare to training convolutional neural networks (CNN) that use measurement outcomes from the Pauli-$6$ POVM \cite{carrasquilla2019reconstructing} as input to learn an observable for classifying the phases.
Since Proposition~\ref{prop:noobs} establishes that no observable (even a global one) can classify topological phases, this CNN approach is doomed to fail.
On the other hand, if the CNN takes classical shadow representations as input, then it can learn nonlinear functions and successfully classify the phases.

\section{Outlook}

We have rigorously established that classical machine learning (ML) algorithms, informed by data collected in physical experiments, can effectively address {\color{black} some} quantum many-body problems. These results boost our hopes that classical ML trained on experimental data can solve practical problems in chemistry and materials science that would be too hard to solve using classical processing alone. 

Our arguments build on the concept of a classical shadow derived from randomized Pauli measurements. We expect, though, that other succinct classical representations of quantum states could be exploited by classical ML with similarly powerful results. 
For example, some currently available quantum simulators are highly programmable, but lack the local control needed to perform arbitrary single-qubit Pauli measurements. Instead, after preparing a many-body quantum state of interest, one might switch rapidly to a different Hamiltonian and then allow the state to evolve for a short time before performing a computational basis measurement. How can we make use of that measurement data to predict properties reliably? Answering such questions, and thereby expanding the reach of near-term programmable quantum platforms, will be an important goal for future research \cite{cotler2021emergent, hu2021hamiltonian}.

Viewed from a broader perspective, by illustrating how experimental data can be exploited to make accurate predictions about features of quantum systems that have never been studied directly, our work exemplifies a potentially powerful methodology for advancing the physical sciences. With further theoretical developments, perhaps we can learn how to use experimental data that is already routinely available to accelerate the discovery of new chemical compounds and materials with remarkable properties that could benefit humanity.

\begin{figure}[t]
    \centering
    \includegraphics[width=0.95\linewidth]{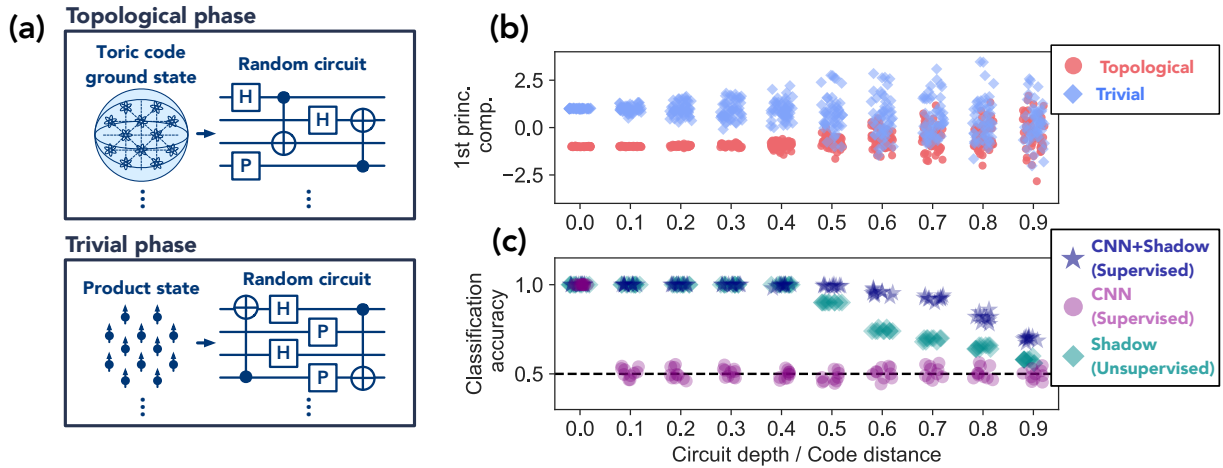}
    \caption{Numerical experiments for distinguishing between trivial and topological phases. (a) \textsc{State generation.} Trivial or topological states are generated by applying local random quantum circuits of some circuit depth to a product state or exactly-solved topological state, respectively. (b)
    \textsc{Unsupervised phase classification.} visualization of the quantum states projected to one dimension using the unsupervised ML ({\color{black} nonlinear} PCA with shadow kernel), shown for varying circuit depth (divided by the ``code distance'' 10, which quantifies the depth at which the topological properties are washed out). The feature space is sufficiently expressive to resolve the phases for a small enough depth without training, with classification becoming more difficult as the depth increases. (c) \textsc{Classification accuracy} for three ML algorithms described in Section~\ref{subsec:numerics-classification}.
    }
    \label{fig:topophase}
\end{figure}

\subsection*{Acknowledgments:}
\vspace{-0.5em}
{ The authors thank Nir Bar-Gill, Juan Carrasquilla, Sitan Chen, Yifan Chen, Andreas Elben, Matthew Fishman, Martin Fraas, Scott Glancy, Jeongwan Haah, Felix Kueng, Jarrod McClean, Spiros Michalakis, Jacob Taylor, Yuan Su, and Thomas Vidick for valuable input and inspiring discussions.
HH thanks Andreas Elben for providing the code on bond-alternating XXZ model.
The numerical simulations were performed on AWS EC2 computing infra-structure, using the software packages ITensors~\cite{itensor} and PastaQ~\cite{pastaq}.  
HH is supported by the J. Yang \& Family Foundation. 
JP acknowledges funding from  the U.S. Department of Energy Office of Science, Office of Advanced Scientific Computing Research, (DE-NA0003525, DE-SC0020290), and the National Science Foundation (PHY-1733907). The Institute for Quantum Information and Matter is an NSF Physics Frontiers Center. Contributions to this work by NIST, an agency of the US government, are not subject to US copyright. Any mention of commercial products does not indicate endorsement by NIST. VVA thanks Olga Albert, Halina and Ryhor Kandratsenia, as well as Tatyana and Thomas Albert for providing daycare support throughout this work.}

\newpage
\appendix

\renewcommand{\appendixname}{APPENDIX}
\renewcommand{\thesubsection}{\MakeUppercase{\alph{section}}.\arabic{subsection}}
\makeatletter
\renewcommand{\p@subsection}{}
\makeatother

\noindent 
\begin{center}\textbf{ \Large{}Appendices: contents \& navigation guide}{\Large\par}
\end{center}

\bigskip

\noindent \textbf{\ref{sec:classical-shadows}.~~\hyperref[sec:classical-shadows]{Background on classical shadows}} \dotfill\textbf{\pageref{sec:classical-shadows}}
\medskip

\noindent \textbf{\ref{sec:relatedwork}.~~\hyperref[sec:relatedwork]{Related work}} \dotfill\textbf{\pageref{sec:relatedwork}}
\medskip

\noindent \qquad \begin{minipage}{\dimexpr\textwidth-1.3cm}
 \hyperref[sec:relatedwork1]{Estimating ground state properties} $\bullet$ 
 \hyperref[sec:relatedwork2]{Classifying quantum phases of matter} $\bullet$ 
 \hyperref[sec:relatedwork3]{Classical representations of quantum systems}
\end{minipage}
\medskip

\noindent \textbf{\ref{sec:neuralnetwork}.~~\hyperref[sec:neuralnetwork]{Neural networks with classical shadow for quantum many-body problems}} \dotfill\textbf{\pageref{sec:neuralnetwork}}
\smallskip

\noindent \qquad \begin{minipage}{\dimexpr\textwidth-1.3cm}
 \hyperref[sec:neuralnetwork1]{Predicting ground state representation} $\bullet$ \hyperref[sec:neuralnetwork2]{Classifying phases of matter}
\end{minipage}
\medskip

\noindent \textbf{\ref{sec:numdetail}.~~\hyperref[sec:numdetail]{Details regarding numerical experiments}} \dotfill\textbf{\pageref{sec:numdetail}}
\medskip

\noindent \qquad \begin{minipage}{\dimexpr\textwidth-1.3cm}
 \hyperref[sec:more-numer]{Additional numerical experiments} $\bullet$ 
 \hyperref[sec:numdetail-groundstate]{Ground state properties of the Rydberg atom chain} $\bullet$ 
 \hyperref[sec:numdetail-groundstate2]{Ground state properties of the 2D antiferromagnetic Heisenberg model} $\bullet$ 
 \hyperref[sec:numdetail-phases]{Classifying phases of the bond-alternating XXZ model} $\bullet$ 
 \hyperref[sec:numdetail-phases2]{Distinguishing a topological phase from a trivial phase}
\end{minipage}
\medskip

\noindent \textbf{\ref{sec:proofideaGSUPP}.~~\hyperref[sec:proofideaGSUPP]{Proof idea for the efficiency in predicting ground states}} \dotfill\textbf{\pageref{sec:proofideaGSUPP}}
\medskip

\noindent \qquad \begin{minipage}{\dimexpr\textwidth-1.3cm}
 \hyperref[sec:qubit-GSUPP]{Main result} $\bullet$ 
 \hyperref[sec:generalize-GSUPP]{Generalization to other systems and settings}
\end{minipage}
\medskip

\noindent \textbf{\ref{sec:proofthmGSUPP}.~~\hyperref[sec:proofthmGSUPP]{Proof of efficiency for predicting ground states}} \dotfill\textbf{\pageref{sec:proofthmGSUPP}}
\medskip

\noindent \qquad \begin{minipage}{\dimexpr\textwidth-1.3cm}
 \hyperref[sec:overviewSCupGS]{Overview for sample complexity upper bound} $\bullet$ 
 \hyperref[sec:truncGS]{Controlling the truncation error} $\bullet$ 
 \hyperref[sec:MLTrun]{Controlling generalization errors from using the training data} $\bullet$ 
 \hyperref[sec:runtimeGS]{Computational time for training and prediction} $\bullet$ 
 \hyperref[sub:smoothness]{Spectral gap implies smooth parametrizations}
\end{minipage}
\medskip

\noindent \textbf{\ref{app:proofthmlowerbound}.~~\hyperref[app:proofthmlowerbound]{Sample complexity lower bound for predicting ground states}} \dotfill\textbf{\pageref{app:proofthmlowerbound}}
\smallskip

\noindent \qquad \begin{minipage}{\dimexpr\textwidth-1.3cm}
 \hyperref[sub:lowerboundsmooth-learning-problem]{Learning problem formulation} $\bullet$ \hyperref[sub:lowerboundsmooth-communication-protocol]{Communication protocol} $\bullet$ \hyperref[sub:lowerboundsmooth-analysis]{Information-theoretic analysis}
\end{minipage}
\medskip

\noindent \textbf{\ref{sec:proofhardnonML}.~~\hyperref[sec:proofhardnonML]{Computational hardness for non-ML algorithms to predict ground state properties}} \dotfill\textbf{\pageref{sec:proofhardnonML}}
\medskip

\noindent \qquad \begin{minipage}{\dimexpr\textwidth-1.3cm}
 \hyperref[sec:proofhardnonML1]{NP-hardness for estimating one-body observables in the ground state of 2D Hamiltonians} $\bullet$
 \hyperref[sec:factoringH]{Computational hardness for a class of Hamiltonians based on factoring}
\end{minipage}
\medskip

\noindent \textbf{\ref{sec:no-nonlocal-obs}.~~~\hyperref[sec:no-nonlocal-obs]{No observable can classify topological phases}} \dotfill\textbf{\pageref{sec:no-nonlocal-obs}}
\medskip

\noindent \textbf{\ref{sec:proofthmPHASEC}.~~~\hyperref[sec:proofthmPHASEC]{Proof of efficiency for classifying phases of matter}} \dotfill\textbf{\pageref{sec:proofthmPHASEC}}
\medskip

\noindent \qquad \begin{minipage}{\dimexpr\textwidth-1.3cm}
 \hyperref[sec:trainSVM]{Training support vector machines} $\bullet$ 
 \hyperref[sec:pred-SVM]{Prediction using support vector machines} $\bullet$ 
 \hyperref[sub:shadow-kernel]{Kernel functions for classical shadows} $\bullet$ 
 \hyperref[sec:phases-assumption]{Physical assumptions about classifying quantum phases of matter} $\bullet$ 
 \hyperref[sec:train-shadowkernel]{Training with shadow kernels} $\bullet$ 
 \hyperref[sec:predict-shadowkernel]{Prediction based on shadow kernels}
\end{minipage}
\bigskip

\noindent \textbf{\ref{sec:tasaki}.~~\hyperref[sec:tasaki]{Classifying SPT phases with $O(2)$ symmetry using a few-body observable}} \dotfill\textbf{\pageref{sec:tasaki}}
\smallskip

\noindent \qquad \begin{minipage}{\dimexpr\textwidth-1.3cm}
 \hyperref[sec:tasaki1]{Symmetry-protected topological phases} $\bullet$ 
 \hyperref[sec:tasaki2]{$O(2)$-symmetric qutrit spin chains}
\end{minipage}

\bigskip\bigskip
For those unfamiliar with the classical shadow formalism \cite{huang2020predicting}, Appendix \ref{sec:classical-shadows} provides a concise introduction and contains all the necessary information on classical shadows to follow this paper. Discussion of related literature on methods in many-body physics, in particular works for training machine learning models to solve quantum many-body problems, is given in Appendix~\ref{sec:relatedwork}.

Readers interested in the numerical experiments can jump directly to Appendix~\ref{sec:neuralnetwork}, which provides a detailed discussion on practical approaches for combining neural networks with classical shadows.
The reader could then continue to Appendix \ref{sec:numdetail}, which gives additional numerical experiments, details of numerical experiments, and example codes (in Python) for using the machine learning models.

The rest of the appendices are dedicated to providing a rigorous understanding in using classical machine learning models to solve quantum many-body problems.

\paragraph{Predicting ground states}: We recommend that readers start with Appendix~\ref{sec:proofideaGSUPP}, which provides the idea for why classical machine learning models can be trained to predict ground state representations of quantum systems. In Appendix~\ref{sec:proofideaGSUPP}, we also provide the proof idea for generalizing to other settings (such as Fermionic systems, long-range interacting systems, etc).
The detailed proof is given in Appendix~\ref{sec:proofthmGSUPP}.
In Appendix~\ref{app:proofthmlowerbound}, we give a fundamental lower bound in the required data size for learning to predict ground state properties for general classes of Hamiltonians.
In Appendix~\ref{sec:proofhardnonML}, we show why non-ML algorithms cannot achieve a similar guarantee as ML algorithms in predicting ground state representations.

\paragraph{Classifying phases of matter}: The reader could begin with the basic proposition given in Appendix~\ref{sec:no-nonlocal-obs}, which shows that no (local or global) observable $\Tr(O \rho)$ can be used to classify topological phases. This motivates the need to consider stronger machine learning models that can learn nonlinear functions in the quantum state $\rho$.
The readers could then proceed to Appendix~\ref{sec:proofthmPHASEC}, which provides a general theory for establishing provable guarantees in training machine learning models based on classical shadows to classify quantum phases of matter.
In Appendix~\ref{sec:tasaki}, we briefly introduce symmetry-protected topological phases and prove that the proposed machine learning model can classify a particular subset of such phases.

\newpage

\section{Background on classical shadows} \label{sec:classical-shadows}

The classical shadows formalism uses randomized (single-shot) measurements to predict many properties of an unknown quantum state $\rho$ at once \cite{huang2020predicting}, see also \cite{paini2019approximate, evans2019scalable}. The underlying idea dates back to \cite{Ohliger_2013} and also features prominently in \cite{VanEnk2012,Elben2018,Vermersch2018}. 
In particular, the classical shadows formalism comes with rigorous performance guarantees in terms of approximation accuracy, classical storage, as well as data processing.
Here, we focus on randomized single-qubit Pauli measurements and repeat the following procedure a total of $T$ times:
(i) prepare an independent copy of $\rho$; (ii) select $n$ single qubit Pauli measurements uniformly at random ($Z$, $X$ and $Y$ occur with probability $1/3$ each) and (iii), perform the associated measurement to obtain 
$n$ classical bits ($+1$ if we measure `up' and $-1$ if we measure `down'). Subsequently, we store the associated post-measurement state
\begin{equation}
| s_1^{(t)}\rangle \otimes \cdots \otimes | s_n^{(t)} \rangle \quad \text{with} \quad |s_1^{(t)}\rangle,\ldots,|s_n^{(t)} \rangle \in \left\{ |0 \rangle, |1 \rangle, |+ \rangle, |- \rangle, |\mathrm{i}+\rangle,|\mathrm{i}-\rangle \right\} \subset \mathbb{C}^2
\label{eq:snapshot}
\end{equation}
in classical memory. This is very cheap, because there are only six possibilities for each qubit. After $T$ repetitions, we obtain an entire collection of $nT$ single-qubit states
that we arrange in a two-dimensional array:
\begin{equation}
S_T (\rho) =\left\{ |s_i^{(t)} \rangle :\; i \in \{1,\ldots,n\}, t \in \{1,\ldots,T\}\right\} 
\in \left\{ |0 \rangle, |1 \rangle, |+ \rangle, |- \rangle, |\mathrm{i}+\rangle,|\mathrm{i}-\rangle \right\}^{n \times T}
\end{equation}
The distribution of product states contains valuable information about the underlying $n$-qubit density matrix $\rho$. 
In fact, we can use $S_T (\rho)$ to approximate $\rho$
via
\begin{equation}
\sigma_T (\rho) = \frac{1}{T}\sum_{t=1}^T \left( 3 |s_1^{(t)}\rangle \! \langle s_1^{(t)}| - \mathbb{I} \right) \otimes \cdots \otimes \left( 3 |s_n^{(t)} \rangle \! \langle s_n^{(t)}|-\mathbb{I} \right), \label{eq:classical-shadow}
\end{equation}
where $\mathbb{I}$ denotes the identity matrix (here, a 2-by-2 identity). 
It is instructive to view this as the empirical average of $T$ independent and identically (\textit{iid}) random matrices.
Each random matrix is an \textit{iid} copy of $\sigma_1 (\rho) = \big( 3 |s_1 \rangle \! \langle s_1|-\mathbb{I}\big) \otimes \cdots \otimes \big( 3|s_n \rangle \! \langle s_n|-\mathbb{I} \big)$. Each tensor factor is guaranteed to have eigenvalues $\lambda_+ = 2$ and $\lambda_- = -1$. This  ensures that \begin{subequations}
\label{eq:shadow-norm}
\begin{align}
\mathrm{tr} \left( \sigma_1 (\rho) \right) =& \mathrm{tr} \left( |s_1 \rangle \! \langle s_1| - \mathbb{I} \right) \cdots \mathrm{tr} \left( |s_n\rangle \! \langle s_n|-\mathbb{I} \right) =1 \quad \text{and} \\
\left\| \sigma_1 (\rho) \right\|_p =& \left\| 3 |s_1 \rangle \! \langle s_1| - \mathbb{I} \right\|_p \cdots \left\| 3 |s_n \rangle \! \langle s_n| - \mathbb{I} \right\|_p
= \left( |\lambda_+|^p + |\lambda_-|^p \right)^{n/p}
= \left( 2^p + 1^p \right)^{n/p},
\end{align}
\end{subequations}
regardless of the concrete realization (and the underlying quantum state $\rho$).
The most relevant Schatten-$p$ norms are $\|\sigma_1 (\rho)\|_1=3^n$, $\|\sigma_1 (\rho) \|_2 = 5^{n/2}$ and $\| \sigma_1 (\rho) \|_\infty = 2^n$.
Note, however, that the matrix $\sigma_1 (\rho)$ is never positive semidefinite. 

The random matrix $\sigma_1 (\rho)$ is a highly structured tensor product that can assume a total of $6^n$ values.
Each of them reflects the outcome of performing randomly selected single-qubit Pauli measurements on the $n$-qubit state $\rho$.
Let us denote these Pauli matrices by $W_1,\ldots,W_n \in \left\{X,Y,Z\right\}$ and let $o_1,\ldots,o_n \in \left\{ \pm1 \right\}$ be the observed outcomes ($+1$ if we measure `spin up' and $-1$ if we measure `spin down').
Elementary reformulations and Born's rule then imply
\begin{align}
\sigma_1 (\rho) = \frac{1}{2}\left( \mathbb{I}+ 3 o_1 W_1 \right) \otimes \cdots \otimes \frac{1}{2}\left( \mathbb{I} + 3 o_n W_n \right)
\quad \text{with prob.} \quad \frac{1}{3^n} \mathrm{tr} \left( \frac{1}{2}(\mathbb{I} + o_1 W_1) \otimes \cdots \otimes \frac{1}{2} (\mathbb{I} + o_n W_n) \rho \right).
\end{align}
This construction ensures that $\sigma_1 (\rho)$ exactly reproduces the underlying quantum state $\rho$ in expectation. That is, if we average over all $3^n$ choices of Pauli measurements and the associated (single-shot) outcomes $o_i \in \left\{\pm 1 \right\}$, we obtain
\begin{subequations}
\begin{align}
\E_{s_1,\ldots,s_n} \left[ \sigma_1 (\rho) \right]
&= \E_{s_1,\ldots,s_n}
\left[ \frac{1}{2}\left( \mathbb{I} +3 o_1  W_1 \right) \otimes \cdots \otimes \frac{1}{2}\left( \mathbb{I} +3 o_n 3 W_n \right)\right] \\
&=
\sum_{W_1,\ldots,W_n =X,Y,Z}\sum_{o_1,\ldots,o_n=\pm 1} \frac{1}{3^n} \mathrm{tr} \left( \frac{1}{2}(\mathbb{I} +o_1 W_1) \otimes \cdots \otimes \frac{1}{2} (\mathbb{I} +o_n W_n) \rho \right) \\
&~~~~~~~~~~~~~~~~~~~~~~~~~~~~~~~~~~~~~~~~~~~~~~~\times \frac{1}{2}\left( \mathbb{I} +o_1 3 W_1 \right) \otimes \cdots \otimes \frac{1}{2}\left( \mathbb{I} +o_n 3 W_n \right)  \\
&= \rho. 
\end{align}
\label{eq:snapshot-expectation}
\end{subequations}
We refer to Ref.~\cite{huang2020predicting} for a more detailed derivation and context. 

The classical shadow~\eqref{eq:classical-shadow} attempts to approximate this expectation value by an empirical average over $T$ independent samples, much like Monte Carlo sampling approximates an integral. 
The accuracy of the approximation increases with $T$,
but insisting on accurate approximations of the global state $\rho$ is prohibitively expensive. Known fundamental lower bounds~\cite{flammia2012quantum,haah2017sample} state that classical shadows of exponential size (at least) $T= \Omega \left( 2^n/\epsilon^2 \right)$ are required to $\epsilon$-approximate $\rho$ in trace distance.
This quickly becomes intractable in terms of both measurement budget, as well as classical storage and processing.

This bleak picture lightens up considerably if we restrict our attention to subsystem approximations. 
The classical shadow size required to accurately approximate \emph{all} reduced $r$-body density matrices scales exponentially in subsystem size $r$, but is independent of the total number of qubits $n$. 

\begin{lemma} \label{lem:subsystem-approximation}
Fix $\epsilon,\delta \in (0,1)$, a subsystem size $r \leq n$ and
let $\sigma_T (\rho)$ be a classical shadow~(\ref{eq:classical-shadow}) of an $n$-qubit quantum state $\rho$ with size 
\begin{equation}
T=(8/3) 12^r \left( r \left( \log (n) + \log (12) \right) + \log (1/\delta) \right)/\epsilon^2 = \mathcal{O}\left( r 12^r \log (n/\delta)/\epsilon^2 \right).
\end{equation}
Then, with probability at least $1-\delta$,
\begin{equation}
\left\| \mathrm{tr}_{\neg A} \left( \sigma_T (\rho) \right) - \mathrm{tr}_{\neg A}\left( \rho \right) \right\|_1 \leq \epsilon \quad \text{for all subsystems $A \subset \{1,\ldots,n\}$ with size $|A|\leq r$.}
\end{equation}
\end{lemma}

\begin{proof}
Let us start by considering a fixed subsystem $A=\{i_1,\ldots,i_r\}$ comprised of (at most) $r$ qubits. 
Use linearity to exchange partial trace with expectation value to obtain
\begin{subequations}
\begin{align}
& \E_{s_{i_1}^{(t)},\ldots,s_{i_r}^{(t)}}
\left[ \big( 3|s_{i_1}^{(t)}\rangle \! \langle s_{i_1}^{(t)}| - \mathbb{I} \big) \otimes \cdots \otimes \big( 3 |s_{i_r}^{(t)} \rangle \! \langle s_{i_r}^{(t)}| -\mathbb{I} \big) \right] \\
&= \mathrm{tr}_{\neg A}
\left(
\E_{s_{1}^{(t)},\ldots,s_{n}^{(t)}}
\left[ \big( 3|s_{1}^{(t)}\rangle \! \langle s_{1}^{(t)}| - \mathbb{I} \big) \otimes \cdots \otimes \big( 3 |s_{n}^{(t)} \rangle \! \langle s_{n}^{(t)}| -\mathbb{I} \big) \right]
\right) \\
&= \mathrm{tr}_{\neg A} (\rho),
\end{align}
\end{subequations}
according to Eq.~\eqref{eq:snapshot-expectation}.
In words, each reduced tensor product 
is an independent random matrix that reproduces the $r$-qubit state $\mathrm{tr}_{\neg A}(\rho)$ exactly in expectation. 
Empirical averages of $T$ such independent and identically distributed (\textit{iid}) random matrices tend to concentrate sharply around this expectation value.
The matrix Bernstein inequality,  see e.g.\ \cite{Tro12:User-Friendly}, provides powerful tail bounds in terms of operator norm deviation. Let $X_1,\ldots,X_T$ be \textit{iid} random $D$-dimensional matrices that obey $\|X_t -\E X_t\|_\infty \leq R$ almost surely. Then, for $\tilde{\epsilon} >0$
\begin{equation}
\mathrm{Pr} \left[ \left\| \tfrac{1}{T} \sum_{t=1}^T \left( X_t - \E X_t \right) \right\|_\infty \geq \tilde{\epsilon} \right] \leq 2 D \exp \left( - \frac{T\tilde{\epsilon}^2/2 }{\sigma^2+R\tilde{\epsilon}/3}\right)
\quad \text{where} \quad \sigma^2 = \left\|\tfrac{1}{T} \sum_t \E X_t^2 \right\|_\infty.
\end{equation}
Let us apply this tail bound to classical shadow concentration.
We have $D \leq 2^r$ (at most $r$ qubits) and set $X_t = \big(3 |s_{i_1}^{(t)} \rangle \! \langle s_{i_1}^{(t)}-\mathbb{I}\big) \otimes \cdots \otimes \big( 3 |s_{i_r}^{(t)} \rangle \! \langle s_{i_r}^{(t)} | - \mathbb{I} \big)$, such that $\E X_t = \mathrm{tr}_{\neg A} (\rho)$.
Eq.~\eqref{eq:shadow-norm} then implies $\|X_t - \E X_t \|_\infty \leq \| X_t \| + \| \E X_t \|_\infty  \leq 2^r+1 =:R$.
Accurately bounding $\sigma^2$ is somewhat more involved, and we turn to existing literature.
A computation detailed in  \cite[Appendix~C.3]{guta2020fast} yields $\sigma^2=3^r$.
We are now ready to apply the matrix Bernstein inequality. For $\tilde{\epsilon} >0$,
\begin{align}
\mathrm{Pr} \left[ \left\| \mathrm{tr}_{\neg A} \left( \sigma_T (\rho) \right) - \mathrm{tr}_{\neg A}(\rho) \right\|_\infty \geq \tilde{\epsilon} \right]
\leq 2^{r+1} \exp \left( - \frac{T \tilde{\epsilon}^2/2}{3^r+(2^r+1) \tilde{\epsilon}/3} \right)
\leq 2^{r+1} \exp \left( - \frac{3 T \tilde{\epsilon}^2}{8 \times 3^r} \right),
\end{align}
for $\tilde{\epsilon} \in (0,1)$. This is a powerful concentration statement in operator norm.
We can use the  equivalence relation between trace- and operator norm, $\|X \|_\infty \leq \|X \|_1 \leq D \|X \|_\infty$, to obtain a tail bound for trace norm deviations:
\begin{equation}
\mathrm{Pr} \left[ \left\| \mathrm{tr}_{\neg A} \left( \sigma_T (\rho) \right) - \mathrm{tr}_{\neg A}(\rho) \right\|_\infty \geq \epsilon \right]
\leq \mathrm{Pr} \left[ \left\| \mathrm{tr}_{\neg A} \left( \sigma_T (\rho) \right) - \mathrm{tr}_{\neg A}(\rho) \right\|_1 \geq \epsilon /2^r \right]
\leq 2^{r+1} \exp \left( - \frac{3 T \epsilon^2}{8 \times 12^r} \right).
\end{equation}
We see that, for a fixed subsystem $A=\{i_1,\ldots,i_r\}$, the probability of an $\epsilon$-deviation in trace distance is exponentially suppressed in the size $T$ of the classical shadow.
A union bound allows us to extend this assertion to \emph{all} subsystems comprised of (at most) $r$ qubits:
\begin{subequations}
\begin{align}
\mathrm{Pr} \left[\max_{A \subset \{1,\ldots,n\},|A| \leq r} \left\| \mathrm{tr}_{\neg A}\left( \sigma_T (\rho) \right) - \mathrm{tr}_{\neg A} (\rho) \right\|_1 \geq \epsilon\right]
&\leq  \sum_{A \subset \{1,\ldots,n\},|A| \leq r}
\mathrm{Pr} \left[ \left\| \mathrm{tr}_{\neg A}\left( \sigma_T (\rho) \right) - \mathrm{tr}_{\neg A} (\rho) \right\|_1 \geq \epsilon \right] \\
&\leq  n^r 2^{r+1} \exp \left( - \frac{3 T \epsilon^2}{8 \times 12^r}\right).
\end{align}
\end{subequations}
Setting $T=(8/3)12^r \left(\log (n^r 12^r)+\log (1/\delta) \right)/\epsilon^2=
(8/3) 12^r \left( r \left( \log (n) + \log (12) \right) + \log (1/\delta) \right)/\epsilon^2$
ensures that this upper bound on failure probability does not exceed $\delta$.
\end{proof}

\section{Related work} \label{sec:relatedwork}

\subsection{Estimating ground state properties} \label{sec:relatedwork1}

Determining the ground state of a system governed by a known many-body Hamiltonian is a long-standing problem in quantum science.
Despite having several well-established and practically successful algorithms at our disposal, we are typically faced with either a runtime that scales exponentially with system size, and/or a lack of rigorous performance guarantees.
The literature is vast, and surveying it is beyond the scope of this article. Instead, we review a few families of established algorithms in order to put our work into proper context.

Density functional theory (DFT) has been a workhorse for determining properties of interacting electronic systems in quantum chemistry and solid-state physics. DFT recasts the problem of finding the many-body state with minimal energy into finding a few-body energy functional. While the ``true'' functional corresponding to the ground state is known to exist in theory~\cite{HohenbergKohn, NobelKohn}, determining it to polynomial accuracy in the number of electrons is QMA-hard \cite{Schuch2009}. Various efficient approximations to the true functional have seen much practical success, but they are difficult to justify rigorously (except so far for some special cases \cite{Lewin2020}). These limitations present an opportunity for ML approaches to be used instead of or to supplement DFT methods \cite{Schleder2019, moreno2020deep}.
In \cite{moreno2020deep}, they utilized the Hohenberg-Kohn theorem in DFT which implies the existence of a mapping from single-particle densities to correlation function. Then \cite{moreno2020deep} proposed to train deep learning models to approximately learn the mapping between correlation functions and the densities.
The authors in \cite{moreno2020deep} found that across the phase boundary, the mapping becomes hard to learn. This empirical observation could be seen as a manifestation of the fact that Theorem~\ref{thm:mainFourier} only holds when we do not consider the presence of phase transitions.
However, our numerical experiments on Rydberg atoms suggest that ML models could still predict well when there are multiple phase boundaries.
Understanding when predictions would be accurate in the presence of phase transitions is an interesting and important question to study further.

The family of algorithms known as Quantum Monte Carlo (QMC)~\cite{CEPERLEY555,SandvikSSE} utilizes probabilistic sampling techniques to estimate observable properties at either finite or zero temperature. For ground states, expectation values can be obtained using an imaginary-time evolution projector or a high-power of the model Hamiltonian~\cite{KaulQMC}. The efficiency of QMC methods depends on the structure of the Hamiltonian, specifically on whether all of its off-diagonal matrix elements are negative (i.e.\ it is {\it stoquastic}~\cite{BravyiStoquastic}). In this case, the ground state wavefunction is real-valued and positive, and the algorithmic complexity of the QMC estimators scales polynomially with the number of particles. For non-stoquastic Hamiltonians, the QMC suffers from the so-called {\it sign problem}, which makes evaluation of statistical properties of the system NP-hard and renders QMC intractable for large systems or low temperatures~\cite{TroyerSignProblem}. It is important to note that, even for stoquastic Hamiltonians with polynomial computational complexity, the success of QMC simulations heavily relies on the existence of efficient sampling schemes (e.g. cluster updates) which are sufficiently ergodic, and leading to small auto-correlation time~\cite{EvertzLoop}. In general, it is not possible to prove the existence of such update algorithms, nor their ergodicity.  

An alternative approach to solve for the ground state properties of a many-body Hamiltonian is based on the variational principle in quantum mechanics, which states that the expectation value of the energy on any valid wavefunction is always greater or equal than the ground state energy. It is then possible to design classical parametric representations of the many-body wavefunction, and update their parameters to minimize the corresponding energy estimator. A notable example is the density-matrix renormalization group~\cite{DMRG1,DMRG2} (DMRG). This algorithm can be interpreted as a variational optimization of a Matrix Product State (MPS)~\cite{Garcia07,Vestraete2008,SCHOLLWOCK201196}, which is a local decomposition of a wavefunction as a one-dimensional tensor network. These parametrized wavefunctions display area law of entanglement and exponentially-decaying correlations~\cite{ORUS2014117}, which lend themselves most effective for systems described by one-dimensional gapped Hamiltonians. Furthermore, a standout feature of DMRG is that modifications of the original procedure, such as rigorous renormalization group algorithms \cite{arad2017rigorous}, are guaranteed to find the ground state of one-dimensional geometrically-local gapped Hamiltonians in polynomial time \cite{landau2015polynomial, arad2017rigorous, abrahamsen2020polynomialtime}. In two spatial dimensions, MPS-based DMRG can still be applied to solve for ground states~\cite{Stoudenmire12,Wu2019,Moore2020}, though it suffers an exponential scaling in one of the two linear dimensions of the system. Projected entangled pair states (PEPS)~\cite{Jordan2008, Vanderstraeten2016, Corboz2016}, the two-dimensional generalization of MPSs, are instead a more suitable ansatz in this context. However, while improved algorithms for PEPS optimization are routinely put forward~\cite{Zalatel2020, Haghshenas2019, Hyatt2019}, the same level of performance achieved by DMRG in 1d systems is still out of reach.

Another class of variational wavefunctions that has recently received a lot of attention are {\it neural-network quantum states}~\cite{Carleo_2017}. In this framework, a neural network is used as a parametric function approximator of a many-body wavefunction $\psi_{\lambda}(\sigma)=\langle\sigma|\psi\rangle$, where the classical state $\sigma$ is interpreted as the neural-network input, and $\lambda$ is a set of neural-network parameters (i.e. weights and biases). In a variational setting, these parameters are iteratively optimized to lower the total energy~\cite{becca_sorella_2017}, or additionally the energy variance. Neural-network quantum states have been explored in a variety of setups, including topological phases~\cite{dassarma2017}, Fermi-Hubbard models~\cite{Nomura2017}, molecular ground states~\cite{choo_fermionicnqs2020}, frustrated magnetism~\cite{Ferrari2019}, and more~\cite{Glasser2018, Choo2018, rbm_su2, morawetz2020, Luo2021}. The auto-regressive property of some types of neural networks (e.g. recurrent neural networks, transformers, etc.) has also been leveraged to improve convergence of variational Monte Carlo~\cite{hibat2020recurrent}. In contrast to tensor-network states, this class of wavefunctions can more easily display non-local correlations, allowing in principle to capture quantum states with higher entanglement~\cite{Deng2017}.

Another class of machine learning methods \cite{gilmer2017neural,vargas2018extrapolating,schutt2019unifying,qiao2020orbnet,kawai2020predicting} train neural networks to predict the ground state or excited state properties directly. The input to the neural network is a description of the Hamiltonian, and the output is a ground state property of interest.
The training data is a set of different Hamiltonians (inputs) and their corresponding ground state properties (outputs).
This class of ML methods is closest to the setting considered in this paper. Methods in this class lack rigorous guarantees, so it is not clear when such approaches could outperform non-ML algorithms.
In one of our main contributions, given in Theorem~\ref{thm:mainFourier} and Proposition~\ref{thm:hardnonML}, we introduce an ML model that, when trained with experimental data, can accurately predict ground state representations better than any classical algorithm that does not learn from data.
Our model is relatively basic, utilizing the well-known $l_2$-Dirichlet kernel, but it is already enough to establish a rigorous guarantee.
Similarly determining when other ML models yield an advantage over non-ML algorithms is an interesting topic for future work.

\subsection{Classifying quantum phases of matter} \label{sec:relatedwork2}

Proposals for classifying quantum phases of matter abound. These include quantum neural networks \cite{cong2019quantum}, classical neural networks \cite{carrasquilla2017nature,evert2017nature,neupeurt2017,Beach2017,Greplova_2020}, or other classical machine learning models \cite{leiwang2016,Wetzel2017,rodriguez2019identifying, che2020topological}.  
Since these models do not come with rigorous guarantees, relying on them too much can lead one astray. For example, some deep neural networks can misclassify the original phase if the corresponding state is distorted by noise, even if the distortion is very slight \cite{jiang2019vulnerability}.

In this work, we provide rigorous machine learning approaches that are guaranteed to classify accurately under the specified conditions given in Theorem~\ref{thm:phaseclassification}. We believe similar analyses can be performed on other machine learning models to understand their limitation and potential, which will be an important future direction.
The ML model used in \cite{rodriguez2019identifying}, which is based on defining diffusion maps over classical spins systems, is the ML approach that seems most similar to that used in our work. Hence, it is very likely that the models considered in \cite{rodriguez2019identifying} can be rigorously analyzed via similar techniques.
Neural network approaches for classifying phases of matter will be harder to analyze, but one should be able to study neural network with large hidden layers using the theory developed in this work and the theory of neural tangent kernels \cite{jacot2018neural, neuraltangents2020}.

\subsection{Classical representations of quantum systems} \label{sec:relatedwork3}

One of the most important ingredients in designing classical ML procedures for understanding quantum spin systems is the construction of efficient classical representations of the underlying quantum systems. The properties of the quantum system retained by the classical representation directly determine the set of functions the classical ML procedure can learn.
The classical shadow formalism \cite{huang2020predicting,paini2019approximate}, developed by some of us and others and used throughout this work, is a versatile framework for this purpose. It has been extended to fermionic systems \cite{zhao2021fermionic, hadfield2020measurements}, suggesting that our ML approaches may be extendable as well. Classical shadows have also been shown to allow sample-efficient reconstruction of Hamiltonian from thermal states \cite{anurag2021sample}, although such an algorithm is not yet time efficient. However, classical shadows provide only one of many promising and actively studied approaches for efficient representation \cite{cotler2020quantum, paini2019approximate, huang2020predicting, evans2019scalable, koh2020classical, chen2020robust, huang2021efficient}. 

While the curse of dimensionality prevents one from representing general quantum spin systems both exactly and efficiently, simplifying assumptions can lift the curse and drastically reduce both the overhead and complexity of representing and characterizing the system. A prime example is a classical system, whose Hamiltonian is diagonal in the computational basis. ML methods for such systems, such as those in Ref.~\cite{rodriguez2019identifying} (discussed above), do not require an additional quantum-to-classical compression. Another example, relevant to electronic material characterization, is the family of solid-state band insulators \cite{topobook} --- gapped two- and three-dimensional non-interacing fermionic systems with various crystalline symmetries. Their myriad topological phases can typically be characterized by data at a discrete set of high-symmetry points in the Brillouin zone \cite{Bacry,Bradlyn2017,Po2017,Kruthoff2017,Po2020,Cano2021}. The techniques developed here should pave the way for certifying accuracy of current ML methods for band insulator characterization (e.g., \cite{Sun2018,Ming2019,Balabanov2021,Scheurer2020,Claussen2020,Andrejevic2020,Peano2021}) as well as developing new ones.

\section{Neural networks with classical shadow for quantum many-body problems} \label{sec:neuralnetwork}

Imposing inductive biases in the ML model is a common technique for boosting the prediction performance of ML models.
One approach is to enhance the proposed ML algorithms with neural networks, such as convolutional or graph neural networks. These neural networks could better capture structure of the underlying function we are trying to learn and hence may require significantly less data than the very expressive ML model given in the main text. We leave the proof that neural network enhancements can lead to better prediction performance as a goal for future work.

There are multiple ways of combining classical shadows and neural networks. Here, we will only showcase one such approach by utilizing the theory of neural tangent kernels \cite{jacot2018neural}.
Remarkably, this theory allows us to efficiently train various types of neural networks (convolutional/graph/etc.) with an infinite number of neurons in each hidden layer (\textit{infinite width}).
As such, this line of work has gained a lot of attention \cite{du2019graph,arora2019exact, neuraltangents2020} in recent years.
In the limit of infinite width, one can analytically solve for the neural network after training on a set of data $\{x_{\ell}, y_{\ell}\}_{\ell = 1}^N$, where $x_\ell$ and $y_\ell$ are vectors of some size. For example, consider training a neural network that takes in a vector $x$ and produces a vector $f^{\mathrm{NN}}_\theta(x)$ through the following optimization problem using gradient descent,
\begin{equation}
    \min_{\theta} \sum_{\ell = 1}^N \norm{f^{\mathrm{NN}}_\theta(x_{\ell}) - {y}_{\ell}}_2^2,
\end{equation}
where we begin on a randomly initialized $\theta$.
Note that due to the infinite number of neurons, $\theta$ is a vector of infinite dimension.
The trained neural network $f^{\mathrm{NN}}_{\theta^*}({x})$ can always be written in the following form
\begin{equation}\label{eq:ntkfunc}
    f^{\mathrm{NN}}_{\theta^*}(x) = \sum_{\ell = 1}^N \sum_{\ell'=1}^N k^{\mathrm{(NTK)}}(x, x_{\ell}) (K^{-1})_{\ell \ell'} y_{\ell'},
\end{equation}
where $k^{\mathrm{(NTK)}}(x, x')$ is a function called the neural tangent kernel \cite{jacot2018neural}, and $K_{\ell, \ell'} = k^{\mathrm{(NTK)}}({x}_{\ell}, {x}_{\ell'})$ is the kernel matrix of the neural tangent kernel.
One can see that the infinite-dimensional vector $\theta^*$ does not appear on the right hand side of Eq.~\eqref{eq:ntkfunc}.
And as long as we can efficiently evaluate the neural tangent kernel $k^{\mathrm{(NTK)}}(x, x')$, we can evaluate the infinite-dimensional neural network in polynomial time.
This is the main contribution of \cite{jacot2018neural}, which enables one to efficiently train infinite-width neural networks.
For a given neural network architecture, one can compute $k^{\mathrm{(NTK)}}({x}, {x}')$ efficiently using open-source software, such as \cite{neuraltangents2020}.
In Appendix~\ref{sec:numdetail-groundstate}, we give the code for training infinite-width neural networks using the open-source software: Neural Tangents \cite{neuraltangents2020}.

\subsection{Predicting ground state representation} \label{sec:neuralnetwork1}

For the task of predicting ground state representation, we consider the training data to be
\begin{equation}
\big\{{x}_{\ell} \rightarrow \sigma_T(\rho({x}_{\ell}))\big\}_{\ell = 1}^N~,
\end{equation}
where $\sigma_T(\rho({x}_{\ell}))$ is the classical shadow representation of $\rho({x}_{\ell})$ given in Eqs.~\eqref{eq:sigma-T-shadow} and \eqref{eq:classical-shadow} based on $T$ randomized Pauli measurements.
Recall that $\sigma_T(\rho({x}_{\ell}))$ is a $2^n \times 2^n$ matrix that reproduces $\rho({x}_{\ell})$ in expectation over the randomized Pauli measurements.
Suppose we now train an infinite-width neural network parameterized by $\theta$ that takes in an input ${x}$ and produces an exponential-size matrix $\sigma^{\mathrm{NN}}_\theta({x})$, by solving the optimization problem
\begin{align}
    \min_{\theta} & \quad \sum_{\ell = 1}^N \norm{\sigma^{\mathrm{NN}}_\theta({x}_{\ell}) -  \sigma_T(\rho({x}_{\ell}))}_F^2.
\end{align}
The squared Frobenius difference between two matrices is equal to the squared Euclidean norm of their vectorizations (flattenings).
In turn, the theory of infinite-width neural networks \cite{jacot2018neural} shows that the trained neural network $\sigma^{\mathrm{NN}}_{\theta^*}({x})$ could be written in the form
\begin{equation} \label{eq:nngroundstate}
    \sigma^{\mathrm{NN}}_{\theta^*}({x}) = \sum_{\ell = 1}^N \sum_{\ell'=1}^N k^{\mathrm{(NTK)}}({x}, {x}_{\ell}) (K^{-1})_{\ell \ell'} \sigma_T(\rho({x}_{\ell'})).
\end{equation}
The kernel function $k^{\mathrm{(NTK)}}({x}, {x}')$ depends on the neural network architecture and could be calculated utilizing existing open-source software \cite{neuraltangents2020}.
This also falls into the general form shown in the main text; see Eq.~\eqref{eq:sigma(x)}.
Hence, training an infinite-width neural networks to predict an exponentially large density matrix can be done efficiently on a classical computer.
For a given neural network architecture, all one has to do is compute the kernel function $k^{\mathrm{(NTK)}}({x}, {x}')$. Then the neural network optimized using the training data could be analytically solved as given in Eq.~\eqref{eq:nngroundstate}.
To estimate a property on the predicted ground state using the neural network is as simple as evaluating
\begin{equation}
    \Tr(O \sigma^{\mathrm{NN}}_{\theta^*}({x})) = \sum_{\ell = 1}^N \sum_{\ell'=1}^N k^{\mathrm{(NTK)}}({x}, {x}_{\ell}) (K^{-1})_{\ell \ell'} \Tr(O \sigma_T(\rho({x}_{\ell'}))),
\end{equation}
which can be done by first computing $\Tr(O \sigma_T(\rho({x}_{\ell}))), \forall \ell = 1, \ldots, N$ and compute the linear interpolation.

\subsection{Classifying phases of matter} \label{sec:neuralnetwork2}

We want to learn how to classify two phases of $n$-qubit states.
A fully classical training set would simply consist of $N$ labeled classical representations of quantum states $\{\rho_\ell \rightarrow y_\ell \}_{\ell=1}^N$, where $y_\ell=+1$ ($-1$) if $\rho_\ell$ belongs to phase $A$ ($B$). 
However, insisting on perfect knowledge of each $\rho_\ell$ is  impractical for a variety of reasons. Instead, we assume that we have access to classical shadows of $\rho_\ell$. 
The raw data $S_T (\rho_\ell)$ behind each classical shadow is a 2-dimensional array,
\begin{align}
S_T (\rho_\ell) = \left\{ |s_i^{(t)} \rangle:\; i \in \{1,\ldots,n\},t \in \{1,\ldots,T\} \right\}
\quad \text{where} \quad |s_i^{(t)}\rangle \in \{\ket{0}, \ket{1}, \ket{+}, \ket{-}, \ket{\mathrm{i}+}, \ket{\mathrm{i}-}\} . \label{eq:snapshot-array}
\end{align}
In the main text, we propose to use this data to train a support vector machine based on the shadow kernel 
\begin{equation}
k^{\mathrm{(shadow)}} \left( S_T (\rho_\ell),\tilde{S}_T (\rho_{\ell'}) \right)
= \exp \left( \frac{\tau}{T^2} \sum_{t,t'=1}^T \exp \left( \frac{\gamma}{n} \sum_{i=1}^n \mathrm{tr} \left( \big ( 3|s_i^{(t)}\rangle \! \langle s_i^{(t)}|-\mathbb{I} \big) \big( 3|\tilde{s}_i^{(t)}\rangle \! \langle \tilde{s}_i^{(t)}|-\mathbb{I} \big) \right) \right) \right).
\label{eq:shadow-kernel-ml}
\end{equation}
This specific choice of (deterministic) kernel function allows us to carry out a thorough theoretical analysis of the entire learning procedure; see Appendix~\ref{sec:proofthmPHASEC}. 

But there are other sensible kernels that may perform even better in practice.
For instance, we could feed the two-dimensional data array~\eqref{eq:snapshot-array} into a neural network architecture, e.g.\ a convolutional neural network. 
In the limit of an infinite number of neurons in each hidden layer, this produces the neural tangent kernel 
$k^{\mathrm{(NTK)}} \left( S_T (\rho_\ell),\tilde{S}_T (\rho_{\ell'}) \right)$ \cite{jacot2018neural}. 
This kernel is positive-semidefinite and should be viewed as a measure of similarity induced by the trained neural network. Mercer's theorem~\cite{mercer1909functions} allows us to make this intuition precise by reformulating the neural tangent kernel as a Gram matrix in feature space:
\begin{equation}
k^{\mathrm{(NTK)}} \left( S_T (\rho_\ell), \tilde{S}_T (\rho_{\ell'}) \right) = \left\langle \phi^{\mathrm{(NTK)}} \left( S_T (\rho_\ell) \right), \phi^{\mathrm{(NTK)}} \left( \tilde{S}_T (\rho_{\ell'}) \right) \right\rangle~.
\end{equation}
Hence, any infinite-width neural network with input array $S_T (\rho)$ induces a feature map $\phi^{\mathrm{(NTK)}}$ that can be used instead of the doubly-infinite feature map $\phi^{\mathrm{(shadow)}}$ (\ref{eq:featureshadow}) 
that is associated with the shadow kernel~\eqref{eq:shadow-kernel-ml}.

\section{Details regarding numerical experiments} \label{sec:numdetail}

In this appendix, we provide additional numerical experiments as well as more details about the numerical experiments described in the main text.

\subsection{Additional numerical experiments} \label{sec:more-numer}

\begin{figure}[t]
    \centering
    \includegraphics[width=1.0\linewidth]{rydberg_Z_all.pdf}
    \caption{Numerical experiment for predicting ground state properties (Pauli-$Z$ in each atom) in a 1D Rydberg atom system with 51 atoms. We use the classical ML to predict the ground state properties at the three testing points. Also shown are ``predictions'' obtained from the training data nearest to the testing points. The markers denote predicted values, while the solid lines denote exact values obtained from DMRG. }
    \label{fig:rydberg-Z}
    \centering
    \includegraphics[width=1.0\linewidth]{rydberg_X_all.pdf}
    \caption{Numerical experiment for predicting ground state properties (Pauli-$X$ in each atom) in a 1D Rydberg atom system with 51 atoms. We use the classical ML to predict the ground state properties at the three testing points. Also shown are ``predictions'' obtained from the training data nearest to the testing points. The markers denote predicted values, while the solid lines denote exact values obtained from DMRG. }
    \label{fig:rydberg-X}
    \vspace{1em}
    \centering
    \includegraphics[width=1.0\linewidth]{rydberg_Spline_all.pdf}
    \caption{``Predictions'' obtained by performing bivariate B-spline interpolation using the training data.
    The markers denote interpolated values, while the solid lines denote exact values obtained from DMRG. }
    \label{fig:rydberg-spl}
\end{figure}

\begin{figure}[t]
    \centering
    \includegraphics[width=0.82\linewidth]{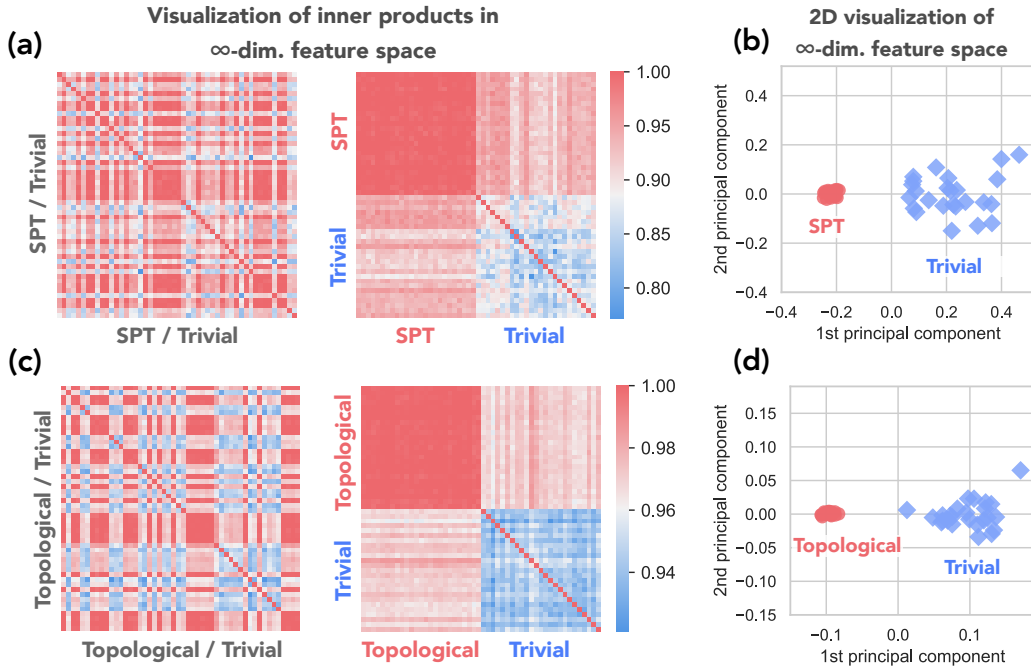}
    \caption{ Numerical experiments for distinguishing trivial and topological phases. Trivial or topological states are generated by applying low-depth local random quantum circuits to a product state or exactly solved topological state respectively. (a) \textsc{Kernel matrix for SPT/trivial phases} The exactly solved topological state is the cluster state. The $(i, j)$-entry denotes the inner product of the $i$-th and $j$-th feature vectors in the infinite-dimensional feature space defined by the classical shadow representation. To the left, states from the two phases are randomly mixed. To the right, the two phases are ordered. (b) \textsc{Kernel matrix for topologically-ordered/trivial phases.} The exactly solved topological state is the toric code ground state.}
    \label{fig:kernelmatrix}
\end{figure}

\emph{Rydberg atom chain} --- In the main text, we have provided partial prediction outcomes for a one-dimensional chain of $n=51$ Rydberg atoms; see Figure~\ref{fig:rydberg}. Here, we supply predictions of expectation values of Pauli operators $Z_i$ and $X_i$ on all $51$ atoms at the testing points marked in Figure~\ref{fig:rydberg}(b). These are shown in Figure~\ref{fig:rydberg-Z} and Figure~\ref{fig:rydberg-X}, respectively.
These extend the more restricted presentation in the main text to all qubits.
In Figure~\ref{fig:rydberg-spl}, we show a different baseline considering bivariate B-spline interpolation from the training data.

\emph{Distinguishing an SPT phase from a trivial phase} --- We consider a one-dimensional chain of $n=50$ qubits with $Z_2 \times Z_2$ symmetry. The 1D cluster state is in the nontrivial SPT phase. We generate other representatives of the nontrivial SPT phase by applying  symmetric depth-$3$ geometrically local random quantum circuits to the cluster state, and we generate representatives of the trivial phase by applying symmetric depth-$3$ random circuits to a product state.

Randomized Pauli measurements are performed $T=500$ times to convert the states to their classical shadows, and these classical shadows are mapped to feature vectors in the infinite-dimensional feature space using the feature map $\phi^{\mathrm{(shadow)}}$ (\ref{eq:featureshadow}).
In Figure~\ref{fig:kernelmatrix}(a), inner products of feature vectors (matrix elements of the shadow kernel) are displayed. 
Figure~\ref{fig:kernelmatrix}(b) shows the feature vectors projected onto a two-dimensional subspace using \emph{nonlinear} principal component analysis (PCA) based on the shadow kernel $k^{\mathrm{(shadow)}}$. Both figures show that feature vectors representing distinct phases can be distinguished easily. 
Correspondingly, the classical ML efficiently learns how to classify phases accurately, even if the training data is unlabeled. 

\emph{Distinguishing a topologically-ordered phase from a trivial phase} --- We consider the task of distinguishing the toric code \cite{kitaev2003fault} topologically-ordered phase from the trivial phase in a system of $n=200$ qubits. We generate other representatives of the topologically-ordered phase by applying two-dimensional depth-$3$ geometrically local random quantum circuits to the toric code state, and we generate representatives of the trivial phase by applying two-dimensional depth-$3$ random circuits to a product state.

Randomized Pauli measurements are performed $T=500$ times to convert the states to their classical shadows, and these classical shadows are mapped to feature vectors in the infinite-dimensional feature space using the feature map $\phi^{\mathrm{(shadow)}}$.
In Figure~\ref{fig:kernelmatrix}(c, d), inner products of feature vectors (matrix elements of the shadow kernel) and the projection of feature space data onto the two-dimensional subspace spanned by the largest principal components is shown. Once more, one can clearly see 
that feature vectors representing distinct phases can be distinguished easily.
Correspondingly, the classical ML efficiently learns how to classify phases accurately, even if the training data is unlabeled. 

\subsection{Ground state properties of the Rydberg atom chain}
\label{sec:numdetail-groundstate}

Our first example is a one-dimensional chain of $n = 51$ Rydberg atoms~\cite{Fendley2004,browaeys_many-body_2020,bernien2017probing}. Each atom can be in either its ground state or a highly excited Rydberg state. Such systems can effectively be regarded as a qubit, where the basis state $|0\rangle$ is the ground state $|g \rangle$ and the basis state $|1\rangle$ is the Rydberg state $|r\rangle$. The Hamiltonian of the atomic chain is
\begin{equation}
    H = \frac{\Omega}{2}\sum_i X_i - \Delta\sum_i N_i + \Omega\sum_{i<j}\left(\frac{R_b}{a|i-j|}\right)^6N_iN_j~,
\end{equation}
where $\Omega$ is the (fixed) Rabi frequency, $\Delta$ is the laser detuning, $N_i$ is the Rydberg occupation number operator, $a$ is the separations of the atoms, and $R_b$ is the so called Rydberg blockade radius. For large and negative $\Delta$, the ground state of $H$ is a vacuum state, where all atoms are in the ground state $|g\rangle$. In contrast, for large and positive $\Delta$, different broken-symmetry ground states can be engineered depending on the value of $R_b$.

Approximations of the exact ground states of the Rydberg chain were found using the density-matrix renormalization group (DMRG) based on matrix product states (MPS). Starting from a random MPS with bond dimension $\chi=10$, we variationally optimize the MPS using a singular value decomposition (SVD) cutoff of $10^{-9}$. We perform a number of DMRG sweeps until the change in energy is below $\epsilon=10^{-6}$. Upon convergence, we perform randomized Pauli measurements simply by performing local rotations into the corresponding Pauli bases, and sampling the resulting state~\cite{FerrisSampling}.

In Figure~\ref{fig:rydberg}(b), the color in the phase diagram corresponds to the phase obtained by two order parameters for characterizing $Z_2$ and $Z_3$ order.
For $Z_2$ order, where the atoms are in $\ket{rgrgrg\ldots}$ or $\ket{grgrgr\ldots}$, we consider the order parameter,
\begin{equation}
    O_{Z_2} = \frac{1}{n-1} \sum_{i=1}^{n-1} \left(\ketbra{r_i g_{i+1}}{r_i g_{i+1} } + \ketbra{g_i r_{i+1} }{g_i r_{i+1}} \right).
\end{equation}
For $Z_3$ order, where the atoms are in $\ket{rggrgg\ldots}$ or $\ket{grggrg\ldots}$ or $\ket{ggrggr\ldots}$, we consider the order parameter,
\begin{equation}
    O_{Z_3} = \frac{1}{n-2} \sum_{i=1}^{n-2} \left(\ketbra{r_i g_{i+1} g_{i+2}}{r_i g_{i+1} g_{i+2}} + \ketbra{g_i r_{i+1} g_{i+2}}{g_i r_{i+1} g_{i+2}} + \ketbra{g_i g_{i+1} r_{i+2}}{g_i g_{i+1} r_{i+2}} \right).
\end{equation}
We estimate the two order parameters of the ground state $\rho$.
First we check which order parameter ($O_{Z_2}$ or $O_{Z_3}$) yields a larger expectation value.
Then, we check if that expectation value is larger than the threshold value $0.8$.
If $O_{Z_2} > O_{Z_3}$ and $O_{Z_2} > 0.8$, we associate the state with the $Z_2$-order phase (red color).
Else if $O_{Z_3} > O_{Z_2}$ and $O_{Z_3} >0.8$, we say that the state is in the $Z_3$-order phase (vanilla color).
If neither of these conditions is satisfied (both expectation values are less than $0.8$), we assign the  disordered phase (blue color) to this state.

For the Rydberg atom experiment, the input parameter vector $x$ is two-dimensional.
We first normalize the values to lie within a square $[-1, 1]^2$.
Then we consider classical machine learning models given by
\begin{equation}
    \hat{\sigma}_N(x) = \sum_{\ell = 1}^N \kappa(x, x_\ell) \sigma_T(x_\ell) = \sum_{\ell = 1}^N \underbrace{\left( \sum_{\ell' = 1}^N k(x, x_{\ell'}) (K+\lambda I)^{-1}_{\ell' \ell} \right)}_{\kappa(x, x_\ell)} \sigma_T(x_\ell),
\end{equation}
where $\lambda > 0$ is a parameter to regularize the model when $K$ is not invertible, $\sigma_T(x_\ell)$ is shorthand for $\sigma_T \left( \rho_\ell \right)$ and denotes the classical shadow representation of the ground state $\rho_\ell = \rho(x_\ell)$ under $T$ randomized Pauli measurements.
Moreover, $K_{ij} = k(x_i, x_j)$ is the kernel matrix, $k(x, x')$ is a kernel function, and $\kappa(x, x_\ell)$ is a function that depends on the kernel function, the kernel matrix $K$, and $\lambda$.
We consider a set of different regularization parameters,
\begin{equation}
    \lambda \in \{0.0125, 0.025, 0.05, 0.125, 0.25, 0.5, 1.0, 2.0, 4.0, 8.0\},
\end{equation}
and we also consider a set of different kernel functions $k(x, x') = \tilde{k}(x, x') / \sqrt{\tilde{k}(x, x) \tilde{k}(x', x')}$, where
\begin{subequations}
\begin{align}
    \tilde{k}(x, x') &= \exp(-\gamma \norm{x - x'}_2^2), &\quad \mbox{(Gaussian kernel)},\\
    \tilde{k}(x, x') &= \sum_{k_1 = -3}^3 \sum_{k_2 = -3}^3 \cos\left(\pi (k_1 (x_1 - x_1') + k_2 (x_2 - x_2') ) \right), &\quad \mbox{(Dirichlet kernel)},\label{eq:numer-dirichlet-kernel}\\
    \tilde{k}(x, x') &= k^{\mathrm{(NTK)}}(x, x'), &\quad \mbox{(Neural tangent kernel)}.
\end{align}
\end{subequations}
The hyperparameter $\gamma > 0$ in the Gaussian kernel is chosen to be equal to $N^2 / \sum_{i=1}^N \sum_{j=1}^N \norm{x_i - x_j}_2^2$, the inverse of the average distance between $x_i$ and $x_j$.
We consider the neural tangent kernel $k^{\mathrm{(NTK)}}(x, x')$ \cite{jacot2018neural, neuraltangents2020} that is equivalent to an infinite-width feed-forward neural network with $2, 3, 4, 5$ hidden layers and that uses the rectified linear unit (ReLU) as the activation function. Computing the neural tangent kernel can be implemented easily using the open-source software Neural Tangents \cite{neuraltangents2020}.
Suppose that the input data $\{x_{\ell}\}_{\ell=1}^N$ is stored in a \pythoninline{numpy} array of size $N \times m$, denoted as \pythoninline{dataX} in the following code. We can use then use following code to generate the neural tangent kernel matrix. The imported package \pythoninline{neural_tangents} can be downloaded from \url{https://github.com/google/neural-tangents}.

\begin{python}
import jax
import numpy as np
from neural_tangents import stax

init_fn, apply_fn, kernel_fn = stax.serial(
    stax.Dense(32), stax.Relu(),
    stax.Dense(32), stax.Relu(),
    stax.Dense(1)
)
kernel_NN2 = kernel_fn(dataX, dataX, 'ntk')

init_fn, apply_fn, kernel_fn = stax.serial(
    stax.Dense(32), stax.Relu(),
    stax.Dense(32), stax.Relu(),
    stax.Dense(32), stax.Relu(),
    stax.Dense(1)
)
kernel_NN3 = kernel_fn(dataX, dataX, 'ntk')
                
init_fn, apply_fn, kernel_fn = stax.serial(
    stax.Dense(32), stax.Relu(),
    stax.Dense(32), stax.Relu(),
    stax.Dense(32), stax.Relu(),
    stax.Dense(32), stax.Relu(),
    stax.Dense(1)
)
kernel_NN4 = kernel_fn(dataX, dataX, 'ntk')

init_fn, apply_fn, kernel_fn = stax.serial(
    stax.Dense(32), stax.Relu(),
    stax.Dense(32), stax.Relu(),
    stax.Dense(32), stax.Relu(),
    stax.Dense(32), stax.Relu(),
    stax.Dense(32), stax.Relu(),
    stax.Dense(1)
)
kernel_NN5 = kernel_fn(dataX, dataX, 'ntk')

list_kernel_NN = [kernel_NN2, kernel_NN3, kernel_NN4, kernel_NN5]

# Normalization of the kernel matrix
for r in range(len(list_kernel_NN)):
    for i in range(len(list_kernel_NN[r])):
        for j in range(len(list_kernel_NN[r])):
            list_kernel_NN[r][i][j] /= (list_kernel_NN[r][i][i] \
                                        * list_kernel_NN[r][j][j]) ** 0.5
\end{python}

In order to predict 
the expectation value $\Tr(O \hat{\sigma}_N(x))$ of an observable $O$ for a new ground state
$\hat{\sigma}_N(x)$, we utilize the following property of expectation values,
\begin{equation}
    \Tr(O \hat{\sigma}_N(x)) = \sum_{\ell = 1}^N \kappa(x, x_\ell) \Tr(O \sigma_T(x_\ell)).
\end{equation}
Hence, we first compute $\Tr(O \sigma_T(x_\ell))$, which can be done efficiently for $r$-body observables that factorize nicely into tensor products. Indeed, an $O = O_{i_1} \otimes \ldots \otimes O_{i_r}$ ensures
\begin{equation}
    \Tr(O \sigma_T(x_\ell)) = \frac{1}{T} \sum_{t=1}^T \Tr\left(O \sigma_{1}^{(t)}(x_\ell) \otimes\dots\otimes \sigma_{n}^{(t)}(x_\ell) \right) = \frac{1}{T} \sum_{t=1}^T \Tr\left(O_{i_1} \sigma_{i_1}^{(t)}(x_\ell) \right) \ldots \Tr\left(O_{i_r} \sigma_{i_r}^{(t)}(x_\ell) \right),
\end{equation}
and the right hand side can be computed with $\mathcal{O}(Tn)$ arithmetic operations. 
Then, we can compute $\Tr(O \hat{\sigma}_N(x))$ by extrapolating $\Tr(O \sigma_T(x_\ell))$ using $\kappa(x, x_\ell)$.
We utilize scikit-learn, a Python package \cite{scikit-learn}, for the training of these machine learning models.

Due to the different classical ML models one could consider (corresponding to different regularization parameters $\lambda$ and kernel functions $k(x, x')$), we have to perform model selection to find an appropriate ML model.
Typically, the prediction performance will be quite sensitive to these parameters, so one has to select them carefully. 
To evaluate the ML models, we consider $100$ different points $x \in [-1,1]^2$ in parameter space.
Among these $100$ points, we select $N = 20$ to be training data. These are the circled points in Figure~\ref{fig:rydberg}(b).
For each property we would like to predict, we choose one of the the three kernels and the different values of $\lambda$ such that the prediction error is minimized on a validation set containing $80 - 3$ inputs of $x$.
The validation set is disjoint from the $20$ training points and the $3$ testing points for evaluating the prediction performances (special markers in Figure~\ref{fig:rydberg}(a)). Their purpose is to perform model selection.
Finally, we test on the three input $x$'s shown by the special markers (cross, diamond and star) in Figure~\ref{fig:rydberg}(b).

We found that for each property we would like to predict, the prediction performance for different classical ML model varies moderately.
When we have sufficiently large training data size $N$, most choices of $\lambda$ and the kernel function should yield good prediction performance.
However, we are using a very small number of training data in our experiments, hence the choice of these options becomes more important.
In particular, the best choice of $\lambda$ can differ quite significantly over the different properties we would like to predict.

\begin{figure}[t]
    \centering
    \includegraphics[width=0.9\linewidth]{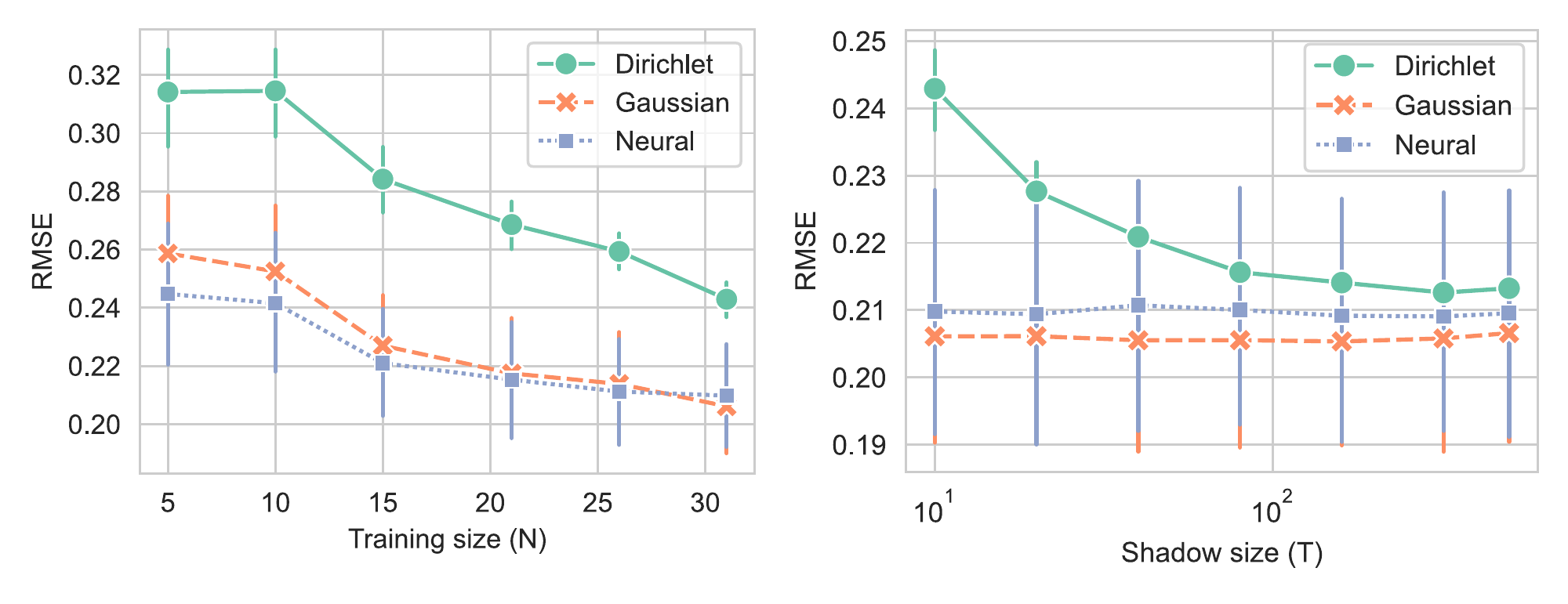}
    \caption{Numerical experiment for predicting ground state properties (Pauli-$X$ and $Z$ in each atom) in a 1D Rydberg atom system with 51 atoms under different hyperparameters. \textsc{(Left)} The prediction error (root-mean-square error) over different training sizes $N$ with a fixed number $T=10$ of randomized Pauli measurements, also referred to as the shadow size. \textsc{(Right)} The prediction error over different shadow sizes $T$ with a fixed training data size $N=31$. }
    \label{fig:rydberg-scaling}
\end{figure}

For completeness, we include a set of experiments where we vary the training data size $N$ or the classical shadow size $T$, where by ``shadow size'' we mean the number of randomized Pauli measurements used to approximate each state.
The result in given in Figure~\ref{fig:rydberg-scaling}.
{\color{black} For this set of experiments, we consider a fixed set of $70$ validation points in the phase space.
Recall that we are using the ML model to predict ground state properties.
Here, we consider the properties to be the expectation values of single-site Pauli-$X$ and Pauli-$Z$ operators.
Because there are a total of $51$ atoms, there are a total of $51 \times 2 = 102$ properties.
For each property, we randomly draw $10$ different points in the phase space (not in the training set or the validation set).
Therefore, the test set is of size $1020$, where each instance in the test set corresponds to a property of a point in the phase space.
The prediction error is given by the root-mean-square error over the $1020$ instances in the test set.
We can see that as training set size $N$ increases, the prediction becomes better.
However, we see that as the training size increases, the slope of the prediction error (RMSE) over $N$ flattens.
This is expected from the theorem we established showing that $N = m^{\mathcal{O}(1 / \epsilon)}$, where $N$ is the training set size, $m$ is the number of parameters, and $\epsilon$ is the prediction error.
While the theorem only provides an upper bound on $N$, if we assume the upper bound is saturated, then we can use elementary calculus to derive
\begin{equation}
\frac{d \epsilon}{d N} \mbox{ is proportional to } - \frac{\epsilon^2}{N\log(m)}.
\end{equation}
Hence, the analysis is compatible with the observation that the slope of RMSE over $N$ flattens as $N$ becomes larger.
}
While we proved a rigorous result using the Dirichlet kernel, other commonly used ML models may yield a better prediction performance in practice.
Proving rigorous prediction guarantees and understanding the limitations and strengths for other more commonly used ML models are important future directions.

\subsection{Ground state properties of the 2D antiferromagnetic Heisenberg model}
\label{sec:numdetail-groundstate2}

Our next example is the two-dimensional antiferromagnetic Heisenberg model. Spin-$\tfrac{1}{2}$ particles (i.e. qubits) occupy sites on a square lattice, and for each pair $(ij)$ of neighboring sites the Hamiltonian contains a term $J_{ij}\left(X_iX_j+Y_iY_j+Z_iZ_j\right)$ where the couplings $\{J_{ij}\}$ are uniformly sampled from the interval $[0,2]$. The parameter $x$ is a list of all  $J_{ij}$ couplings; hence in this case the dimension of the parameter space is $m=O(n)$, where $n$ is the number of qubits. The Hamiltonian $H(x)$ on a $5\times 5$ lattice is shown in Figure~\ref{fig:heisenberg}(a).  
The exact ground state was found using DMRG. Analogously to the Rydberg atoms experiments, we fixed the SVD cutoff to $10^{-8}$ and stopped the DMRG runs when the difference in energy was below $10^{-4}$.

The classical ML models we considered are the same as the Rydberg atom chain experiment. The only difference is that we slightly modify the Dirichlet kernel (\ref{eq:numer-dirichlet-kernel}) to
\begin{equation}
    k(x, x') = \sum_{i \neq j} \sum_{k_i = -3}^3 \sum_{k_j = -3}^3 \cos\left(\pi (k_i (x_i - x_i') + k_j (x_j - x_j') ) \right), \quad \mbox{(Dirichlet kernel)}.
\end{equation}
We trained the classical ML model using a training set containing $N = 90$ randomly chosen values of the parameter $x=\{J_{ij}\}$. 
Then, for each property we would like to predict, we find the top-performing ML model setting (out of all $\lambda$ parameters and kernel functions $k(x, x')$) on a validation set containing $100$ parameters $x$ distinct from the training set.
Finally, we test on $10$ newly sampled parameters $x$ to estimate the prediction error.
Figure~\ref{fig:heisenberg}(b) shows the prediction outcome from one of the input parameter $x$.
Figure~\ref{fig:heisenberg}(c) shows the RMSE from all $10$ input parameters.

\begin{figure}[t]
    \centering
    \includegraphics[width=0.9\linewidth]{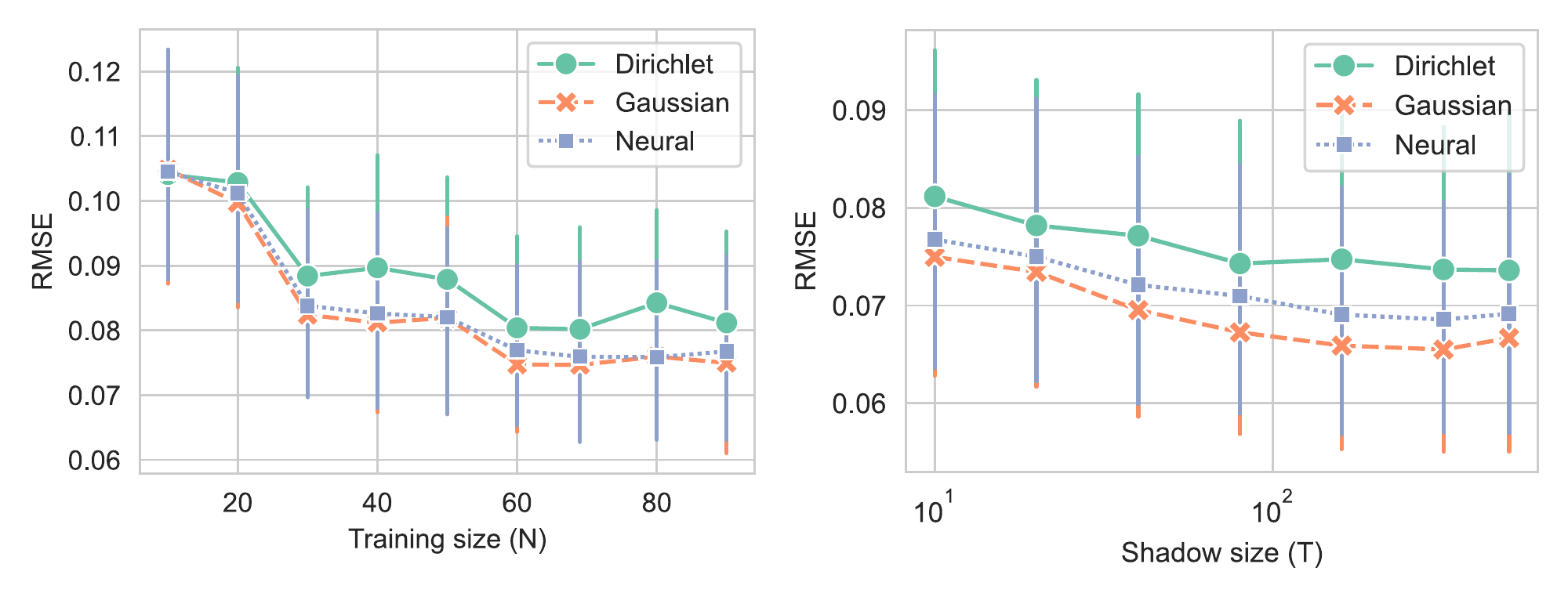}
    \caption{Numerical experiment for predicting ground state properties (two-point correlation functions) in a 2D antiferromagnetic Heisenberg model with $5\times 5$ spins under different hyperparameters. \textsc{(Left)} The predict error (root-mean-square error) over different training size $N$ with a fixed number of randomized Pauli measurements $T=10$, also referred to as the shadow size. \textsc{(Right)} The prediction error (root-mean-square error) over different shadow size $T$ with a fixed training data size $N=90$. }
    \label{fig:Heisenberg-scaling}
\end{figure}

Similar to the Rydberg atom experiment, the best-performing ML model setting differs across the properties we would like to predict.
The three kernels perform similarly at larger training data size $N$ and larger number of randomized Pauli measurements $T$.
But neural networks and Gaussian kernel methods tend to perform better in most cases.
The best choice of $\lambda$ differs substantially across the different properties: there is no single choice of $\lambda$ that performs uniformly better than the other choices.

To showcase these effects, we also include a set of experiments where we vary the training data size $N$ or the classical shadow size $T$, that is, the number of randomized Pauli measurements used to approximate each state.
The numerical results are summarized in Figure~\ref{fig:Heisenberg-scaling}.
For this set of experiments, we consider fixed sets of $100$ validation points.
{\color{black} For this set of experiments, we consider a fixed set of $70$ validation points in the phase space.
Recall that we are using the ML model to predict ground state properties.
Here, we consider the properties to be the two-point correlation functions over every pair of the $25$ spins.
This results in a total of $25 \times 25 = 625$ properties.
For each property, we randomly draw $10$ testing points in the $m=\mathcal{O}(n)$ dimensional parameter space (not in the training set or the validation set).
Therefore, the test set is of size $6250$, where each instance in the test set corresponds to a property of a point in the parameter space.
The prediction error is given by the root-mean-square error over the $6250$ instances in the test set.
The results resemble what was found in the Rydberg atom experiments, but with one notable difference --- in the Rydberg experiments, but not for the 2D antiferromagnet, the  Dirichlet kernel has the best performance for the largest shadow size $T$ we considered. This may be because the dimension $m$ of the parameter space is much lower in the Rydberg case. 
}

\subsection{Classifying phases of the bond-alternating XXZ model}
\label{sec:numdetail-phases}

To illustrate our classical ML for classifying quantum phases of matter, we consider the bond-alternating XXZ model with $n=300$ spin-$\frac{1}{2}$ particles (i.e.\ qubits). The Hamiltonian is given by
\begin{equation}
    \sum_{i: \mathrm{odd}} J(X_i X_{i+1} + Y_i Y_{i+1} + \delta Z_i Z_{i+1}) + \sum_{i: \mathrm{even}} J' (X_i X_{i+1} + Y_i Y_{i+1} + \delta Z_i Z_{i+1}),
\end{equation}
and encompasses the bond-alternating Heisenberg model ($\delta=1$), as well as the bosonic version of the Su-Schrieffer-Heeger model \cite{su1979solitons} ($\delta=0$).
The phase diagram in Figure~\ref{fig:XXZ}(b) is obtained by evaluating the partial reflection many-body topological invariant \cite{pollmann2012detection, elben2020many}.
It is given by
\begin{equation}
    \tilde{\mathcal{Z}}_{\mathcal{R}} = \frac{\mathcal{Z}_{\mathcal{R}}}{\sqrt{[\Tr(\rho_{I_1}^2) + \Tr(\rho_{I_2}^2)] / 2}}, \quad \mbox{where} \quad \mathcal{Z}_{\mathcal{R}} = \Tr(\rho_{I_1 \cup I_2} \mathcal{R}_{I_1 \cup I_2}),
\end{equation}
and we consider $I_1$ with $6$ spins: the $145$-th spin to the $150$-th spin.
Likewise, we fix $I_2$ to also contain $6$ spins: the $151$-th spin to the $156$-th spin.
Hence, the union $I_1 \cup I_2$ contains $12$ spins. The symbols $\rho_{I_1},\rho_{I_2}$ and $\rho_{I_1 \cup I_2}$ denote the reduced density matrices associated with each local region.
The reflection operator $\mathcal{R}_{I_1 \cup I_2}$ acts on the local region $I_1 \cup I_2$ and is given by
\begin{equation}
    \mathcal{R}_{I_1 \cup I_2} \ket{s_1, \ldots, s_{|I_1 \cup I_2|}} = \ket{s_{|I_1 \cup I_2|}, \ldots, s_1}, \quad \forall s_1, \ldots, s_{|I_1 \cup I_2|} \in \{0, 1\}.
\end{equation}
The partial reflection many body-topological invariant can resolve three phases: trivial ($\tilde{\mathcal{Z}}_{\mathcal{R}}=+1$), symmetry-protected topological (SPT) ($\tilde{\mathcal{Z}}_{\mathcal{R}}=-1$) and symmetry broken ($\tilde{\mathcal{Z}}_{\mathcal{R}}=0$).
In Figure~\ref{fig:XXZ}(b), we use the colors blue (trivial), red (SPT) and gray (symmetry broken) to visualize these different types of phases.

For each value of $J' / J$ and $\delta$ considered, we construct the exact ground state using DMRG, and find its classical shadow by performing randomized single-qubit Pauli measurements a total of $T=500$ times.
To simulate this experiment, we follow the same setting for DMRG used in \cite{elben2020many}.
We limit the maximum number of sweeps to $100$ and set the DMRG cutoff to $10^{-9}$.
We initialize the state to be the N\'eel state $\ket{0101\ldots}$.
To pin one of the degenerate ground state in the symmetry broken phase, we include a penalty term given by $0.1 J Z_1$ in the Hamiltonian.

After obtaining the classical shadow representation $S_T(\rho_\ell)$ for each quantum state $\rho_\ell$, we compute the kernel matrix $K \in \mathbb{R}^{N \times N}$, where each entry is given by the shadow kernel $k^{\mathrm{(shadow)}}(S_T(\rho_\ell), S_T(\rho_{\ell'}))$.
Recall that the shadow kernel is defined as
\begin{equation}
    k^{\mathrm{(shadow)}}(S_T(\rho), S_T(\tilde{\rho})) = \exp \left(\frac{1}{T^2} \sum_{t,t'=1}^T \exp\left( \frac{1}{n} \sum_{i=1}^n \Tr\left( \sigma_{i}^{(t)}\tilde{\sigma}_{i}^{(t^{\prime})} \right) \right)\right),
     \,\, \text{where} \,\,
  \sigma_{i}^{(t)}=3\ketbra{s_{i}^{(t)}}{s_{i}^{(t)}}-\mathbb{I},
\end{equation}
and the classical shadow representation is given by
\begin{equation}
    S_T(\rho) = \left\{ \ket{s_{i}^{(t)}}:\; i \in \left\{1,\ldots,n\right\},\; t \in \left\{1,\ldots,T\right\} \right\}, 
    \quad \text{where} \quad
    \ket{s_{i}^{(t)}} \in \{\ket{0}, \ket{1}, \ket{+}, \ket{-}, \ket{\mathrm{i}+}, \ket{\mathrm{i}-}\}~.
\end{equation}
Care should be taken when computing diagonal elements of the kernel matrix $K$. 
The problem is that for $\rho=\tilde{\rho}$ and $t=t'$, we necessarily have $\mathrm{tr} \left( \sigma_i^{(t)} \tilde{\sigma}_i^{(t)} \right) = 5$ for all $1 \leq i \leq n$. And the double exponential will amplify this already substantial contribution enormously. 
We found that counteracting this blow-up improves the numerical stability of the kernel method substantially. 
When $\ell = \ell'$, when we compute $k^{\mathrm{(shadow)}}(S_T(\rho_\ell), S_T(\rho_\ell))$, we sum over $t \neq t'$ instead of all $t, t'$. In particular, when $\rho = \tilde{\rho}$, we consider a slight modification to the kernel definition,
\begin{equation}
    k^{\mathrm{(shadow)}}(S_T(\rho), S_T(\rho)) = \exp \left(\frac{1}{T(T - 1)} \sum_{t \neq t' }
    \exp\left( \frac{1}{n} \sum_{i=1}^n \Tr\left( \sigma_{i}^{(t)} \sigma_{i}^{(t^{\prime})} \right) \right)\right),
\end{equation}
This modification also seems to slightly improve the classification performance.

After evaluating the kernel matrix $K$, we renormalize the entries to obtain 
the standardized kernel matrix
\begin{equation}
    \overline{K}_{ \ell \ell'} =  \frac{K_{\ell \ell'}}{\sqrt{K_{\ell \ell} K_{ \ell' \ell'}}} \quad \text{for} \quad  \ell,\ell' \in \left\{1, \ldots, N\right\}.
\end{equation}
Subsequently, we perform kernel principal component analysis (PCA) on $\overline{K}$.
The implementation we used for kernel PCA is based on scikit-learn \cite{sklearn_api}.
The output of kernel PCA is a list of low-dimensional vectors (the dimension can be chosen arbitrarily, but we choose two dimensions for this experiment).
Each low-dimensional vector corresponds to a quantum state.
In Figure~\ref{fig:XXZ}(c, d), we can see that the low-dimensional vectors are clustered into different quantum phases of matter.

\emph{Distinguishing an SPT phase from a trivial phase} --- We consider a one-dimensional chain of $n=50$ qubits with $Z_2 \times Z_2$ symmetry. The 1D cluster state is in the nontrivial SPT phase. We generate other representatives of the nontrivial SPT phase by applying symmetric low-depth geometrically local random quantum circuits to the cluster state, and we generate representatives of the trivial phase by applying symmetric random circuits to a product state. 
We simulate the application of symmetric low-depth geometrically local random quantum circuits to the cluster state through matrix product states (MPS). Each circuit layer consists of patterns of random two-qubit gates acting on next-to-nearest neighbors sites. We generate the random gates in a block-sparse structure in the parity symmetry sectors. This choice, together with the choice of connectivity, guarantees that the $Z_2 \times Z_2$ symmetry is conserved during the circuit evolution.

Randomized Pauli measurements are performed $T=500$ times to convert the states to their classical shadows. 
We perform kernel PCA to find low-dimensional representation for the quantum states using exactly the same method as the experiment on bond-alternating XXZ model.

\subsection{Distinguishing a topological phase from a trivial phase}
\label{sec:numdetail-phases2}

We consider the task of distinguishing the toric code topological phase from the trivial phase in a system of $n=200$ qubits. 
Kitaev's toric code state \cite{kitaev2003fault}  is in the nontrivial topologically-ordered phase, while a product state represents the trivial phase. To populate both phases, we apply low-depth geometrically local random Clifford circuits~\cite{aaronson2004improved}
 to Kitaev's toric code state \cite{kitaev2003fault} with code distance 10, and we generate representatives of the trivial phase by applying random Clifford circuits to a product state. We utilize Clifford circuits to ensure efficient simulation of in total $n=200$ qubits (and with a depth up to $9$) by means of the Gottesman-Knill theorem.
We again perform kernel PCA to find low-dimensional representations for the quantum states using exactly the same method as the experiment on bond-alternating XXZ model. This is used to generate the plot in Figure~\ref{fig:topophase}(b) for a one-dimensional projection of the feature space, as well as the plot in Figure~\ref{fig:kernelmatrix}(d) for a two-dimensional projection.

For the unsupervised ML model shown in Figure~\ref{fig:topophase}(c), we consider a combination of kernel PCA and randomized projections \cite{pmlr-v23-karnin12}.
First we perform kernel PCA to map the data to a six-dimensional subspace of the infinite-dimensional feature space.
Then we repeat the following procedure $500$ times.
We select a one-dimensional subspace uniformly at random in the six-dimensional subspace.
We project all the quantum states to the one-dimensional subspace.
Then, we find the center point (according to median instead of mean) to split up the quantum states into two phases.
We also record the sum of the absolute values from all points to the center point in the one-dimensional subspace.
Finally, we consider the classification obtained from the random one-dimensional projection that results in the largest sum of the absolute values.

For the convolutional neural network (CNN) approach shown in Figure~\ref{fig:topophase}(c), we consider the following CNN built from Keras \cite{chollet2015keras}.

\begin{python}
import tensorflow as tf
from tensorflow.keras import datasets, layers, models

model = models.Sequential()
model.add(layers.Conv2D(32, (2, 2), activation='relu', padding='same',
            input_shape=(2*L, L, 6)))
model.add(layers.MaxPooling2D((2, 2)))
model.add(layers.Conv2D(32, (2, 2), activation='relu', padding='same'))
model.add(layers.MaxPooling2D((2, 2)))
model.add(layers.Conv2D(32, (2, 2), activation='relu', padding='same'))
model.add(layers.Flatten())
model.add(layers.Dense(32, activation='relu'))
model.add(layers.Dense(2))
\end{python}

In the above code, $L$ is the code distance for the toric code and is equal to $10$ in this experiment (recall that toric code ground state has $n = 2L^2$ qubits).
This CNN model is supervised and requires a training data with a corresponding label for indicating which phase the training data point is in. 
We first perform the Pauli-$6$ POVM on each qubit \cite{carrasquilla2019reconstructing} to transform the quantum state into a array of size $n$ where each entry has six outcomes.
We perform one-hot encoding to yield a classical vector of size $6n$, where each entry in the classical vector is either $0$ or $1$.
Because the toric code ground state is two-dimensional ($2L \times L$), we restructure the classical vector into a three-dimensional tensor of size $2L \times L \times 6$.
The first two dimensions corresponds to the spatial dimension of the toric code ground state.
The last dimension corresponds to the one-hot encoded vector for the six-outcome POVM.
We then train the above model using the Adam optimizer \cite{kingma2014adam} with the categorical cross entropy as the loss function. The code is given below.

\begin{python}
model.compile(optimizer='adam',
    loss=tf.keras.losses.SparseCategoricalCrossentropy(from_logits=True),
    metrics=['accuracy'])
\end{python}

We train the convolutional neural network using $100$ training points (half are topologically-ordered states, and the other half are trivial states).
Then we use a validation set of $100$ points to perform early stopping.
This is because the longer we train, the more likely the neural network is going to overfit.
Hence, it is a good practice to perform model selection by choosing which model to use at different time points (during the training process).
We choose the model that performs the best on the validation set.
Then we test the classification accuracy (the percentage that the prediction of the phases is correct) on a testing set consisting of $100$ points.

The performance of the above ML model is not substantially different from random guessing.
Hence, we also consider a very simple CNN enhanced with classical shadow under $T = 500$ randomized Pauli measurements.
In particular, we compute the local reduced density matrix using the classical shadow.
Then for each qubit, we represent it with the local reduced density matrix.
For simplicity, we consider the $i$-th qubit to be represented by a vector of size $16$, which includes the $2$-body reduced density matrix for the subsystem consisting of the $i$-th and the $i+1$-th qubit.
Hence, each quantum state is now represented by a classical vector of dimension $2L^2 \times 16$.
We reshape the classical vector into a three-dimensional tensor of size $2L \times L \times 16$.
The classical vector is feed into the convolutional neural network structured as follows.
We also apply the Adam optimizer \cite{kingma2014adam} with the categorical cross entropy as the loss function.
The evaluation process is exactly the same as the CNN approach based on the Pauli-$6$ POVM.

\begin{python}
import tensorflow as tf
from tensorflow.keras import datasets, layers, models

model = models.Sequential()
model.add(layers.Conv2D(16, (1, 1), activation='relu',\
            padding='same', input_shape=(2*L, L, 16)))
model.add(layers.MaxPooling2D((2, 2)))
model.add(layers.Conv2D(16, (2, 2), activation='relu',\
            padding='same'))
model.add(layers.MaxPooling2D((2, 2)))
model.add(layers.Conv2D(16, (2, 2), activation='relu',\
            padding='same'))
model.add(layers.MaxPooling2D((2, 2)))
model.add(layers.Flatten())
model.add(layers.Dense(32, activation='relu'))
model.add(layers.Dense(2))

model.compile(optimizer='adam',
    loss=tf.keras.losses.SparseCategoricalCrossentropy(
            from_logits=True),
    metrics=['accuracy'])
\end{python}

\section{Proof idea for the efficiency in predicting ground states}
\label{sec:proofideaGSUPP}

\subsection{Main result}
\label{sec:qubit-GSUPP}

In order to illustrate the proof of Theorem~\ref{thm:mainFourier}, let us begin by looking at a simpler task: training a machine learning model to predict a specified ground state property instead of the classical representation of the ground state.
Consider the property $\Tr(O \rho)$, where $\rho$ is the ground state and $O$ is a local observable.
In this simpler task, we consider the training data to be
\begin{equation}
\big\{x_1 \rightarrow \Tr(O \rho(x_1)), \quad \ldots, \quad x_N \rightarrow \Tr(O \rho(x_N)) \big\},
\end{equation}
where $x_{\ell} \in [-1, 1]^m$ is a classical description of the Hamiltonian $H(x_\ell)$ and $\rho(x_\ell)$ is the ground state of $H(x)$.
Intuitively, in a quantum phase of matter, the ground state property $\Tr(O \rho(x))$ changes smoothly as a function of the input parameter $x$.
The smoothness condition can be rigorously established as an upper bound on the average magnitude of the gradient of $\Tr(O \rho(x))$ using quasi-adiabatic evolution \cite{hastings2005quasiadiabatic, bachmann2012automorphic}, assuming that the spectral gap of $H(x)$ is bounded below by a nonzero constant throughout the parameter space. 
The upper bound on the average gradient magnitude enables us to design a simple classical ML model based on an $l_2$-Dirichlet kernel for generalizing from the training set to a new input $x \in \left[-1,1\right]^m$:
\begin{equation}
    \hat{O}_N(x) = \frac{1}{N} \sum_{\ell=1}^N \kappa(x, x_{\ell}) \Tr(O \rho(x_{\ell})) \,\,\, \mbox{with} \,\,\, \kappa(x, x_{\ell}) = \sum_{k \in \mathbb{Z}^m, \norm{k}_2 \leq \Lambda} \cos(\pi k \cdot (x - x_{\ell})) \in \mathbb{R}.
\end{equation}
The $l_2$-Dirichlet kernel is often used in the study of high-dimensional Fourier series \cite{weisz2012summability} and the proposed ML model is equivalent to learning a truncated Fourier series to approximate the function $\Tr(O \rho(x))$, where the parameter $\Lambda$ is a cutoff on the wavenumber $k$ that depends on the upper bound on the gradient of $\Tr(O \rho(x))$.
Using statistical analysis, one can guarantee that $ \E_{x} |\hat{O}_N(x) - \Tr(O \rho(x))|^2 \leq \epsilon$ as long as the amount of training data $N = m^{\mathcal{O}(1 / \epsilon)}$ where our big-$\cal{O}$ notation is with respect to the $m \to\infty$ limit.
Hence, we can achieve a small \emph{constant} prediction error with an amount of training data and computational time that are both polynomial in the number $m$ of input parameters. The training is efficient because the number of modes needed for the truncated Fourier series to provide an accurate approximation to $\Tr(O \rho)$ scales polynomially with $m$.

The key to the statistical analysis is to bound the model complexity of the above machine learning model.
In particular, the model complexity depends on the number of wave vectors we consider in the $l_2$-Dirichlet kernel.
The more wave vectors $k$ we include, the higher the model complexity; and we would have to use more data to train the ML model to achieve good generalization performance.
Furthermore, one could show that the amount of data is proportional to the number of wave vectors we consider.
In order to achieve a prediction error $ \E_{x} |\hat{O}_N(x) - \Tr(O \rho(x))|^2 \leq \epsilon$, we would need to select $\Lambda$ to be of order $\sqrt{1 / \epsilon}$.
Hence, the number of wave vectors is proportional to the number of lattice points in an $m$-dimensional $l_2$ ball of radius $\Lambda$.
The volume of an $m$-dimensional $l_2$ ball with radius $\Lambda$ is proportional to $\Lambda^{m} = (1 / \epsilon)^{m/2}$.
If the number of lattices points is proportional to the volume, then this would imply an exponential scaling in the number of parameters $m$.
However, through a proper combinatorial analysis, we show that the number of lattices points is actually proportional to $m^{\mathcal{O}(\Lambda^2)} = m^{\mathcal{O}(1 / \epsilon)}$, which is only polynomial in the number of parameters $m$.

We can build on this idea to address the task of predicting ground state representations. Now instead of predicting $\Tr(O \rho)$ for a new input $x$, the goal is to predict the classical shadow of the ground state $\rho(x)$. 
We consider the training data to be $\big\{x_{\ell} \rightarrow \sigma_1(\rho(x_{\ell}))\big\}_{\ell = 1}^N$, where $\sigma_1(\rho(x_{\ell}))$ is the classical shadow representation of $\rho(x_{\ell})$ obtained from just a \textit{single} randomized Pauli measurement of the state (the $T=1$ case of Eq.~\eqref{eq:sigma-T-shadow}). Following the same approach as outlined above for the case of predicting a single property, the predicted ground state representation is now given by
\begin{equation}
    \hat{\sigma}_N(x) = \frac{1}{N} \sum_{\ell = 1}^N \kappa(x, x_{\ell}) \sigma_1(\rho(x_{\ell}))  \,\,\, \mbox{with} \,\,\, \kappa(x, x_{\ell}) = \sum_{k \in \mathbb{Z}^m, \norm{k}_2 \leq \Lambda} \cos(\pi k \cdot (x - x_{\ell})) \in \mathbb{R}.
\end{equation}
One can then guarantee that this representation accurately predicts expectation values for a wide range of observables.

The fact that only a single snapshot $\sigma_1$ per parameter point is required for our protocol may be surprising. However, since the snapshots depends on the parameters, sampling over training data indirectly samples over different snapshots, and is thus sufficient for a reasonable estimate of properties of the phase. The estimate can of course be further improved if multiple snapshots are used for each parameter point, and we leave proving such improved bounds as an exciting goal for future work.

\color{black}
\subsection{Generalization to other systems and settings}
\label{sec:generalize-GSUPP}

In this subsection, we discuss how one could generalize the proof of Theorem~\ref{thm:mainFourier} to various different scenarios.

\subsubsection{Prediction based on other quantum measurements}

Throughout this work, we considered classical shadows based on randomized Pauli measurements \cite{huang2020predicting}.
However, it may be difficult to perform randomized Pauli measurements in some experimental systems.
Theorem~\ref{thm:mainFourier} can be directly generalized to other kinds of measurement procedures.
Consider a restricted setting where the experimentalist can only obtain training data of the form
\begin{equation}
    \left\{ x_\ell \rightarrow \Tr(O \rho(x_\ell)) \right\}_{\ell = 1}^N,
\end{equation}
for a single observable $O$ (that can be written as a sum of local observables).
In this case, the classical ML model can no longer predict a classical representation of $\rho(x)$ for a new $x$.
Nevertheless, the classical ML model can still predict $\Tr(O \rho(x))$ accurately for a new $x$ by following the proof sketch in Appendix~\ref{sec:qubit-GSUPP}.

More generally, suppose the experimentalist can construct some classical representation of the ground state $\rho(x_\ell)$ through the available measurements, such as classical shadows based on another random unitary ensemble \cite{hu2021hamiltonian}, or simply a list of properties of $\rho(x_\ell)$. And suppose that the classical representation allows us to predict the expectation values of observables $O_1, O_2, \ldots, O_M$ in the ground state $\rho(x_\ell)$.
Then for a new $x$, the classical ML model can predict $\Tr(O_i \rho(x))$ accurately for $i = 1, \ldots, M$.

\subsubsection{A variable number of parameters}

So far, we have considered the input vector $x$ to be of a fixed dimension $m$.
Here we briefly discuss how to generalize Theorem~\ref{thm:mainFourier} to a setting where the input is not a fixed dimensional vector.
We can think of the input as $\xi = (m, x)$, where $m \in \mathbb{N}$ is a discrete variable specifying the number of parameters, and $x \in \mathbb{R}^m$ is an $m$-dimensional vector with continuous entries. The number of parameters $m$ may range from 
$m_{\mathrm{min}}$ to $m_{\mathrm{max}}$.
We consider a class of Hamiltonians $H(\xi) = H((m, x))$ that depends on both the discrete parameter $m$ and the continuous vector $x$.
For example, we may have
\begin{align}
    m&=1: & H((m, x)) &= \sum_{i=1}^n x_1 (X_i X_{i+1} + Y_i Y_{i+1}),\\
    m&=2: & H((m, x)) &= \sum_{i=1}^n x_1 (X_i X_{i+1} + Y_i Y_{i+1}) + x_2 ( Z_i Z_{i+1} ),
\end{align}
where $x_1, x_2$ denote the first and second entry of the vector $x$.
In order the train the ML model, we can consider training data to be of the form
\begin{equation}
    \big\{ \xi_{\ell} \rightarrow \sigma_T(\rho(\xi_{\ell})) \big\}_{\ell=1}^N,
\end{equation}
where $\rho(\xi_\ell)$ is the ground state of the Hamiltonian $H(\xi_\ell)$ (and $\xi_\ell = (m_\ell,x_\ell)$).
In this most general case, we can now simply train a distinct ML model for each $m \in [m_{\mathrm{min}}, m_{\mathrm{max}}]$. Using this direct method, we only need a training data size $N$ that is $(m_{\mathrm{max}} - m_{\mathrm{min}} + 1)$ times larger than the training data size when $m$ is fixed.

\subsubsection{Systems with long-range interactions}

For simplicity, the proof for our main theorem (Theorem~\ref{thm:mainFourier}) focuses on Hamiltonians that can be written as a sum of geometrically local terms,
\begin{equation}
    H(x) = \sum_{j} h_j(x),
\end{equation}
where $h_j(x)$ acts on a constant number of constituents that are contained in a ball of constant size in a finite-dimensional space.
Our proof can be generalized to some physical systems where $h_j(x)$ acts on constituents that are geometrically non-local.
The main condition we must impose is that the evolution under the Hamiltonian $H(x)$ in the ground state $\rho(x)$ has a bounded speed of information spreading. In the study of quantum many-body systems~\cite{chen2019finite, kuwahara2020strictly, tran2020hierarchy} this assumption is described as a \emph{linear light cone}, meaning that if a perturbation is applied at a point $P$ at time zero, then the effects of that perturbation at a later time $t$ are mostly confined to a region centered at $P$ with radius $vt$; here $v>0$ is called the Lieb-Robinson velocity.

To be more precise, consider two few-body operators, $O_A$ acting on a set of constituents $A$, and $O_B$ acting on a set of constituents $B$; the sets $A$ and $B$ need not be geometrically local. We denote by $d(O_A,O_B)$ the minimum Euclidean distance between constituents in $A$ and constituents in $B$. Recall that in the Heisenberg picture, operators evolve according to $O(t) = \mathrm{e}^{itH(x)}O \mathrm{e}^{-itH(x)}$, where $H(x)$ is the Hamiltonian. We require that the expectation value in the ground state $\rho(x)$ of the commutator of $O_A$ with $O_B(t)$ is highly suppressed when $d(O_A,O_B)$ is small compared to $vt$, i.e.,
\begin{equation} \label{eq:long-range-decay}
    \left| \Tr\left( \left[O_A,  O_B(t) \right] \rho(x) \right) \right| \leq \frac{c |t|^\beta}{\max(0, d(O_A, O_B) - v |t|)^\alpha} \norm{O_A}_\infty \norm{O_B}_\infty,
\end{equation}
where $c$ is a constant, 
and $\alpha > \beta > 0$ are constants that determine the decay, 

Such Lieb-Robinson bounds were proven for geometrically local Hamiltonians decades ago, but linear light cones in physical systems with non-local interactions had not been studied until comparatively recently \cite{chen2019finite, kuwahara2020strictly, tran2020hierarchy}.
It has now been established that, for many long-range interacting systems, Eq.~\eqref{eq:long-range-decay} applies, where $\alpha$ is sufficiently large compared to $\beta$ for our arguments to apply.
Specifically, in the proof given in Appendix~\ref{sec:proofthmGSUPP}, we can replace Eq.~\eqref{eq:uppslope} by 
\begin{align} \label{eq:uppslope-alt}
    |\Tr([O, D_{\uvec}(x)] \rho(x))| & \leq \sum_i \int_{-\infty}^{\infty}  W_\gamma(t) \sum_j \left| \Tr\left( \left[O_i, \mathrm{e}^{\mathrm{i} t H(x)} \frac{\partial h_j}{\partial \uvec}(x) \mathrm{e}^{-\mathrm{i} t H(x)}\right] \rho(x) \right) \right| \mathrm{d}t,
\end{align}
and also replace the Lieb-Robinson bound in Eq.~\eqref{eq:LRbound} by the bound in Eq.~\eqref{eq:long-range-decay}.
When $\alpha$ is sufficiently large compared to $\beta$ in Eq.~\eqref{eq:long-range-decay},
we can guarantee that the right hand side of Eq.~\eqref{eq:uppslope-alt} is upper bounded by
\begin{equation}
    \mbox{const} \times \sum_i \norm{O_i}_\infty,
\end{equation}
using an analysis similar to that given in Appendix~\ref{sub:smoothness}.
After establishing such an upper bound on $|\Tr([O, D_{\uvec}(x)] \rho(x))|$, we can follow exactly the same proof given in the other sections in Appendix~\ref{sec:proofthmGSUPP} to show that the classical ML model can accurately predict the classical representation of the ground state for long-range interacting systems with a similar guarantee as Theorem~\ref{thm:mainFourier}, assuming that the Lieb-Robinson velocity $v$ is bounded above by a constant.

\subsubsection{Fermionic systems}

We can also generalize the proof of Theorem~\ref{thm:mainFourier} to fermionic systems, such as those arising in studies of electronic structure; see for example \cite{helgaker2014molecular}.
We consider second quantization, also known as the occupation number representation, and use the abstract Fock space to represent the Hamiltonians of fermionic systems.
Given a system of $n$ spin orbitals, the Fock space is a $2^n$-dimensional space spanned by $\ket{c_0,c_1,\ldots, c_{n-1}}$, where $c_j=1$ indicates that mode $j$ is occupied and $c_j=0$ indicates that mode $j$ is unoccupied.
A vector in the Fock space is a linear combination of 
these $2^n$ basis states.
Given a mode $j \in \{1, \ldots, n\}$, a fermionic \emph{creation operator} $A_j$ is defined by
\begin{equation} \label{eq:createrule}
\begin{aligned}
	A_j^\dagger\ket{c_0,c_1,\ldots,0_j,\ldots, c_{n-1}}&=(-1)^{\sum_{k=0}^{j-1}c_k}\ket{c_0,c_1,\ldots,1_j,\ldots, c_{n-1}},\\
	A_j^\dagger\ket{c_0,c_1,\ldots,1_j,\ldots, c_{n-1}}&=0,\\
\end{aligned}
\end{equation}
whereas the fermionic \emph{annihilation operator} $A_j$ is defined by
\begin{equation} \label{eq:annihirule}
\begin{aligned}
	A_j\ket{c_0,c_1,\ldots,0_j,\ldots, c_{n-1}}&=0,\\
	A_j\ket{c_0,c_1,\ldots,1_j,\ldots, c_{n-1}}&=(-1)^{\sum_{k=0}^{j-1}c_k}\ket{c_0,c_1,\ldots,0_j,\ldots, c_{n-1}}.\\
\end{aligned}
\end{equation}
For a fermionic system, each local term $h_j(x)$ in the Hamiltonian $H(x) = \sum_j h_j(x)$ is a Hermitian matrix that can be expressed as a product of an
even number of fermionic creation and annihilation operators; we refer to such a Hermitian matrix as an \emph{even fermionic observable}.
For example, we could have $h_{pqrs}(x) = U_{pqrs}(x) A_{p}^\dagger A_{q}^\dagger A_{r} A_{s} + \overline{U_{pqrs}}(x) A_{s}^\dagger A_{r}^\dagger A_{q} A_{p}$, where $U_{pqrs}(x)$ is a complex-valued number. 
(This particular term conserves the total fermion number, but fermion number conservation is not actually required for our arguments to work.)
Two \emph{even fermionic observables} acting on disjoint sets of spin orbitals commute with one another, just as two local observables acting on disjoint sets of qubits commute.
As a result, several results in qubit systems based on the commutation relations of 
disjoint local observables can be easily generalized to \emph{even fermionic observables} in fermionic systems. In particular,
one can generalize the proof of Theorem~\ref{thm:mainFourier} as follows.
\begin{itemize}
    \item First we construct a classical shadow representation for fermionic systems.
    An efficient approach for constructing such a representation is given in \cite{zhao2021fermionic}. This work rigorously analyzes how to predict a large number of properties using outcomes of measurements performed after randomized fermionic Gaussian unitaries.
    We can 
    replace the classical shadow based on randomized Pauli measurements with the fermionic partial tomography introduced in \cite{zhao2021fermionic}.
    \item Secondly we establish a bounded speed of information spreading under evolution governed by $H(x)$ in the ground state $\rho(x)$.
    Intuitively, we would like the ``diameter'' of the support (by ``support'' we mean the set of spin orbitals that an observable acts on substantially) of an \emph{even fermionic observable} under Heisenberg evolution to grow at most linearly in time.
    As for qubit systems, this growth rate is known as the Lieb-Robinson velocity.
    Because two even fermionic observables acting on disjoint sets of spin orbitals commute with one another, one can establish an upper bound on the Lieb-Robinson velocity in fermionic systems by following the argument used for qubit systems \cite{bru2016lieb, nachtergaele2018lieb}. This argument does not work for arbitrary fermionic systems, but it does work if the interaction graph of the spin orbitals is suitably sparse.
\end{itemize}
After these replacements, the rest of the proof follows immediately, yielding a version of Theorem~\ref{thm:mainFourier} for fermionic systems.
As we noted, the argument used to bound the Lieb-Robinson velocity does not work for some fermionic systems; for example it fails in models where orbitals have all-to-all connectivity without any geometrical constraints (the same is true for qubit systems).
But the proof of Theorem~\ref{thm:mainFourier} does go through for tight-binding models, including the Fermi-Hubbard model. Since computing ground state properties of the Fermi-Hubbard model is notoriously difficult for classical computers, it is encouraging to find that our classical ML algorithm can compute these properties efficiently when provided with polynomial-size training data.

\color{black}

\section{Proof of efficiency for predicting ground states} \label{sec:proofthmGSUPP}

This section contains a detailed proof for one of our main contributions. Namely, a rigorous performance guarantee for learning to predict ground state representations.

\begin{theorem}[Theorem~\ref{thm:mainFourier}, detailed restatement] \label{thm:detailedFourier}
Consider any family of $n$-qubit geometrically-local Hamiltonians $\{H(x):\; x \in [-1, 1]^m\}$ in a finite spatial dimension, such that each local term in $H(x)$ depends smoothly on $x$, and the smallest eigenvalue and the next smallest eigenvalues have a constant gap $\gamma \geq \Omega(1)$ between them.
Let $\rho(x)$ be the ground state of $H(x)$, that is
\begin{align}
    \rho(x) =& \lim_{\beta \rightarrow \infty} e^{- \beta H(x)} / \Tr(e^{- \beta H(x)}) \in \left(\mathbb{H}_2\right)^{\otimes n} &\text{(ground state of Hamiltonian $H(x)$)}
\end{align}
where $\mathbb{H}_2$ is the vector space of $2 \times 2$ Hermitian matrices.
Suppose that we are interested in learning to predict a sum $O=\sum_{i=1}^L O_i$ of $L$ local observables that satisfies
$\sum_{i=1}^L \norm{O_i} \leq B$ (bounded norm).
Then, classical shadow data $\{ x_{\ell} \rightarrow \sigma_1(\rho(x_{\ell})) \}_{\ell = 1}^{N}$, with $x_\ell \sim \mathrm{Unif}[-1,1]^m$ and 
\begin{align}
    N &= B^2 m^{\mathcal{O}(B^2 / \epsilon)} & \text{(training data size)},
\end{align}
suffices to produce a ground state prediction model
\begin{equation}
    \hat{\sigma}_N(x) = \frac{1}{N} \sum_{\ell = 1}^N \kappa(x, x_{\ell})  \rho(x_{\ell}) \,\,\, \mbox{with} \,\,\, \kappa(x, x_{\ell}) = \sum_{k \in \mathbb{Z}^m, \norm{k}_2 \leq \Lambda} \cos(\pi k \cdot (x - x_{\ell})) \in \mathbb{R},
\end{equation}
that achieves
\begin{equation}
\E_{x \sim [-1, 1]^m} |\Tr(O \hat{\sigma}_N(x)) - \Tr(O \rho(x))|^2 \leq \epsilon,    
\end{equation}
with high probability.
The classical training time for constructing $\hat{\sigma}_N(x)$ and the prediction time for computing $\Tr(O \hat{\sigma}(x))$ are both upper bounded by $\mathcal{O}((n + L) B^2 m^{\mathcal{O}(B^2 / \epsilon)}).$
\end{theorem}

Theorem~\ref{thm:detailedFourier} can be generalized to the following statement about learning a family of quantum states. In particular, we will prove the following theorem and use it to derive Theorem~\ref{thm:detailedFourier}.

\begin{theorem} \label{thm:stateFourier}
Consider a parametrized family of $n$-qubit states $\{\rho(x):\; x \in [-1, 1]^m\}$
and a sum $O=\sum_{i=1}^L O_i$ of $L$ local observables that obey
\begin{subequations}
\begin{align}
\E_{x \sim [-1, 1]^m} \norm{\nabla_x \Tr(O \rho(x))}_2^2 \leq C &\quad \text{(smoothness condition)}, \\
\sum_i \norm{O_i} \leq B & \quad \text{(bounded norm)}.
\end{align}
\end{subequations}
Then, classical shadow data $\{ x_{\ell} \rightarrow \sigma_1(\rho(x_{\ell})) \}_{\ell = 1}^{N}$, with $x_\ell \sim \mathrm{Unif}[-1,1]^m$ and 
\begin{align}
    N &= B^2 m^{\mathcal{O}(C / \epsilon)} & \text{(training data size)},
\end{align}
suffices to produce a state prediction model
we can learn from classical data $\{ x_{\ell} \rightarrow \sigma_1(\rho(x_{\ell})) \}_{\ell = 1}^{N}$ to produce a model
\begin{equation}
    \hat{\sigma}_N(x) = \frac{1}{N} \sum_{\ell = 1}^N \kappa(x, x_{\ell}) \Tr(O \rho(x_{\ell})) \,\,\, \mbox{with} \,\,\, \kappa(x, x_{\ell}) = \sum_{k \in \mathbb{Z}^m, \norm{k}_2 \leq \Lambda} \cos(\pi k \cdot (x - x_{\ell})) \in \mathbb{R},
\end{equation}
that achieves
\begin{equation}
\E_{x \sim [-1, 1]^m} |\Tr(O \hat{\sigma}_N(x)) - \Tr(O \rho(x))|^2 \leq \epsilon,    
\end{equation}
with high probability.
The classical training time for constructing $\hat{\sigma}_N(x)$ and the prediction time for computing $\Tr(O \hat{\sigma}(x))$ are both upper bounded by $\mathcal{O}((n + L) B^2 m^{\mathcal{O}(C/\epsilon)}).$
\end{theorem}

The following sections are structured as follows. 
In Section~\ref{sec:overviewSCupGS}, we provide an overview to illustrate the proof of the sample complexity upper bound.
The first step, given in Section~\ref{sec:truncGS}, bounds the truncation error when approximating the quantum state function $\rho(x)$ using a truncated Fourier series.
The second step, given in Section~\ref{sec:MLTrun}, bounds the generalization error for learning the Fourier approximation to the quantum state function $\rho(x)$.
Then, in Section~\ref{sec:runtimeGS}, we analyze the training and prediction time of the proposed classical machine learning model.
These three sections establish Theorem~\ref{thm:stateFourier}.
Finally, in Section~\ref{sub:smoothness}, we use Theorem~\ref{thm:stateFourier} and nice properties about ground states of Hamiltonians to prove Theorem~\ref{thm:detailedFourier}.

\subsection{Overview for sample complexity upper bound} 
\label{sec:overviewSCupGS}

The key intermediate step is to construct a truncated Fourier series of the quantum state function $\rho(x)$.
The Fourier series of the matrix-valued function $\rho(x)$ is given as
\begin{equation}
    \rho(x) = \sum_{k \in \mathbb{Z}^m} e^{\mathrm{i} \pi k \cdot x} A_{k},
    \label{eq:fourier-basis}
\end{equation}
where $A_k$ are matrix-valued Fourier coefficients
\begin{equation}
    A_k = \frac{1}{2^m} \int_{[-1, 1]^m} e^{-\mathrm{i} \pi k \cdot x} \rho(x) \mathrm{d}^m x.
\end{equation}
We define the truncated Fourier series as
\begin{equation}
    \rho_\Lambda(x) = \sum_{k \in \mathbb{Z}^m, \norm{k}_2 \leq \Lambda} e^{\mathrm{i} \pi k \cdot x} A_{k},
    \label{eq:truncated-series}
\end{equation}
where $\Lambda >0$ is a pre-specified cutoff value.
Given an observable $O$ that can be written as a sum of local observables $O = \sum_i O_i$ with $\sum_i \norm{O_i}_\infty \leq B$ and $\E_{x \sim [-1, 1]^m} \norm{\nabla_x \Tr(O \rho(x))}_2^2 \leq C$, the proof of Theorem~\ref{thm:detailedFourier} consists of two parts.

First, we bound the error between the truncated Fourier series $\rho_\Lambda(x)$ and the true quantum state function $\rho(x)$ in Section~\ref{sec:truncGS} giving
\begin{equation} \label{eq:errrhoandrhoL}
    \E_{x \sim [-1, 1]^m} \left|\Tr(O \rho(x)) - \Tr\left(O \rho_\Lambda(x)\right) \right|^2 \leq \mathcal{O}\left( \frac{C}{\Lambda^2}\right),
\end{equation}
We choose the truncation $\Lambda = \Theta(\sqrt{C / \epsilon})$ such that the error between truncated Fourier series and the true quantum state function obeys
\begin{equation}
\E_{x \sim [-1, 1]^m} \left|\Tr(O \rho(x)) - \Tr\left(O \rho_\Lambda(x)\right) \right|^2 \leq \frac{\epsilon}{4}.
\end{equation}
In the second part, we bound the error between the machine learning model $\hat{\sigma}(x)$ and the truncated Fourier series $\rho_\Lambda(x)$ in Section~\ref{sec:MLTrun}.
With high probability over the randomness in generating the training data, we have
\begin{equation} \label{eq:errgandrhoL}
    \E_{x \sim [-1, 1]^m} \left|\Tr(O \hat{\sigma}(x)) - \Tr\left(O \rho_\Lambda(x)\right) \right|^2 \leq \frac{B^2 m^{\mathcal{O}(\Lambda^2)}}{N}.
\end{equation}
The training data contains two sources of randomness, one from the sampling of $x_\ell$ and the other from the local randomized measurement to construct approximate classical representation for $\rho(x_\ell)$ that could be feed into the classical machine learning model.
We choose the training data size
\begin{equation}
    N = \frac{2 B^2 m^{\mathcal{O}(C / \epsilon)}}{\epsilon} \leq B^2 m^{\mathcal{O}(C / \epsilon) + \log(1/\epsilon) + 1} = B^2 m^{\mathcal{O}(C / \epsilon)},
\end{equation}
such that the error between the machine learning model and the truncated Fourier series obeys
\begin{equation}
    \E_{x \sim [-1, 1]^m} \left|\Tr(O \hat{\sigma}(x)) - \Tr\left(O \rho_\Lambda(x)\right) \right|^2 \leq \epsilon / 4,
\end{equation}
with high probability.
The two parts can be combined by a triangle inequality to yield
\begin{subequations}
\begin{align}
    & \E_{x \sim [-1, 1]^m} \left|\Tr(O \hat{\sigma}(x)) - \Tr\left(O \rho(x)\right) \right|^2\\
    &\leq \left( \sqrt{\E_{x \sim [-1, 1]^m} \left|\Tr(O \hat{\sigma}(x)) - \Tr\left(O \rho_\Lambda(x)\right) \right|^2} + \sqrt{\E_{x \sim [-1, 1]^m} \left|\Tr(O \rho(x)) - \Tr\left(O \rho_\Lambda(x)\right) \right|^2} \right)^2 = \epsilon,
\end{align}
\end{subequations}
with high probability over the randomness in the training data.
This establishes the sample complexity upper bound for Theorem~\ref{thm:detailedFourier}.

When the Hamiltonians $H(x)$ have  spectral gap $\geq \Omega(1)$ in the domain $x \in [-1, 1]^m$, for any observable $O = \sum_i O_i$ that can be written as a sum of local observables with $\sum_i \norm{O_i}_\infty \leq B$, we have
\begin{equation}
    \E_{x \sim [-1, 1]^m} \norm{\nabla_x \Tr(O \rho(x))}_2^2 \leq \mathcal{O}(B^2).
\end{equation}
Hence, we can prove the sample complexity upper bound in Theorem~\ref{thm:stateFourier} by utilizing Theorem~\ref{thm:detailedFourier} and the fact that $C = \mathcal{O}(B^2)$.

\subsection{Controlling the truncation error}
\label{sec:truncGS}

For a fixed observable $O$, we can define a function
\begin{equation}
    f(x) = \Tr(O \rho(x)) = \sum_{k \in \mathbb{Z}^m} e^{\mathrm{i} \pi k \cdot x} \Tr( O A_{k} ).
\end{equation}
And the truncated Fourier series of the function $f(x)$ is given by
\begin{equation}
    f_\Lambda(x) = \Tr(O \rho_\Lambda(x)) = \rho_\Lambda(x) = \sum_{k \in \mathbb{Z}^m, \norm{k}_2 \leq \Lambda} e^{\mathrm{i} \pi k \cdot x} \Tr( O A_{k} ).
\end{equation}

\begin{lemma}[truncation error] \label{lem:truncation-error}
Let $f(x) = \sum_{k \in \mathbb{Z}^m} \alpha_k e^{\mathrm{i} \pi k \cdot x}$ and $f_\Lambda(x) = \sum_{k \in \mathbb{Z}^m, \norm{k}_2 \leq \Lambda} \alpha_k e^{\mathrm{i} \pi k \cdot x}$. Then
\begin{equation}
\E_{x \sim [-1,1]^m} 
| f(x) - f_\Lambda(x) |^2 \leq \frac{1}{\pi^2 \Lambda^2} 
\E_{x \sim [-1,1]^m} \norm{\nabla_x f(x)}_2^2
\quad \text{for any cutoff} \quad  \Lambda >0.
\end{equation}
\end{lemma}
\begin{proof}
The claim follows from standard Harmonic analysis arguments. More precisely, we combine \emph{orthogonality} (
$
\int_{[-1,1]^m} \mathrm{e}^{\mathrm{i}(\pi (k-k') x} \mathrm{d}^m x = \delta_{(k,k')}$) with the fact that the Fourier transform exchanges differentials (``momentum'') with multiplications (``position''):
\begin{align}
\nabla_x f(x) = \sum_{k \in \mathbb{Z}^m} \alpha_k \nabla_x \mathrm{e}^{\mathrm{i} \pi k x}
= \mathrm{i} \pi \sum_{k \in \mathbb{Z}^m} \alpha_k k \mathrm{e}^{\mathrm{i} \pi k x}. \label{eq:fourier-property}
\end{align}
Use orthogonality to rewrite the truncation error as
\begin{subequations}
\begin{align}
\E_{x \sim [-1,1]^m}
\left| f(x) - f_\Lambda(x) \right|^2 =& \int_{[-1,1]^m} \Big| \sum_{k \in \mathbb{Z}^m: \norm{k}>\Lambda} \mathrm{e}^{\mathrm{i} \pi kx} \alpha_k \Big|^2 \mathrm{d}^m x \\
=& \sum_{k:\norm{k}_2>\Lambda} \sum_{k':\norm{k'}_2>\Lambda}
\Big( \int_{[-1,1]^m} \mathrm{e}^{\mathrm{i} \pi (k-k')x} \mathrm{d}^m x \Big)  \overline{\alpha_k} \alpha_k \\
=& \sum_{k: \norm{k}_2 >\Lambda} \left| \alpha_k \right|^2. \label{eq:truncation-error-aux1}
\end{align}
\end{subequations}
Conversely, we use orthogonality and Rel.~\eqref{eq:fourier-property} to rephrase this upper bound. 
Let $\langle k', k \rangle$ be the Euclidean inner product between two vectors $k,k' \in \mathbb{Z}^m$. Then,
\begin{subequations}
\begin{align}
\E_{x \sim [-1,1]^m}\norm{ \nabla_x f(x)}_2^2 =& \int_{[-1,1]^m} \norm{ \sum_{k \in \mathbb{Z}^m} \pi k \mathrm{e}^{\mathrm{i}\pi kx} \alpha_k}_2^2 \mathrm{d}^m x \\
=& \sum_{k,k' \in \mathbb{Z}^m} \pi^2 \langle k',k \rangle
\int_{[-1,1]^m} \mathrm{e}^{\mathrm{i} \pi (k-k')x} \mathrm{d}^m x \overline{\alpha_{k'}} \alpha_k  \\
=& \pi^2 \sum_{k \in \mathbb{Z}^m} \langle k,k \rangle |\alpha_k|^2 = \pi^2 \sum_{k \in \mathbb{Z}^m} \norm{k}_2^2 \left| \alpha_k \right|^2.
\end{align}
\end{subequations}
In words, the upper bound from Eq.~\eqref{eq:truncation-error-aux1} can be rephrased as the Euclidean norm 
$\norm{ \nabla_x f(x)}_2^2$ of the vector $\nabla_x f(x)$.
The advertised claim readily follows from comparing these two reformulations:
\begin{align}
\sum_{k:\norm{k}_2>\Lambda} \left| \alpha_k \right|^2
\leq \frac{1}{\Lambda^2} \sum_{k:\norm{k}_2>\Lambda} \norm{k}_2^2 \left| \alpha_k \right|^2
\leq \frac{1}{\pi^2 \Lambda^2}\Big( \pi^2 \sum_{k \in \mathbb{Z}^m} \norm{k}_2^2 
\left|\alpha_k \right|^2 \Big).
\end{align}
\end{proof}

Using Lemma~\ref{lem:truncation-error} and the condition that $\E_{x \sim [-1,1]^m} \norm{\nabla_x \Tr(O \rho(x))}_2^2 \leq C$, we can obtain the desired inequality for bounding the truncation error,
\begin{equation}
    \E_{x \sim [-1, 1]^m} \left|\Tr(O \rho(x)) - \Tr\left(O \rho_\Lambda(x)\right) \right|^2 \leq \mathcal{O}\left( \frac{C}{\Lambda^2}\right).
\end{equation}

\subsection{
Controlling generalization errors from using the training data
}\label{sec:MLTrun}

This section is devoted to a practical issue regarding training data based on classical shadows. Each label is obtained by performing a single-shot quantum measurement of a parametrized quantum state $\rho (x_i)$. We can use Eq.~\eqref{eq:classical-shadow} to convert the single-shot outcome into
$\sigma_1(\rho) = \bigotimes_{i=1}^{n} \left( 3 \ketbra{s_i}{s_i} - \mathbb{I} \right)$. Such a classical shadow approximation reproduces the underlying state in expectation, i.e., $\E _{s_1, \ldots, s_n} [\sigma_1(\rho)] = \rho$.
Recall that the training data $\mathcal{T} = \left\{ x_{\ell} \rightarrow \sigma_1(\rho (x_{\ell})) \right\}_{\ell = 1}^{N}$ consists of such classical shadow approximations. The machine learning model makes predictions based on a truncated Fourier kernel for future predictions. For new input $x \in \left[-1,1\right]^n$, we predict
\begin{subequations}
\label{eq:sigma-hat}
\begin{align}
\hat{\sigma} (x) &= \frac{1}{N} \sum_{\ell = 1}^{N}\kappa (x, x_{\ell}) \sigma_1 \left( \rho (x_{\ell}) \right) \quad \text{with} \quad \\
\kappa(x, x_{\ell}) &= \sum_{k \in \mathbb{Z}^m, \norm{k}_2 \leq \Lambda} \mathrm{e}^{\mathrm{i} \pi k \cdot (x - x_{\ell})} = \sum_{k \in \mathbb{Z}^m, \norm{k}_2 \leq \Lambda} \cos(\pi k \cdot (x - x_{\ell})).
\end{align}
\end{subequations}
In the following, we will show that machine learning model $\hat{\sigma} (x)$ is equal to the truncated Fourier series $\rho_\Lambda (x)$ of the true target quantum state if we take the expectation over the training data, which includes the sampled inputs $x_1, \ldots, x_N$ and the randomized measurement outcomes $S_1(\rho(x_{\ell})) = \{s_i\}_{i=1}^n$ for each input $x_\ell$.
Moreover, statistical flucutations due to shot noise will be small provided that we are interested in predicting an observable that decomposes nicely as a sum of local terms. These observations are the content of the following statement.

\begin{lemma}[Statistical properties of the predicted quantum state $\hat{\sigma}(x)$] \label{lem:groundstateprop}
Let  $\mathcal{T} = \{ x_\ell \rightarrow \sigma_1(\rho(x_\ell)) \}_{\ell=1}^{N}$ be a training set featuring uniformly random inputs $x_\ell \overset{\textit{unif}}{\sim} [-1,1]^m$ and classical shadows of the associated quantum states as labels.
Then, the machine learning model obeys
\begin{equation}
\E_\mathcal{T}[\hat{\sigma}(x)] = \rho_\Lambda(x) = \sum_{k \in \mathbb{Z}^m, \norm{k}_2 \leq \Lambda} e^{\mathrm{i} \pi k \cdot x} A_{k}.    
\end{equation}
Moreover, suppose that an observable $O=\sum_i O_i$ decomposes into a sum of $q$-local terms. Then, with probability at least $1-\delta$, we have
\begin{equation}
    \E_{x \sim [-1, 1]^m} \left|\Tr(O \hat{\sigma}(x)) - \Tr\left(O \rho_\Lambda(x)\right) \right|^2 \leq \frac{1}{N}9^q \big(\sum_i \norm{O_i}_\infty\big)^2 (2m+1)^{\Lambda^2} \left( \Lambda^2 \log(2m + 1) + \log\left(4/\delta\right) \right).
\end{equation}
\end{lemma}
The advertised bound can be further streamlined if the observable locality $q$ and confidence level $\delta$ are constant. Assuming $q, \delta = \mathcal{O}(1)$ ensures the following simplified scaling:
\begin{equation}
\E_{x \sim [-1, 1]^m} \left|\Tr(O \hat{\sigma}(x)) - \Tr\left(O \rho_\Lambda(x)\right) \right|^2
= \mathcal{O} \left( \frac{1}{N} \left( \sum_i \norm{O_i} \right)^2 (2m + 1)^{\Lambda^2 + \log (\Lambda^2) +1} \right) = \frac{\left(\sum_i \norm{O_i}\right)^2 m^{\mathcal{O}(\Lambda^2)}}{N}.
\end{equation}
Using the condition that $\sum_i \norm{O_i} \leq B$, we have
\begin{equation}
\E_{x \sim [-1, 1]^m} \left|\Tr(O \hat{\sigma}(x)) - \Tr\left(O \rho_\Lambda(x)\right) \right|^2
 = \frac{B^2 m^{\mathcal{O}(\Lambda^2)}}{N},
\end{equation}
which controls the generalization error from quantum measurements.
The argument is based on fundamental properties of classical shadows that have been reviewed in Appendix~\ref{sec:classical-shadows}.

\begin{proof}[Proof of Lemma~\ref{lem:groundstateprop}]
We begin by condensing notation somewhat. Here, we only consider classical shadows of size $T=1$. Hence, we may replace the superscript $(t)$ by $(x_\ell)$ to succinctly keep track of classical input parameters.
More precisely, we let
$ |s_i^{(x_\ell)}\rangle$ be the randomized Pauli measurement outcome for the $i$-th qubit when measuring the quantum state $\rho (x_\ell)$.
The training data $\mathcal{T} = \{ x_\ell \rightarrow \sigma_1(\rho(x_\ell)) \}_{\ell=1}^{N}$ is determined by the following random variables
\begin{subequations}
\label{eq:training-data}
\begin{align}
    x_\ell \in [-1,1]^m,& & \text{for $\ell \in \{1,\ldots,N\}$}, \label{eq:training-xell} \\
    s^{(x_\ell)}_i \in \left\{\ket{0},\ket{1},\ket{+},\ket{-},\ket{\mathrm{i}+},\ket{\mathrm{i}-}\right\},& & \text{for $i \in \{1,\ldots,n\}$ and $\ell \in \{1,\ldots,N\}$}. \label{eq:training-sell}
\end{align}
\end{subequations}
The first claim is an immediate consequence of Eq.~\eqref{eq:snapshot-expectation}:
\begin{subequations}
\begin{align}
    \E_\mathcal{T}[\hat{\sigma}(x)] &= \frac{1}{N} \sum_{\ell=1}^{N} \E_{x_\ell\sim[-1,1]^{m}}\left[\kappa(x, x_\ell) \E_{s^{(x_\ell)}_1, \ldots, s^{(x_\ell)}_n}[\sigma_1(\rho(x_\ell))]\right] \\
    &= \frac{1}{N} \sum_{\ell=1}^{N} \E_{x_\ell\sim[-1,1]^{m}}\left[\kappa(x, x_\ell) \rho(x_\ell) \right] \\
    &= \E_{x_1\sim[-1,1]^{m}}\left[\kappa(x, x_1) \rho(x_1) \right] \\
    &= \sum_{k \in \mathbb{Z}^m, \norm{k}_2 \leq \Lambda} e^{\mathrm{i} \pi k \cdot x} \E_{x_1\sim[-1,1]^{m}}\left[ e^{-\mathrm{i} \pi k \cdot x_1} \rho(x_1)\right] \\
    &= \sum_{k \in \mathbb{Z}^m, \norm{k}_2 \leq \Lambda} e^{\mathrm{i} \pi k \cdot x} \frac{1}{2^m} \int_{[-1, 1]^m} e^{-\mathrm{i} \pi k \cdot x_1} \rho(x) \mathrm{d}^m x_1 \\
    &= \sum_{k \in \mathbb{Z}^m, \norm{k}_2 \leq \Lambda} e^{\mathrm{i} \pi k \cdot x} A_k \\
    &= \rho_\Lambda (x).
\end{align}
\end{subequations}
Here, we have also used the fact that each $x_\ell$ is sampled independently and uniformly from $[-1, 1]^m$. 

The second result is contingent on the training data for predicting the ground state representation $\mathcal{T} = \{ x_\ell \rightarrow \sigma_1(\rho(x_\ell)) \}_{\ell=1}^{N}$. We begin with using the definitions of $\hat{\sigma} (x)$ (\ref{eq:sigma-hat}) and $\rho_\Lambda (x)$ to rewrite the expression of interest as
\begin{subequations}
\label{eq:shot-noise-concentration-aux1}
\begin{align}
  &\!\!\!\!\!\!\!\!\!\!\!\!\!\!\!\!\!\!\E_{x\sim[-1,1]^{m}}\left|\Tr(O\hat{\sigma}(x))-\Tr\left(O\rho_{\Lambda}(x)\right)\right|^{2}\nonumber\\=&\frac{1}{2^{m}}\int_{[-1,1]^{m}}\mathrm{d}^{m}x\left|\sum_{k\in\mathbb{Z}^{m},\norm k_{2}\leq\Lambda}\mathrm{e}^{\mathrm{i}\pi k\cdot x}\left(\frac{1}{N}\sum_{\ell=1}^{N}\mathrm{e}^{-\mathrm{i}\pi k\cdot x_{\ell}}\Tr\left(O\sigma_{1}\left(\rho(x_{\ell})\right)\right)-\Tr\left(OA_{k}\right)\right)\right|^{2}\\=&\sum_{k\in\mathbb{Z}^{m},\norm k_{2}\leq\Lambda}\left|\frac{1}{N}\sum_{\ell=1}^{N}\mathrm{e}^{-\mathrm{i}\pi k\cdot x_{\ell}}\Tr\left(O\sigma_{1}\left(\rho(x_{\ell})\right)\right)-\Tr\left(OA_{k}\right)\right|^{2}~,\\
  \equiv&\sum_{k\in\mathbb{Z}^{m},\norm k_{2}\leq\Lambda}D_{(k)}({\cal T})^{2}~,\label{eq:shot-noise-concentration-aux1.1}
\end{align}
\end{subequations}
where we have evaluated the Fourier integral over $x$ and introduced shorthand notation $D_{(k)}({\cal T})^{2}$ for each summand. 

The next key step is to notice that each $A_k$ is an expectation value over both the parameters and the shadows. Writing out $A_k$ and expressing $\rho$ in terms of shadows using Eq.~(\ref{eq:snapshot-expectation}),
\begin{subequations}
\begin{align}
  \Tr \left( O A_{k} \right)&=\frac{1}{2^{m}}\int_{[-1,1]^{m}}\mathrm{e}^{-\mathrm{i}\pi k\cdot x_{\ell}}\Tr\left(O\rho(x_{\ell})\right)\mathrm{d}^{m}x_{\ell}\\&=\E_{x_{\ell}\sim[-1,1]^{m}}\mathrm{e}^{-\mathrm{i}\pi k\cdot x_{\ell}}\Tr\left(O\rho(x_{\ell})\right)\\&=\E_{x_{\ell}\text{ and }s_{1}^{(x_{\ell})},\ldots,s_{n}^{(x_{\ell})}}\mathrm{e}^{-\mathrm{i}\pi k\cdot x_{\ell}}\Tr\left(O\sigma_{1}\left(\rho(x_{\ell})\right)\right)~.
\end{align}
\end{subequations}
Plugging this back into the summand in Eq.~(\ref{eq:shot-noise-concentration-aux1.1}) yields
\begin{equation}
    D_{(k)}({\cal T})^{2}=\left|\frac{1}{N}\sum_{\ell=1}^{N}\mathrm{e}^{-\mathrm{i}\pi k\cdot x_{\ell}}\Tr\left(O\sigma_{1}\left(\rho(x_{\ell})\right)\right)-\E_{x_{\ell}\text{ and }s_{1}^{(x_{\ell})},\ldots,s_{n}^{(x_{\ell})}}\mathrm{e}^{-\mathrm{i}\pi k\cdot x_{\ell}}\Tr\left(O\sigma_{1}\left(\rho(x_{\ell})\right)\right)\right|^{2}~.
\end{equation}
Therefore, each $D_{(k)}({\cal T})^{2}$ is the (square-)deviation of an empirical average from the true expectation value $A_k$. Hence, we can use Hoeffding's inequality to bound it, provided that $O$ is local and bounded. This may come as a surprise, as the empirical average samples only different parameters $x_\ell$ and not different shadows $\sigma$. However, the shadows depend on the parameters, so sampling only over the parameters turns out to be sufficient for a reasonable estimate.

In order to apply Hoeffding's inequality, we first have to make sure the expectation value is bounded. Recall that $O=\sum_i O_i$ decomposes nicely into a sum of $q$-body terms. More formally, $\mathrm{supp}(O_j) \subset\{1,\ldots,n\}$ contains at most $q$ qubits.
We also know trace and trace norm of each single-qubit contribution to $\sigma_1 (\rho (x_\ell))$, $\mathrm{tr} \big( 3|s_j^{(x_\ell)} \rangle \! \langle s_j^{(x_\ell)} |- \mathbb{I} \big)=1$, and Eq.~\eqref{eq:shadow-norm} asserts $\big\| 3 |s_j^{(x_\ell)} \rangle \! \langle s_j^{(x_\ell)}|-\mathbb{I} \big\|_1 =3$.
The matrix Hoelder inequality then implies, for every $x_{\ell} \in [-1,1]^m$,
\begin{subequations}
\begin{align}
\left| \mathrm{e}^{\mathrm{i} \pi k \cdot x_{\ell}} \Tr \left( O \sigma_1 \left( \rho (x_{\ell}) \right) \right)
\right|
& \leq  \sum_i \left| \mathrm{tr} \left( O_i \sigma_1 (\rho (x_\ell)) \right) \right| \\
&= \sum_i \left| \mathrm{tr} \left( O_{A_i} \mathrm{tr}_{\neg A_i} \left( \sigma_1 \left( \rho (x_\ell) \right) \right) \right) \right| \\
&\leq  \sum_i \left\| O_{A_i} \right\|_\infty \left\| \mathrm{tr}_{\neg A_i} \left( \sigma_1 \left( \rho (x_\ell) \right) \right) \right\|_1 \\
&= \sum_i \|O_i \|_\infty \prod_{j \in \mathrm{supp}(O_i)} \left\| 3 |s_j^{(x_\ell)}\rangle \! \langle s_j^{(x_\ell)}|-\mathbb{I} \right\|_1 \\
&= \sum_i \|O_i \|_\infty 3^{|\mathrm{supp}(O_j)|} \leq 3^q \sum_i \|O_i\|_\infty.
\end{align}
\end{subequations}
Thus, the expectation value is bounded.

We are now ready to bound the likelihood of a large deviation $D_{(k)}(\mathcal{T})^{2}$. To recap, for each $k \in \mathbb{Z}^m$ obeying $\norm{k}_2 \leq \Lambda$, we face a contribution that collects the (square-)deviation of a sum of \textit{iid} and \emph{bounded} random variables around their expectation value. These variables are complex, but one can analyze their real and imaginary parts separately and collect them into a complex version of Hoeffding's inequality:
\begin{subequations}
\label{eq:shot-noise-concentration-aux2}
\begin{align}
    \mathrm{Pr}\left[D_{(k)}(\mathcal{T})^{2}\geq\tau^{2}\right]&=\mathrm{Pr}\left[D_{(k)}(\mathcal{T})\geq\tau\right]\\&\leq2\exp\left(-\frac{2N\tau^{2}}{9^{q}\left(\sum_{i}\norm{O_{i}}\right)^{2}}\right)\quad\text{for all}\quad\tau>0.
\end{align}
\end{subequations}
This concentration bound connects training data size $N = |\mathcal{T}|$ with the size of a (fixed, but arbitrary) contribution to the expected deviation \eqref{eq:shot-noise-concentration-aux1}. For fixed magnitude $\tau$ and confidence $\delta$, there is always a (finite) training data size $N=N (\tau,\delta)$ that ensures $D_{(k)}(\mathcal{T})^2 \leq \tau$ with probability at least $1-\delta$. 
We can extend this reasoning to the entire sum in Eq.~\eqref{eq:shot-noise-concentration-aux1} by exploiting that Rel.~\eqref{eq:shot-noise-concentration-aux2} 
is independent of $k$, and the summation only ranges over finitely many terms. 
Introduce $K_\Lambda = \left| \left\{k \in \mathbb{Z}^m:\; \norm{k}_2 \leq \Lambda \right\} \right|$ --- the number of wave-vectors $k \in \mathbb{Z}^m$ whose Euclidean norm is bounded by $\Lambda$ --- and apply a union bound to conclude
\begin{subequations}
\begin{align}
\mathrm{Pr} \left[ \sum_{k \in \mathbb{Z}^m,\norm{k}_2 \leq \Lambda} D_{(k)} (\mathcal{T})^2 \geq K_\Lambda \tau^2 \right]
&\leq \mathrm{Pr} \left[ \exists k \in \mathbb{Z}^m: \norm{k}_2 \leq \Lambda, \,\, \mathrm{s.t.} \,\, D_{(k)} (\mathcal{T})^2 \geq \tau^2 \right] \\
&\leq \sum_{k \in \mathbb{Z}^m,\norm{k}_2 \leq \Lambda} \mathrm{Pr} \left[ D_{(k)}(\mathcal{T})^2 \geq \tau^2 \right] \\
&\leq 2 K_\Lambda \exp \left( - \frac{2 N \tau^2}{9^q \left(\sum_i \norm{O_i}\right)^2} \right).
\end{align}
\end{subequations}
for all $\tau >0$.
To finish the argument, we take guidance from Eq.~\eqref{eq:shot-noise-concentration-aux2}. Fix a confidence level $\delta \in (0,1)$ and set 
\begin{equation}
\tau^2 = \frac{1}{N}9^q \left( \sum_i \norm{O_i}\right)^2 \log (2 K_\Lambda/\delta)
\quad \text{to ensure} \quad \mathrm{Pr} \left[  \E_{x \sim [-1, 1]^m} \left|\Tr(O \hat{\sigma}(x)) - \Tr\left(O \rho_\Lambda(x)\right) \right|^2 \geq K_\Lambda \tau^2 \right]
\leq \delta.
\end{equation}
The advertised bound follows from inserting an explicit bound on the number of relevant wavevectors: 
\begin{equation}\label{eq:numklambda}
K_\Lambda = \left| \left\{ k \in \mathbb{Z}^m:\; \norm{k}_2 \leq \Lambda \right\} \right| \leq (2m+1)^{\Lambda^2}. 
\end{equation} 
To see this, note that $\norm{k}_2^2 = \sum_{i=1}^n |k_i|^2 \geq \sum_{i=1}^n |k_i| = \norm{k}_{1}$, because $k_i \in \mathbb{Z}$. Conversely, every $k \in \mathbb{Z}^m$ that obeys $\norm{k}_2 \leq \Lambda$ also obeys $\norm{k}_1 \leq \Lambda^2$. 
Next, we enumerate all wave-vectors that obey the relaxed condition $\norm{k}_1 \leq \Lambda^2$.
To this end, we consider a simple process: select an index $i \in [m]$,
and update the associated wave number by $+1$ (increment), $0$ (do nothing) or $-1$ (decrement). 
Repeating this process a total of $\Lambda^2$ times allows us to generate no more than $(2m+1)^{\Lambda^2}$ different wavevectors. But, at the same time, every wave vector $k \in \mathbb{Z}^m$ that obeys $\norm{k}_2 \leq \Lambda^2$ can be reached in this fashion. Hence, we conclude $K_\Lambda \leq \left| \left\{ k \in \mathbb{Z}^m:\; \norm{k}_1 \leq \Lambda^2 \right\} \right|\leq (2m+1)^{\Lambda^2}$.
\end{proof}

\subsection{Computational time for training and prediction} \label{sec:runtimeGS}

We have proposed a very simple prediction model that is based on approximating a truncated Fourier series ($l_2$-Dirichlet kernel). The training time is equivalent to loading the training data $\mathcal{T} = \left\{ x_{\ell} \rightarrow \sigma_1(\rho (x_{\ell})) \right\}_{\ell = 1}^{N}$. Only a single snapshot is provided for each sampled parameter $x_\ell$ (i.e., $T=1$), so we relabel $s^{(t)}\to s^{(x_\ell)}$.  The training data is given by the collection of $x_\ell$ and shadows $\{s_i^{x_\ell}\}_{i=1}^n$, following Eq.~(\ref{eq:training-data}). Therefore, one only needs 
\begin{align}
    \mathcal{O}\left((n + m) N\right) =  \mathcal{O}\left( (n+m) B^2 m^{\mathcal{O}(C / \epsilon)} \right)  = \mathcal{O}\left( n B^2 m^{\mathcal{O}(C / \epsilon)} \right)  &\quad \text{(training time)}
\end{align}
computational time to load the relevant data into a classical memory. 
Next, suppose that $O=\sum_{i=1}^L O_i$ is comprised of $L$ $q$-local terms. Then, we can compute the associated expectation value for the predicted quantum state $\hat{\sigma}(x)$ by evaluating
\begin{equation}
    \Tr(O \hat{\sigma}(x)) = \frac{1}{N} \sum_{\ell=1}^{N} \sum_{i=1}^L \kappa(x, x_{\ell}) \Tr(O_i \sigma_1(\rho(x_{\ell}))).
 \end{equation}
Recall that the kernel function is defined as
\begin{align}
\kappa(x, x_{\ell}) = \sum_{k \in \mathbb{Z}^m, \norm{k}_2 \leq \Lambda} \mathrm{e}^{\mathrm{i} \pi k \cdot (x - x_{\ell})} = \sum_{k \in \mathbb{Z}^m, \norm{k}_2 \leq \Lambda} \cos(\pi k \cdot (x - x_\ell)). 
\end{align}
This can be computed in time 
$\mathcal{O} \left( K_\Lambda \right)$, where $K_\Lambda = \left| \left\{ k \in \mathbb{Z}^m:\; \norm{k}_2 \leq \Lambda \right\} \right| \leq (2m+1)^{\Lambda^2}$, according to Rel.~\eqref{eq:numklambda} above.
Because we have chosen $\Lambda = \Theta(\sqrt{C / \epsilon})$, the runtime to evaluate one kernel function is upper bounded by $m^{\mathcal{O}(C / \epsilon)}$.

On the other hand, the computation of each $\Tr(O_j \sigma_1(\rho(x_{\ell})))$ can be performed in constant time after storing the data in a classical memory. This is a consequence of the tensor product structure of $\sigma_1(\rho(x_\ell)) = \bigotimes_{i=1}^{n} \left( 3 \ketbra{s^{(x_\ell)}_i}{s^{(x_\ell)}_i} - \mathbb{I} \right)$ which ensures
\begin{equation}
    \Tr(O_j \sigma_1(\rho(x_{\ell}))) = \Tr\Big(O_j \bigotimes_{i \in \mathrm{supp}(O_j)} \left( 3 \ketbra{s^{(x_\ell)}_i}{s^{(x_\ell)}_i} - \mathbb{I} \right)\Big),
\end{equation}
where $\mathrm{supp}(O_j)$ is the set of qubits in $\{1, \ldots, n\}$ the local observable $O_j$ acts on.
Because $|\mathrm{supp}(O_j)| \leq q= \mathcal{O}(1)$, computing $\Tr(O_j \sigma_1(\rho(x_{\ell})))$ takes only constant time. However, the computation time does scale exponentially in $|\mathrm{supp}(O_j)|$.
This can become a problem if $|\mathrm{supp}(O_j)|$ ceases to be a \emph{small} constant.
Putting everything together implies that  $\Tr(O \hat{\sigma}(x))$ can be computed in time (at most)
\begin{align}
    \mathcal{O}\left(N L m^{\mathcal{O}(C / \epsilon)}\right) = \mathcal{O}\left(L B^2 m^{\mathcal{O}(C / \epsilon)}\right) & ~~~~~~~~\text{(prediction time)}.
\end{align}
We conclude that both classical training time and prediction time for $\Tr(O \hat{\sigma}(x))$ are upper bounded by
\begin{equation}
    \mathcal{O}((n + L) B^2 m^{\mathcal{O}(C/\epsilon)}).
\end{equation}
This concludes the proof of all statements given in Theorem~\ref{thm:stateFourier}.

\subsection{Spectral gap implies smooth parametrizations} \label{sub:smoothness}

We attempt to deduce Theorem~\ref{thm:detailedFourier} from Theorem~\ref{thm:stateFourier}. 
The key step involves showing that the ground state $\rho(x)$ in a quantum phase of matter satisfies the following condition: For any observable $O = \sum_i O_i$ that can be written as a sum of local observables with $\sum_i \norm{O_i}_\infty \leq B$, we have
\begin{equation}
    \E_{x \sim [-1, 1]^m} \norm{\nabla_x \Tr(O \rho(x))}_2^2 \leq \mathcal{O}(B^2).
\end{equation}
Then we can apply Theorem~\ref{thm:stateFourier} with $C = \mathcal{O}(B^2)$  to derive Theorem~\ref{thm:detailedFourier}.

The average gradient magnitude $\E_{x \sim [-1, 1]^m} \norm{\nabla_x \Tr(O \rho(x))}_2^2$ depends on the observable $O$ in question, but also on the parametrization $x \mapsto H(x) \mapsto \rho (x)$. This section provides a useful smoothness bound based on physically meaningful assumptions:
\begin{enumerate}
    \item[(a)] \emph{Physical system:} 
    We consider $n$ finite-dimensional quantum many-body systems that are arranged at locations, or sites, in a $d$-dimensional space, e.g.,\ a spin chain ($d=1$), a square lattice ($d=2$), or a cubic lattice ($d=3$). Unless specified otherwise, our big-$\mathcal{O},\Omega, \Theta$ notation will be with respect to the thermodynamic limit $n\to\infty$.
    \item[(b)] \emph{Hamiltonian:} $H(x)$ decomposes into a sum of geometrically local terms $H(x) =\sum_j h_j (x)$, each of which only acts on an $\mathcal{O}(1)$ number of sites in a ball of $\mathcal{O}(1)$ radius. Individual terms $h_j (x)$ obey $\norm{h_j(x)}_\infty \leq 1$ and also have bounded directional derivative: $\norm{\partial h_j / \partial \uvec}_\infty \leq 1$, where $\uvec$ is a unit vector in parameter space. However, each term $h_j(x)$ can depend on the entire input vector $x \in [-1,1]^m$.
    
    \item[(c)] \emph{Ground-state subspace:} We consider ``the'' ground state $\rho(x)$ for the Hamiltonian $H(x)$ to be defined as $\rho(x) = \lim_{\beta \rightarrow \infty} \mathrm{e}^{-\beta H(x)} / \Tr(\mathrm{e}^{-\beta H(x)})$. This is equivalent to a uniform mixture over the eigenspace of $H(x)$ with the minimum eigenvalue.
    
    \item[(d)] \emph{Observable:} $O$ decomposes into a sum of few-body observables $O = \sum_i O_i$, each of which only acts on an $\mathcal{O}(1)$ number of sites. Each few-body observables $O_i$ can act on geometrically-nonlocal sites.
\end{enumerate}

Assumptions (a)--(c) should be viewed as mild technical assumptions that are often met in practice. The main result of this section bounds the smoothness condition based on an additional requirement.

\begin{lemma}[Spectral gap implies smoothness condition]\label{lem:truncatedFourier}
Consider a class of local Hamiltonians
\begin{equation}
    \left\{H(x):\;x \in [-1,1]^m\right\}
\end{equation}
and an observable $O=\sum_i O_i$ that obey the technical requirements (a)--(c) above. 
Moreover, suppose that the \emph{spectral} gap of each $H(x)$ is lower bounded by (constant) $\gamma > \Omega (1)$. Then,
\begin{equation}
\E_{x \sim [-1, 1]^m} \norm{\nabla_x \Tr(O \rho(x))}_2^2
\leq c_{\mathrm{all}}\Big(\sum_i \norm{O_i}_\infty\Big)^2.
\end{equation}
Here, $c_{\mathrm{all}}>0$ is a constant that depends on spatial dimension $d$, spectral gap $\gamma$, as well as the Lieb-Robinson velocities.
\end{lemma}

The proof is based on combining two powerful techniques from quantum many body physics. Namely, Lieb-Robinson bounds \cite{Lieb1972} to exploit locality and the spectral flow formalism \cite{bachmann2012automorphic}, also referred to as quasi-adiabatic evolution or continuation \cite{hastings2005quasiadiabatic, osborne2007simulating}, to exploit the spectral gap.

\paragraph{Quasi-adiabatic continuation for gapped Hamiltonians \cite{hastings2005quasiadiabatic, osborne2007simulating, bachmann2012automorphic}:} 
Given a quantum system satisfying the above assumptions (a)-(c), it is reasonable to expect that small changes in $x$
only lead to small changes in the associated ground state $\rho (x)$. 
Spectral flow makes this intuition precise.
Let the spectral gap of $H(x)$ be lower bounded by a constant $\gamma$ over $[-1,1]^m$. Then, the directional derivative of an associated ground state, in the direction defined by the parameter unit vector $\uvec$, obeys
\begin{equation}
\frac{\partial \rho}{\partial \uvec}(x) = \mathrm{i} [D_{\uvec}(x), \rho(x)] \quad \text{where} \quad 
    D_{\uvec}(x) = \int_{-\infty}^{\infty} W_\gamma(t) \mathrm{e}^{\mathrm{i} t H(x)} \frac{\partial H}{\partial \uvec}(x) \mathrm{e}^{-\mathrm{i} t H(x)} \mathrm{d} t .
\end{equation}
Here, $W_\gamma(t)$ is a fast decaying weight function that obeys $\sup_t \left|W_\gamma (t) \right|=1/2$ and only depends on the spectral gap. More precisely,
\begin{equation} \label{eq:Wbound}
    |W_\gamma(t)| \leq \begin{cases} \frac{1}{2} & 0 \leq \gamma |t| \leq \theta, \\ 35 \mathrm{e}^2 (\gamma |t|)^4 \mathrm{e}^{-\frac{2}{7} \frac{\gamma |t|}{\log(\gamma |t|)^2}} & \gamma |t| > \theta. \end{cases}
\end{equation}
The constant $\theta$ is chosen to be the largest real solution of $35 \mathrm{e}^2 \theta^4 \exp(-\frac{2}{7} \frac{\theta}{\log(\theta)^2}) = 1/2$.

\paragraph{Lieb-Robinson bounds for local Hamiltonians/observables \cite{Lieb1972,hastings2010locality}:} 
Let $\mathrm{supp}(X)$ denote the sites on which a many-body operator $X$ acts nontrivially.
Furthermore, for any two operators $X_1, X_2$, we define the distance $\Delta(X_1, X_2)$ to be the minimum distance between all pairs of sites acted on by $X_1$ and $X_2$, respectively, in the $d$-dimensional space.
We also consider the number of local terms in a ball of radius $r$. For any operator $X$ acting on a single site, 
this ball contains $\mathcal{O}(r^d)$ local terms in $d$-dimensional space,
\begin{equation} \label{eq:termsinball}
    \sum_{j: \Delta(X, h_j) \leq r} 1 \leq b_d + c_d r^d, \forall r \geq 0,
\end{equation}
where we recall the definition that $H = \sum_j h_j$ is a sum of local terms $h_j$.
The bound on the number of local terms in a ball of radius $r$ implies the existence of a Lieb-Robinson bound \cite{bravyi2006lieb, hastings2010locality}.
It states that for any two operator $X_1, X_2$ and any $t \in \mathbb{R}$, we have
\begin{equation} \label{eq:LRbound}
    \norm{[\exp(\mathrm{i}t H(x)) X_1 \exp(-\mathrm{i}t H(x)), X_2]}_\infty \leq c_{\mathrm{lr}} \norm{X_1}_\infty \norm{X_2}_\infty |\mathrm{supp}(X_1)| \mathrm{e}^{-a_{\mathrm{lr}} (\Delta(X_1, X_2) - v_{\mathrm{lr}} |t|)},
\end{equation}
for some constants $a_{\mathrm{lr}}, c_{\mathrm{lr}}, v_{\mathrm{lr}} = \Theta(1)$.

Apart from these two concepts, we will also need a bound on integrals of certain fast-decaying functions.

\begin{lemma}[Lemma~2.5 in \cite{bachmann2012automorphic}]
\label{lem:boundintua}
Fix $a>0$ and define the function $u_a (x)= \exp (- ax/\log(x)^2)$ on the domain $x \in (1,\infty)$. Then,
\begin{equation}
\int_{t}^{\infty} x^k u_a(x) \mathrm{d} x \leq \frac{2k + 3}{a} t^{2k + 2} u_a(t)
\quad \text{for all $t>\mathrm{e}^4$ and $k \in \mathbb{N}$ that obey} \quad 2k+2 \leq  \frac{at}{\log(t)^2}.
\end{equation}
\end{lemma}

\begin{proof}[Proof of Lemma~\ref{lem:truncatedFourier}]

Fix an input $x \in \left[-1,1\right]^n$ and a unit vector $\uvec \in \mathbb{R}^n$ (direction). We may then rewrite the associated directional derivative of $\rho (x)$ in two ways, namely
\begin{subequations}
\begin{align}
\frac{\partial \rho}{\partial \uvec}(x) &=
\uvec \cdot \nabla_x \rho(x), \quad \text{and}\\
\frac{\partial \rho}{\partial \uvec}(x)
&= -\mathrm{i} \left[D_{\uvec}(x),\rho (x) \right] \quad \text{with} \quad D_{\uvec}(x) = \int_{-\infty}^{\infty} \mathrm{d}t W_\gamma(t) \mathrm{e}^{\mathrm{i} t H(x)} \frac{\partial H}{\partial \uvec}(x) \mathrm{e}^{-\mathrm{i} t H(x)}.
\end{align}
\end{subequations}
When evaluated on an observable $O$, this establishes the following correspondence:
\begin{equation}
\uvec \cdot \nabla_x \Tr(O \rho(x))
= \Tr \left( O \left[ D_{\uvec}(x),\rho (x) \right] \right) = \Tr \left( \left[O,D_{\uvec}(x)\right] \rho (x) \right),
\end{equation}
for any $\uvec$. Choosing $\uvec= \uvec(x,O) = \frac{ \nabla_x \Tr(O \rho(x)) }{\norm{\nabla_x \Tr(O \rho(x))}_2}$ implies
\begin{equation}
    \norm{\nabla_x \Tr(O \rho(x))}_2^2 = \left|\Tr([O, D_{\uvec(x, O)}(x)] \rho(x))\right|^2.
\end{equation}
The left hand side is the magnitude of steepest slope in a phase for the particular observable $O$.
The average slope over the entire domain $[-1, 1]^m$ is thus given as
\begin{equation} \label{eq:gradaverage}
    \E_{x \sim [-1, 1]^m} \norm{\nabla_x \Tr(O \rho(x))}_2^2 = \frac{1}{2^m} \int_{[-1, 1]^m} \left|\Tr([O, D_{\uvec(x, O)}(x)] \rho(x))\right|^2 \mathrm{d}^m x.
\end{equation}
Intuitively, thermodynamic observables should not change too rapidly within a phase.
Making this intuition precise will allow us to upper bound the average slope by a constant $C$.

We first expand $D_{\uvec}(x)$ and apply a triangle inequality to obtain
\begin{align} \label{eq:uppslope}
    |\Tr([O, D_{\uvec}(x)] \rho(x))| & \leq \sum_i \int_{-\infty}^{\infty}  W_\gamma(t) \sum_j \norm{\left[O_i, \mathrm{e}^{\mathrm{i} t H(x)} \frac{\partial h_j}{\partial \uvec}(x) \mathrm{e}^{-\mathrm{i} t H(x)}\right]}_\infty \mathrm{d}t.
\end{align}
For fixed $t$, we can separate local Hamiltonian terms into two groups, defines using the constants in the Lieb-Robinson bound (\ref{eq:LRbound}). The first group contains all terms $h_j$ that obey $\Delta(O_i, h_j) \leq v_{\mathrm{lr}} |t|$. The second group contains all $h_j$ that obey $\Delta(O_i, h_j) > v_{\mathrm{lr}} |t|$ instead.
Equation~\eqref{eq:termsinball} above provides a useful bound on the size of the first group.
It contains at most $|\mathrm{supp}(O_i)| (b_d + c_d (v_{\mathrm{lr}} |t|)^d) \leq c_O (b_d + c_d (v_{\mathrm{lr}} |t|)^d)$ local terms $h_j$, for some constant $c_O\leq \mathrm{supp}(O)$.
We can bound the summation over these terms using $\norm{[A, B]}_\infty \leq 2 \norm{A}_\infty \norm{B}_\infty$ to obtain
\begin{subequations}
\begin{align}
    \sum_{j: \Delta(O_i, h_j) \leq v_{\mathrm{lr}} t} \norm{\left[O_i, \mathrm{e}^{\mathrm{i} t H(x)} \frac{\partial h_j}{\partial \uvec}(x) \mathrm{e}^{-\mathrm{i} t H(x)}\right]}_\infty &\leq c_O (b_d + c_d (v_{\mathrm{lr}} |t|)^d) \times 2 \norm{O_i}_\infty \norm{\frac{\partial h_j}{\partial \uvec}}_\infty\\
    &\leq 2 c_O \norm{O_i}_\infty (b_d + c_d (v_{\mathrm{lr}} |t|)^d).
\end{align}
\end{subequations}
The second inequality follows from 
technical assumption (b): $\norm{\partial h_j / \partial \uvec}_\infty \leq 1$.

The contributions from the second group can be controlled via the 
Lieb-Robinson bound from Eq.~\eqref{eq:LRbound}. For every 
$h_j$ that obeys $\Delta(O_i, h_j) > v_{\mathrm{lr}} |t|$, we have
\begin{subequations}
\begin{align}
    \norm{\left[O_i, \mathrm{e}^{\mathrm{i} t H(x)} \frac{\partial h_j}{\partial \uvec}(x) \mathrm{e}^{-\mathrm{i} t H(x)}\right]}_\infty &\leq c_{\mathrm{lr}} \norm{O_i}_\infty \norm{\partial h_j / \partial \uvec}_\infty |\mathrm{supp}(h_j)| \mathrm{e}^{-a_{\mathrm{lr}} (\Delta(O_i, h_j) - v_{\mathrm{lr}} |t|)}\\
    &\leq c_{\mathrm{lr}} c_h \norm{O_i}_\infty \mathrm{e}^{-a_{\mathrm{lr}} (\Delta(O_i, h_j) - v_{\mathrm{lr}} |t|)}.
\end{align}
\end{subequations}
Reusing Eq.~\eqref{eq:termsinball}, we conclude that there are at most $|\mathrm{supp}(O_i)| (b_d + c_d (v_{\mathrm{lr}} |t| + r + 1)^d)$ local terms $h_j$ with $\Delta(O_i, h_j) \in [v_{\mathrm{lr}} |t| + r, v_{\mathrm{lr}} |t| + r + 1]$. This ensures
\begin{subequations}
\begin{align}
    & \sum_{j: \Delta(O_i, h_j) > v_{\mathrm{lr}} |t|} \norm{\left[O_i, \mathrm{e}^{\mathrm{i} t H(x)} \frac{\partial h_j}{\partial \uvec}(x) \mathrm{e}^{-\mathrm{i} t H(x)}\right]}_\infty \nonumber\\ 
    & \leq \sum_{r = 0}^\infty \sum_{j: \Delta(O_i, h_j) \in [v_{\mathrm{lr}} |t| + r, v_{\mathrm{lr}} |t|+ r + 1]} \norm{\left[O_i, \mathrm{e}^{\mathrm{i} t H(x)} \frac{\partial h_j}{\partial \uvec}(x) \mathrm{e}^{-\mathrm{i} t H(x)}\right]}_\infty \\
    & \leq \int_{r = 0}^{\infty} \mathrm{d} r c_{\mathrm{lr}} c_h \norm{O_i}_\infty \mathrm{e}^{-a_{\mathrm{lr}} r} \times \mathrm{supp}(O_i) (b_d + c_d (v_{\mathrm{lr}} |t| + r + 1)^d)\\
    & \leq c_{\mathrm{lr}} c_h c_O \norm{O_i}_\infty \int_{r = 0}^{\infty} \mathrm{d} r \mathrm{e}^{-a_{\mathrm{lr}} r} (b_d + c_d (v_{\mathrm{lr}} |t| + r + 1)^d)\\
    & \leq c_{\mathrm{lr}} c_h c_O \norm{O_i}_\infty \left( \frac{b_d}{a_{\mathrm{lr}}} + c_d \sum_{p=0}^d \frac{d!}{p! a_{\mathrm{lr}}^{d-p+1}} (v_{\mathrm{lr}} |t| + 1)^p \right).
\end{align}
\end{subequations}
We can now combine the two bounds into a single statement:
\begin{equation}
    \sum_j \norm{\left[O_i, \mathrm{e}^{\mathrm{i} t H(x)} \frac{\partial h_j}{\partial \uvec}(x) \mathrm{e}^{-\mathrm{i} t H(x)}\right]}_\infty \leq \norm{O_i}_\infty \sum_{p=0}^d C_p |t|^p. \label{eq:gap-aux1}
\end{equation}
Here, we have implicitly defined a new set of constants $C_p$ that depend on the constants $c_O, c_h, c_{\mathrm{lr}}, c_d, a_{\mathrm{lr}}, v_{\mathrm{lr}}, d$ that had already featured before. Plugging the above into Eq.~\eqref{eq:uppslope} and substituting the spectral flow weight function $W$ \eqref{eq:Wbound} for its absolute value allows us to bound the maximum slope of $\Tr(O \rho(x))$ when the Hamiltonian moves from $H(x)$ to $H(x+\mathrm{d}\uvec)$. Indeed,
\begin{equation}
    |\Tr([O, D_{\uvec}(x)] \rho(x))| \leq
    \Big(\sum_i \norm{O_i}_\infty\Big) \sum_{p=0}^d C_p \int_{-\infty}^{\infty}  |W_\gamma(t)| |t|^p \mathrm{d}t~.
\end{equation}
To bound the resulting integral, we recall that $W_\gamma (t)$ obeys $\sup_t |W_\gamma(t)| = 1/2$, define $t^* = \max(\mathrm{e}^4, 7(d+5), \theta) / \gamma$, and split up the integration into two parts, $t \in [-t^*, t^*]$ and $t \notin [-t^*, t^*]$. Symmetry then ensures
\begin{subequations}
\begin{align}
    \int_{-\infty}^{\infty} \mathrm{d} t |W_\gamma(t)| |t|^p &\leq \frac{1}{2} \int_{-t^*}^{t^*} \mathrm{d} t |t|^p + 2 \int_{t^*}^\infty \mathrm{d} t \, 35 \mathrm{e}^2 (\gamma t)^4 \mathrm{e}^{-\frac{2}{7} \frac{\gamma t}{\log(\gamma t)^2}} t^p \\
    & = \int_{0}^{t^*} \mathrm{d} t \, t^p + 70 \mathrm{e}^2 \gamma^{-p-1} \int_{x = \gamma t^*}^\infty \mathrm{d} x \, x^{p+4} \mathrm{e}^{-\frac{2}{7} \frac{x}{\log(x)^2}}.
\end{align}
\end{subequations}
The first integral is straightforward, and the second integral can be bounded using Lemma~\ref{lem:boundintua}. Set $a=2/7$, $k =p+4$ and note that we have chosen $t^*$ such that all assumptions are valid. 
Applying Lemma~\ref{lem:boundintua} ensures
\begin{subequations}
\begin{align}
    \int_{-\infty}^{\infty} \mathrm{d} t |W_\gamma(t)| |t|^p \mathrm{d} t &\leq \frac{|t^*|^{p+1}}{p+1} + 70 \mathrm{e}^2 \gamma^{-p-1} \frac{2k + 3}{a} (\gamma t^*)^{2k+2} \mathrm{e}^{-\frac{2 \gamma t^*}{7 \log(\gamma t^*)^2}}\\
    & = \frac{|t^*|^{p+1}}{p+1} + 35 \mathrm{e}^2 \gamma^{-p-1} 7(2p + 11) (\gamma t^*)^{2p + 10} \mathrm{e}^{-\frac{2 \gamma t^*}{7 \log(\gamma t^*)^2}},
\end{align}
\end{subequations}
for any integer $0 \leq p \leq d$. Inserting these bounds into the sum \eqref{eq:gap-aux1} implies
\begin{align}
    |\Tr([O, D_{\uvec}(x)] \rho(x))| &\leq \Big(\sum_i \norm{O_i}_\infty\Big) \sum_{p=0}^d C_p \left( \frac{|t^*|^{p+1}}{p+1} + 35 \mathrm{e}^2 \gamma^{-p-1} 7(2p + 11) (\gamma t^*)^{2p + 10} \mathrm{e}^{-\frac{2 \gamma t^*}{7 \log(\gamma t^*)^2}} \right).
\end{align}
Recall that $t^* = \max(\mathrm{e}^4, 7(d+5), \theta) / \gamma$ is a constant
that only depends on $d$ and $\gamma$, and the $C_p$'s are also constants that depend on on $c_O, c_h, c_{\mathrm{lr}}, c_d, a_{\mathrm{lr}}, v_{\mathrm{lr}}, d$.
We may subsume all of these constant contributions in a new constant $c_{\mathrm{all}}$ and conclude
\begin{equation} \label{eq:TrODrho}
    |\Tr([O, D_{\uvec}(x)] \rho(x))| \leq c_{\mathrm{all}} \Big(\sum_i \norm{O_i}_\infty\Big).
\end{equation}
Inserting this uniform upper bound into Eq.~\eqref{eq:gradaverage} completes the proof of Lemma~\ref{lem:truncatedFourier}.
\end{proof}

\section{Sample complexity lower bound for predicting ground states}
\label{app:proofthmlowerbound}

This section establishes an information-theoretic lower bound for the task of predicting ground state approximations.
It highlights that, without further assumptions on the Hamiltonians, the training data size required in Theorem~\ref{thm:stateFourier} is essentially tight.

\begin{theorem} \label{thm:lowerboundsmooth-restatement}
Fix a prediction error tolerance $\epsilon$, a number $m$ of parameters, as well as constants $C,B>0$ such that
$C / (9 \epsilon) \leq m^{0.99}$.
Consider a quantum ML model that learns from  quantum data $\{x_\ell \rightarrow \rho(x_\ell)\}_{\ell=1}^{N}$ of size $N$ to generate ground state predictions $\hat{\sigma}(x)$, where $x \in [-1, 1]^m$.
Suppose the quantum ML model can achieve
\begin{equation}
    \E_{x \sim [-1, 1]^m} |\Tr(O \hat{\sigma}(x)) - \Tr(O \rho(x))|^2 \leq \epsilon,
\end{equation}
with high probability, for every class of Hamiltonians $H(x)$ and for every observable $O$ given as a sum of local observables $\sum_i O_i$ that obey
\begin{subequations}
\begin{align}
\E_{x \sim [-1, 1]^m} \norm{\nabla_x \Tr(O \rho(x))}_2^2 &\leq C & \text{(smoothness condition)}, \label{eq:smthcond}\\
\sum_i \norm{O_i} &\leq B  & \text{(bounded norm)}. \label{eq:bddnormO}
\end{align}
\end{subequations}
Then, the (quantum) training data size must obey
\begin{equation}
    N \geq B^2 m^{\Omega(C / \epsilon)} / \log(B). \label{eq:lower-bound-appendix}
\end{equation}
This is also a lower bound on quantum computational time associated with the quantum ML model.
\end{theorem}

The assumption $C / (9 \epsilon) \leq m^{0.99}$ is required for technical reasons outlined below. It is equivalent to demanding that the prediction error tolerance is large enough compared to the inverse of $m$, i.e., $\epsilon \geq C / (9 m^{0.99})$.
If the quantum ML model can achieve an even smaller prediction error, such that $\E_{x \sim [-1, 1]^m} |\Tr(O \hat{\sigma}(x)) - \Tr(O \rho(x))|^2 < C / (9 m^{0.99})$, then we choose $\epsilon = C / (9 m^{0.99})$.
For such a choice of $\epsilon$, the training data size lower bound becomes $N \geq B^2 m^{\Omega(m^{0.99})} / \log(B)$, which is exponential in $m^{0.99}$.
Hence, in all cases, we need $\epsilon$ to be a constant for any (quantum or classical) machine learning algorithm to obtain a sample complexity that scales polynomially in $m$.

We prove Theorem~\ref{thm:lowerboundsmooth-restatement} by means of an information-theoretic analysis. Conceptually, it resembles arguments developed in prior work \cite{huang2021information} (sample complexity lower bound for general quantum machine learning models).
Section~\ref{sub:lowerboundsmooth-learning-problem} formulates a learning problem that involves predicting ground state properties of a certain class of Hamiltonians.
Subsequently, Section~\ref{sub:lowerboundsmooth-communication-protocol} incorporates a hypothetical (quantum ML) solution to this learning problem as a decoding procedure in a communication protocol. Information-theoretic bottlenecks then beget fundamental restrictions on the sample complexity of any ML model that solves the learning problem, see Section~\ref{sub:lowerboundsmooth-analysis}.

\subsection{Learning problem formulation} \label{sub:lowerboundsmooth-learning-problem}

We consider a family of single-qubit Hamiltonians, i.e.\ $n=1$, that is parametrized by $m$ degrees of freedom. 
We first map $x \in \left[-1,1\right]^m$ to a real number by evaluating a truncated Fourier series $f_a$.
Fix a cutoff $\Lambda = \sqrt{C/(9\epsilon)}$ and let
\begin{equation}\label{eq:defLambda}
K_\Lambda= \left| \left\{k \in \mathbb{Z}^m:\; \norm{k}_2 \leq \Lambda = \sqrt{C/(9\epsilon)}\right\} \right|
\end{equation}
denote the number of $n$-dimensional wave-vectors with Euclidean norm at most $\Lambda$. 
We equip each of these wave vectors $k$ with a sign $a_k \in \left\{ \pm 1 \right\}$ and define the function
\begin{equation}
f_{a}(x) = \sqrt{\frac{9 \epsilon}{K_\Lambda}} \sum_{k \in \mathbb{Z}^m, \norm{k}_2 \leq \Lambda} a_k \cos \left( \pi k \cdot x \right),
\quad \text{where} \quad a \in \left\{ \pm 1 \right\}^{K_\Lambda},\label{eq:lowerboundsmooth-function}
\end{equation}
subsumes all sign choices involved.
We use this function to define a single-qubit Hamiltonian. For Pauli matrices $X$ and $Z$, we set
\begin{equation}
H_{a}(x) = \exp \left(+ \tfrac{\mathrm{i}}{2} \arcsin \left( f_{a}(x)/B\right) X \right) \left(-Z \right) \exp \left(- \tfrac{\mathrm{i}}{2} \arcsin \left(f_{a}(x)/B\right) X \right), \label{eq:lowerboundsmooth-hamiltonian}
\end{equation}
where $B$ is a constant that will reflect the size of the target observable, see Eq.~\eqref{eq:bddnormO}.
To summarize, each choice of $a \in \left\{\pm 1\right\}^{K_\Lambda}$ yields an entire class of single-qubit Hamiltonians $H_a (x)$ that is parametrized by $m$-dimensional inputs $x \in \left[-1,1\right]^m$.
These stylized Hamiltonians are simple enough to compute their (nondegenerate) ground state explicitly:
\begin{equation}
\rho_{a}(x) = \ketbra{\psi_{a}(x)}{\psi_{a}(x)} \quad \text{with} \quad
| \psi_{a}(x) \rangle
= \left(
\begin{array}{c}
\cos \left( \tfrac{1}{2}\arcsin (f_a(x)/B)\right) \\
\mathrm{i} \sin \left(\tfrac{1}{2} \arcsin (f_a(x)/B)\right)
\end{array}
\right) \in \mathbb{C}^2.
\label{eq:lowerboundsmooth-ground-state}
\end{equation}
Finally, we fix the single-qubit observable $O$ to be a scaled version of Pauli $Y$. Setting
$
    O = BY
$
yields a $1$-local observable. And, more importantly,
\begin{subequations}
\label{eq:lowerboundsmooth-learning-problem}
\begin{align}
\Tr \left( O \rho_{a}(x)\right)
&= B \langle \psi_{a}| Y |\psi_{a}\rangle
= B \left(-\mathrm{i} \overline{\langle 0| \psi_{a}(x) \rangle} \langle 1| \psi_{a}(x) \rangle + \mathrm{i} \langle 0| \psi_{a}(x) \rangle \overline{ \langle 1 | \psi_{a}(x) \rangle} \right)\\
&= 
2B \cos \left(\tfrac{1}{2} \arcsin \left( f_a (x)/B\right) \right) \sin \left( \tfrac{1}{2}\arcsin \left( f_a (x)/B\right)\right) \\
&= B \sin \left( \arcsin \left(f_a(x)/B\right)\right) = f_a(x).
\end{align}
\end{subequations}
By construction, the expectation value $\mathrm{tr}\left(O \rho_{a}(x) \right)$ exactly reproduces the function $f_a (x)$ defined in Eq.~\eqref{eq:lowerboundsmooth-function}.
Being able to accurately predict it will be equivalent to accurately learning this function -- regardless of the underlying sign parameter $a \in \left\{\pm 1\right\}^{K_\Lambda}$.

To complete the formulation of the learning problem, we recall that the training parameters are sampled from the uniform distribution over the hypercube, $\mathrm{Unif} \left[-1,1\right]^m$, and that we will evaluate the expectation $\E$ over $x$ with respect to this distribution from now on. This choice of distribution implies a nice closed-form expression for the average squared distance of two functions $f_a, f_b$. For $a,b \in \left\{\pm 1\right\}^{K_\Lambda}$, 
\begin{subequations}
\label{eq:lowerboundsmooth-hamming}
\begin{align}
\E_{x} \left| f_a (x) - f_b (x) \right|^2 &= \frac{9\epsilon}{K_\Lambda}\sum_{k,l \in \mathbb{Z}^m, \norm{k}_2,\norm{l}_2 \leq \Lambda} (a_k-b_k)(a_l-b_l) \int_{[-1,1]^m}\cos \left(\pi k \cdot x \right) \cos \left( \pi l \cdot x \right)\mathrm{d}^m x \\
&= \frac{9\epsilon}{K_\Lambda}\sum_{k \in \mathbb{Z}^m, \norm{k}_2 \leq \Lambda} \left( a_k - b_k \right)^2 \\
&= \frac{9\epsilon}{K_\Lambda}\sum_{k \in \mathbb{Z}^m, \norm{k}_2 \leq \Lambda}4 \times \mathbf{1} \left\{ a_k \neq b_k \right\} \\
&= \frac{36\epsilon}{K_\Lambda}d_H (a,b)~, 
\end{align}
\end{subequations}
where we have used orthonormality of the Fourier basis $\cos (\pi k \cdot x)$, and $d_H (a,b) = \sum_k \mathbf{1} \left\{a_k \neq b_k \right\}$ is the \emph{Hamming distance} on $\left\{\pm 1\right\}^{K_\Lambda}$.

We conclude this expository section by examining whether the construction fulfills the requirement stated in the theorem and presenting a technical lemma.
First of all, we have
\begin{equation}
 \norm{O} = B \norm{Y} = B   ,
\end{equation}
which satisfies the bounded norm constraint in Eq.~\eqref{eq:bddnormO}.
Furthermore, we can use the orthonormality of $\cos(\pi k \cdot t)$ to find that
\begin{equation}
     \E_{x \sim [-1, 1]^m} \norm{\nabla_x \Tr(O \rho_a(x))}_2^2 = \frac{9 \epsilon}{K_\Lambda} \sum_{k \in \mathbb{Z}^m: \norm{k}_2 \leq \Lambda} \norm{k}_2^2 |a_k|^2 \leq \frac{9 \epsilon}{K_\Lambda} K_\Lambda \Lambda^2 = C.
\end{equation}
Thus the smoothness condition in Eq.~\eqref{eq:smthcond} is also satisfied.
Now, we turn our attention to the ground state~\eqref{eq:lowerboundsmooth-ground-state}. The following technical lemma exposes the function $f_a(x)/B$ directly in the amplitudes of ground states.

\begin{lemma} \label{lem:lowerboundsmooth-groundstate}
Let $|\psi_{a}(x) \rangle$ be the ground state of Hamiltonian $H_{a}$ defined in Eq.~\eqref{eq:lowerboundsmooth-ground-state}. Then,
\begin{equation}
\rho_{a}(x)=\ketbra{\psi_{a}(x)}{\psi_{a}(x)}
= \tfrac{1}{2}\left(
\begin{array}{cc}
1+\sqrt{1- (f_a(x)/B)^2} &
- \mathrm{i}f_a (x)/B \\
\mathrm{i} f_a (x)/B &
1- \sqrt{1- (f_a(x)/B)^2}
\end{array}
\right). \label{eq:lowerboundsmooth-ground-state2}
\end{equation}
\end{lemma}

\begin{proof}
The proof is based on double-angle and half-angle trigonometric identities. 
Suppressing $x$ dependence, the first diagonal entry becomes
\begin{subequations}
\begin{align}
\langle 0| \rho_a |0 \rangle
&= | \langle 0| \psi_{a}\rangle|^2
= \cos^2 \left( \tfrac{1}{2} \arcsin \left(f_a /B\right) \right) 
= \tfrac{1}{2} \left( 1 + \cos \left( \arcsin \left(f_a /B\right) \right) \right) \\
&= \tfrac{1}{2} \left( 1+ \sqrt{1- \sin^2 \left( \arcsin \left( f_a /B \right) \right)} \right) = \tfrac{1}{2} \left( 1 + \sqrt{1- (f_a/B)^2}\right),
\end{align}
\end{subequations}
and normalization implies that $\langle 1| \rho_{a}|1 \rangle = 1 - \langle 0| \rho_{a} |0 \rangle$.
The off-diagonal entries are
\begin{subequations}
\begin{align}
\overline{\langle 1| \rho_{a}|0 \rangle} &=
 \langle 0| \rho_{a} |1 \rangle
= \langle 0| \psi_{a} \rangle \langle \psi_{a} |1 \rangle
=-\mathrm{i} \cos \left(\tfrac{1}{2} \arcsin \left(f_a /B\right)\right) \sin \left( \tfrac{1}{2} \arcsin \left(f_a /B\right) \right) \\
&= -\tfrac{\mathrm{i}}{2}\sin \left( \arcsin \left( f_a /B\right) \right) 
= -\tfrac{\mathrm{i}}{2}f_a /B~.
\end{align}
\end{subequations}
\end{proof}

\subsection{Communication protocol} \label{sub:lowerboundsmooth-communication-protocol}

Consider the learning problem introduced in the previous section.
Suppose that a quantum ML model can use training data $\mathcal{T}=\left\{x_\ell,\rho_{a}(x_\ell)\right\}_{\ell=1}^{N}$ 
to learn a function $f^Q(x)$ that (on average) predicts $\mathrm{tr} \left( O \rho_{a}(x) \right) = f_a (x)$ for a particular unknown $a \in \left\{\pm 1\right\}^{K_\Lambda}$, up to some accuracy $\epsilon$,
\begin{equation}
    \E_{x } \left| f^Q(x)-f_a (x) \right|^2 \leq
    \epsilon~.
\end{equation}
Such a model will not fare as well in estimating the expectation value associated with $b\neq a$, whenever $b$ is sufficiently far away from $a$. Using the triangle inequality and Eq.~\eqref{eq:lowerboundsmooth-hamming},
\begin{equation}
    \E_{x } \left| f^Q(x)-f_b (x) \right|^2
\geq \E_{x } \left|f_a (x) - f_b (x) \right|^2 - \E_{x } \left| f^Q(x)-f_a (x)\right|^2 \geq \frac{36 \epsilon}{K_\Lambda}d_H (a,b)-\epsilon~.
\end{equation}
The model's accuracy significantly worsens at $d_H (a,b) > K_\Lambda/18$, where we recall $K_\Lambda= \left| \left\{k \in\mathbb{Z}^m:\; \norm{k}_2 \leq \Lambda \right\} \right|$ from Eq.~\eqref{eq:defLambda}. In other words, a good quantum ML model would allow us to use training data $\mathcal{T}$ in order to recover the underlying parameter $a \in \left\{\pm 1 \right\}^{K_\Lambda}$ up to resolution $K_\Lambda/18$ in Hamming distance.

We can use this assertion as an effective decoding procedure in a two-way communication protocol involving Alice and Bob. To accommodate imperfect resolution, Alice and Bob agree on a dictionary of sign vectors $\left\{a^{(1)},\ldots,a^{(M)} \right\} \subset \left\{ \pm 1 \right\}^{K_\Lambda}$ whose pairwise Hamming distance is large enough: $d_H (a_i,a_j) > K_\Lambda/18$ for all $i \neq j$. Let $M$ denote the cardinality of this dictionary.
Alice and Bob use this dictionary and the ML procedure to transmit integers up to size $M$ over a quantum channel.
Alice samples an integer $j \in \left\{1,\ldots,M\right\}$ and sets $a=a^{(j)} \in \left\{\pm 1 \right\}^{K_\Lambda}$. Subsequently, she uses $a$ to generate (quantum) training data $\mathcal{T} = \left\{ (x_\ell,\rho_{a}(x_\ell))\right\}_{\ell=1}^{N}$ with $x_1,\ldots,x_{N}\sim \mathrm{Unif}\left[-1,1\right]^m$ which she passes on to Bob. Subsequently, Bob uses $\mathcal{T}$ to train a quantum ML model to predict the underlying function $\Tr \left( O \rho_{a}(x)\right) = f_a (x)$. By checking $\E_{x } \left| f_{\bar{a}}(x)-f^Q (x) \right|^2 \leq \epsilon$ for every possible dictionary element $\bar{a}$, he will retrieve the correct message with high probability, i.e., $\bar{a} = a$. 

This is a protocol that conveys classical information via a quantum dataset. It is subject to fundamental constraints from information theory. These will allow us to deduce a lower bound on the required training data size $N = |\mathcal{T}|$.
An important figure of merit in this argument is the cardinality $M$ of the dictionary. That is, the number of different integers that can be communicated. The larger $M$, the more powerful the communication protocol, and following result, sometimes attributed to Gilbert and Varshamov \cite{gilbert1952comparison}, is a lower bound on how many bits one can ``pack'' into the space of $L$-bit strings while maintaining the required distance. 

\begin{lemma}[Lemma~5.12 in \cite{rigollet201518}] \label{lem:dictionary}
There exists a dictionary $\left\{a^{(1)},\ldots,a^{(M)}\right\} \in \left\{\pm 1 \right\}^{K_\Lambda}$ of cardinality $M \geq \lfloor \exp \left( K_\Lambda/32\right)\rfloor$ that achieves $d_H \left( a^{(i)},a^{(j)}\right) \geq K_\Lambda/4$ whenever $i \neq j$.
\end{lemma}

\subsection{Information-theoretic analysis}
\label{sub:lowerboundsmooth-analysis}

Let us now take a closer look at the communication protocol introduced above by bounding the correlation between Alice's original randomly chosen message $a$ and Bob's decoded signal $\bar{a}$. Up to now, we have stablished the following.
Per the bound in Lemma~\ref{lem:dictionary}, the dictionary of available $a$'s can be chosen to be rather large: $M=\lfloor \exp \left( K_\Lambda/32\right)\rfloor$. Moreover, the existence of a good quantum ML procedure, in the sense of Proposition~\ref{thm:lowerboundsmooth-restatement}, ensures that $\bar{a}=a$ with high probability. 

Correlations between Alice's and Bob's variables are quantified by the (classical) mutual information
\begin{equation}
I (a: \bar{a} ) \geq \Omega \left( \log (M) \right) = \Omega (K_\Lambda),
\label{eq:lowerboundsmooth-aux1}
\end{equation}
which we have bounded from below using Fano's inequality \cite{yu1997assouad}. Our task now is to provide an upper bound on $I(a:\bar{a})$, in terms of $N,B$ and $\epsilon$, in order to relate those parameters to $K_\Lambda$ and obtain the desired result in Theorem~\ref{thm:lowerboundsmooth-restatement}.

Since the parameters $x_1, \ldots, x_N$ are sampled independently from $a$, we have $I(a : x_1, \ldots, x_N) = 0$ and $a|_{x_1, \ldots, x_N} = a$. Therefore, we can upper bound the mutual information as follows,
\begin{subequations}\label{eq:mutual-information0}
    \begin{align}
    I\left(a:\bar{a}\right)&\leq I\left(a:\bar{a}, x_1, \ldots, x_N\right)\\
    &= I\left(a: x_1, \ldots, x_N\right) + I\left(a: \bar{a} | x_1, \ldots, x_N\right)\\
    &= I\left(a: \bar{a} | x_1, \ldots, x_N\right)\\
    &= \E_{x_{1},\dots, x_{N}} I\left(a|_{x_1, \ldots, x_N} : \bar{a}|_{x_1, \ldots, x_N}\right) \label{eq:xarerandomvar} \\
    &= \E_{x_{1},\dots, x_{N}} I\left(a : \bar{a}|_{x_1, \ldots, x_N}\right)~,
\end{align}
\end{subequations}
where $Q|_{x}$ denotes the random variable $Q$ conditioned on the random variable $x$.

Next, recall that Bob reconstructs the classical $\bar{a}$ by performing quantum operations on the training data $\mathcal{T} = \left\{ (x_\ell, \rho_{a}(x_\ell) \right\}_{\ell=1}^{N}$.
For each instance of randomly chosen parameters $x_1, \ldots, x_{N} \sim \mathrm{Unif}[-1, 1]^m$, Bob performs a quantum measurement on the state $\bigotimes_{\ell = 1}^N \rho_{a}(x_\ell)$ and uses the measurement outcomes to reconstruct $\bar{a}$. Bob's procedure is equivalent to performing the quantum ML algorithm that we have been promised in Sec.~\ref{sub:lowerboundsmooth-communication-protocol}. Thus we can use Holevo's theorem \cite{holevo}[\citealp{wildebook}, Sec.~11.6.1] to write
\begin{equation}
    I\left(a : \bar{a}|_{x_1, \ldots, x_N}\right) \leq \chi\left(a:\bigotimes_{\ell=1}^{N} \rho_{a}\left(x_{\ell}\right)\Big|_{x_1, \ldots, x_N} \right)~,
    \label{eq:mutual-information}
\end{equation}
where the Holevo information $\chi$ quantifies correlations between a random variable $z$ and a quantum state $\rho_z$,
\begin{equation}\label{eq:holevo}
    \chi(z : \rho_z) = S\left(\E_z \rho_z \right) - \E_z S(\rho_z)~,
\end{equation} 
and $S(\rho) = - \Tr(\rho \log \rho)$ is the von Neumann entropy. In other words, for each instance of parameters, the correlation between $a$ and $\bar{a}$ is bounded by the Holevo information of Bob's ensemble of quantum states.

Next, we use the subadditivity of von Neumann entropy, $S(\E_z \rho_z \otimes \sigma_z) \leq S(\E_z \rho_z) + S(\E_z \sigma_z)$, and the additivity of entropy for independent systems, $S(\rho \otimes \sigma) = S(\rho) + S(\sigma)$, to obtain
\begin{equation}\label{eq:subadd}
    \chi\left(a:\bigotimes_{\ell=1}^{N} \rho_{a}\left(x_{\ell}\right)\Big|_{x_1, \ldots, x_N} \right) \leq\sum_{\ell=1}^{N}\chi\left(a:\rho_{a}\left(x_{\ell}\right)\big|_{x_{1},\cdots,x_{N}}\right)~.
\end{equation}
Plugging Eqs.~(\ref{eq:mutual-information}) and (\ref{eq:subadd}) into Eq.~(\ref{eq:mutual-information0}) and using the fact that $\rho_a(x_\ell)$ is independent to $x_{\ell'}$ for any $\ell' \neq \ell$, we obtain
\begin{subequations}
    \begin{align}
        I\left(a:\bar{a}\right)&\leq\sum_{\ell=1}^{N}\E_{x_{1},\cdots,x_{N}}\chi\left(a:\rho_{a}\left(x_{\ell}\right)\big|_{x_{1},\cdots,x_{N}}\right)\\&=\sum_{\ell=1}^{N}\E_{x_{\ell}}\chi\left(a:\rho_{a}\left(x_{\ell}\right)\big|_{x_{\ell}}\right)\\&=N\E_{x}\chi\left(a:\rho_{a}\left(x\right)\right)~.
    \end{align}
\end{subequations}
The last equality follows from the fact that each $(x_\ell, \rho_a(x_\ell))$ is generated independently and in an identical fashion for all $\ell = 1, \ldots, N$.

We have thus reduced the problem of bounding the correlations between classical variables $a$ and $\bar{a}$ to that of bounding the Holevo information of the ensemble of states $\rho_a$ --- a much simpler problem because $\rho_a$ is a two-by-two matrix. In Lemma~\ref{lem:technical-bound} at the end of section, we obtain the bound
\begin{equation}\label{eq:technical-bound}
\E_{x} \chi \left( a: \rho_{a}(x) \right)
\leq \frac{9 \epsilon}{4B^2} \log \left( \frac{4 \mathrm{e}B^2}{9 \epsilon} \right)~. 
\end{equation}
Using this bound, the first claim in Theorem~\ref{thm:lowerboundsmooth-restatement} readily follows, provided that we are allowed to choose 
\begin{equation}\label{eq:assumption}
    K_\Lambda=m^{\Omega (C/\epsilon)}~.
\end{equation}
This assumption, combined with Eqs.~(\ref{eq:lowerboundsmooth-aux1}-\ref{eq:technical-bound}) ensures that
\begin{equation}
N \frac{9 \epsilon}{4B^2} \log \left( \frac{4 \mathrm{e}B^2}{9 \epsilon} \right)\geq \Omega (K_\Lambda) = m^{\Omega (C/\epsilon)}
\quad \text{which implies} \quad 
N \geq \frac{B^2 m^{\Omega (C/\epsilon)}}{\log (B)}.
\end{equation}
Because the quantum ML has to process quantum training data of size $N \geq \frac{B^2 m^{\Omega (C/\epsilon)}}{\log (B)}$, the runtime of the quantum ML has to be lower bounded by that amount as well.

Let us now verify the assumption (\ref{eq:assumption}) on the number of Fourier modes $K_\Lambda$ available for estimating the quantum state. While we have already determined that $K_\Lambda \leq m^{\mathcal{O}(C/\epsilon)}$ in Eq.~(\ref{eq:numklambda}), here we need a lower bound. We utilize the assumption that $C / (9 \epsilon) \leq m^{0.99}$, which implies $\left\lfloor C / (9 \epsilon) \right\rfloor \leq m^{0.99}$.
To establish Eq.~(\ref{eq:assumption}), we restrict our attention to binary wavevectors $k \in \left\{0,1\right\}^m$, such that the number of ones is exactly equal to $\lfloor C/(9\epsilon)\rfloor$. Clearly, every such wavevector obeys $\norm{k}_2 \leq \sqrt{C/(9\epsilon)}$, so the number of such wavevectors lower bounds $K_\Lambda$. This observation, along with some combinatorics, yields
\begin{subequations}
\begin{align}
K_\Lambda &\geq \Big| \big\{k \in \left\{0,1\right\}^m:\; \sum_{j=1}^m k_j = \lfloor C/(9 \epsilon) \rfloor \big\} \Big|\\
&= \binom{m}{\lfloor C/(9\epsilon) \rfloor} 
\geq  \frac{m^{\lfloor C/(9\epsilon) \rfloor}}{(\lfloor C/9\epsilon \rfloor)^{\lfloor C/(9\epsilon) \rfloor}} \\
&= m^{\lfloor C/(9\epsilon) \rfloor - (\lfloor C/9\epsilon \rfloor)\log(\lfloor C/(9 \epsilon) \rfloor) /\log (m)} 
\geq m^{0.01 \lfloor C/(9\epsilon) \rfloor}=m^{\Omega (C/\epsilon)}.
\end{align}
\end{subequations}

We now prove the upper bound (\ref{eq:technical-bound}) on the mutual information. It follows from analyzing the ground state representations provided by Lemma~\ref{lem:lowerboundsmooth-groundstate}.

\begin{lemma}\label{lem:technical-bound}
The learning problem from Section~\ref{sub:lowerboundsmooth-learning-problem} is set up to obey 
\begin{equation}
\E_{x \sim \mathrm{Unif}[-1,1]^m} \chi \left( a: \rho_{a}(x) \right)
\leq \frac{9 \epsilon}{4B^2} \log \left( \frac{4 \mathrm{e}B^2}{9 \epsilon} \right).
\end{equation}
\end{lemma}

\begin{proof}
Using the definition (\ref{eq:holevo}) of the Holevo information and the von Neumann entropy,
\begin{subequations}
\begin{align} 
    \E_{x \sim \mathrm{Unif}[-1,1]^m} \chi \left( a: \rho_{a}(x) \right) &= \E_x \left[ \E_a [ \Tr( \rho_a(x) \log \rho_a(x) )] - \Tr \left(\left(\E_a \rho_a(x)\right) \log \left(\E_a \rho_a(x)\right) \right) \right]\\
    &= -\E_x \Tr\left[ \left(\E_a \rho_a(x)\right) \log \left(\E_a \rho_a(x)\right) \right]~.
    \end{align}\label{eq:Isxrho}
\end{subequations}
The second equality follows from the fact that $\rho_a(x)$ is a pure state, so we have $\Tr( \rho_a(x) \log \rho_a(x) ) = 0$. We also consider $\E_x$ to be $\E_{x \sim \mathrm{Unif}[-1,1]^m}$.
Recalling Lemma~\ref{lem:lowerboundsmooth-groundstate} yields
\begin{equation}
\E_a \rho_a(x) = \frac{1}{2} \E_a \begin{pmatrix}
1 + \sqrt{1 - (f_a(x) / B)^2} & -\mathrm{i} f_a(x) / B \\
\mathrm{i} f_a(x) / B & 1 - \sqrt{1 - (f_a(x) / B)^2}  \\
\end{pmatrix}.
\end{equation}
The eigenvalues $\lambda_\pm$ of $\E_a \rho_a(x)$, like those of any two-by-two matrix, can be expressed in terms of the trace and determinant. Using the formula for the eigenvalues and evaluating the trace and determinant yield
\begin{subequations}
\begin{align}
    \lambda_{\pm} &= \frac{1}{2}\Tr\left[\E_a \rho_a(x)\right] \pm \frac{1}{2} \sqrt{\left(\Tr \left[\E_a \rho_a(x)\right]\right)^2 - 4 \det\left[\E_a \rho_a(x)\right]}\\
    &= \frac{1}{2} \pm \frac{1}{2} \sqrt{ \left(\E_a f_a(x)\right)^2 / B^2 + \left(\E_a \sqrt{1 - f_a(x)^2 / B^2}\right)^2}~.
\end{align}
\end{subequations}
We will use following lower bound for $\lambda_{+}$
\begin{subequations}
\begin{align}
    \lambda_{+} & \geq \frac{1}{2} + \frac{1}{2} \E_a \sqrt{1 - f_a(x)^2 / B^2} \\
    &\geq \frac{1}{2} + \frac{1}{2}  (1 - \E_a f_a(x)^2 / B^2)\\
    & = 1 - \frac{1}{2} \E_a f_a(x)^2 / B^2 \geq \frac{1}{2}.
\end{align}
\end{subequations}
The first inequality follows from dropping the term $(\E_a f_a(x))^2 / B^2$.
The second inequality follows from the fact that $\sqrt{1 - z} \geq 1-z$ for all $z \in [0,1]$.

We now proceed to bounding the von Neumann entropy of $\E_a \rho_a(x)$,
\begin{subequations}
\begin{align}
    - \Tr\left( \left(\E_a \rho_a(x)\right) \log \left(\E_a \rho_a(x)\right) \right) &= - \lambda_{+} \log \lambda_{+} - \lambda_{-} \log \lambda_{-} = H(\lambda_{+})\\
    & \leq H\left(1 - \frac{1}{2} \E_a f_a(x)^2 / B^2\right)\\
    & = H\left(\frac{1}{2} \E_a f_a(x)^2 / B^2\right) \\
    & \equiv H(g(x)) \leq g(x) \log(\mathrm{e} / g(x))~, 
\end{align}
\label{eq:intrhocalc}
\end{subequations}
where $H(x) = -x \log x - (1-x) \log(1-x)$ is the binary entropy, and $g(x) = \frac{1}{2} \E_a f_a(x)^2 / B^2$.
The first inequality follows from the fact that $H(x) \leq H(x')$
for all $1/2 \leq x' \leq x$.
Going back to Eq.~\eqref{eq:Isxrho}, 
\begin{subequations}
\begin{align}
    \E_{x} \chi \left( a: \rho_{a}(x) | x \right) &= - \E_x  \Tr\left[ \left(\E_a \rho_a(x)\right) \log \left(\E_a \rho_a(x)\right) \right] \\
    & \leq \E_x [g(x) \log(\mathrm{e} / g(x))]\\
    & \leq \left(\E_x g(x)\right) \log\left(\frac{\mathrm{e}}{\E_x g(x)}\right)\\
    & = \frac{\E_{x,a}  f_a(x)^2}{2 B^2} \log\left( \frac{2 \mathrm{e} B^2}{ \E_{x,a}  f_a(x)^2 } \right).\label{eq:finishing}
\end{align}
\end{subequations}
The first inequality follows from Eq.~\eqref{eq:intrhocalc}. The second inequality follows Jensen's inequality and the fact that $z \log(\mathrm{e} / z)$ is concave for all $z \geq 0$.
Orthogonality of the $\cos(\pi k \cdot x)$ terms in $f_a$ \eqref{eq:lowerboundsmooth-function} yields
\begin{equation}
    \E_{x,a} f_a(x)^2 = \frac{1}{2} \times \frac{9 \epsilon}{L} \sum_{k \in \mathbb{Z}^m, \norm{k}_2 \leq \Lambda} \E_a |a_k|^2 = \frac{9 \epsilon}{2}.
\end{equation}
Plugging the above into Eq.~\eqref{eq:finishing}, we obtain the advertised bound.
\end{proof}

\section{Computational hardness for non-ML algorithms to predict ground state properties}
\label{sec:proofhardnonML}

\subsection{NP-hardness for estimating one-body observables in the ground state of 2D Hamiltonians} \label{sec:proofhardnonML1}

We begin by showing that the task of estimating one-body observables in the ground state of any smooth class of two-dimensional Hamiltonians with a constant spectral gap is NP-hard.
The task is hard even if we consider the computation to yield a small error averaged over the smooth class of Hamiltonians.

\begin{figure}
    \centering
    \includegraphics[width=1.0\linewidth]{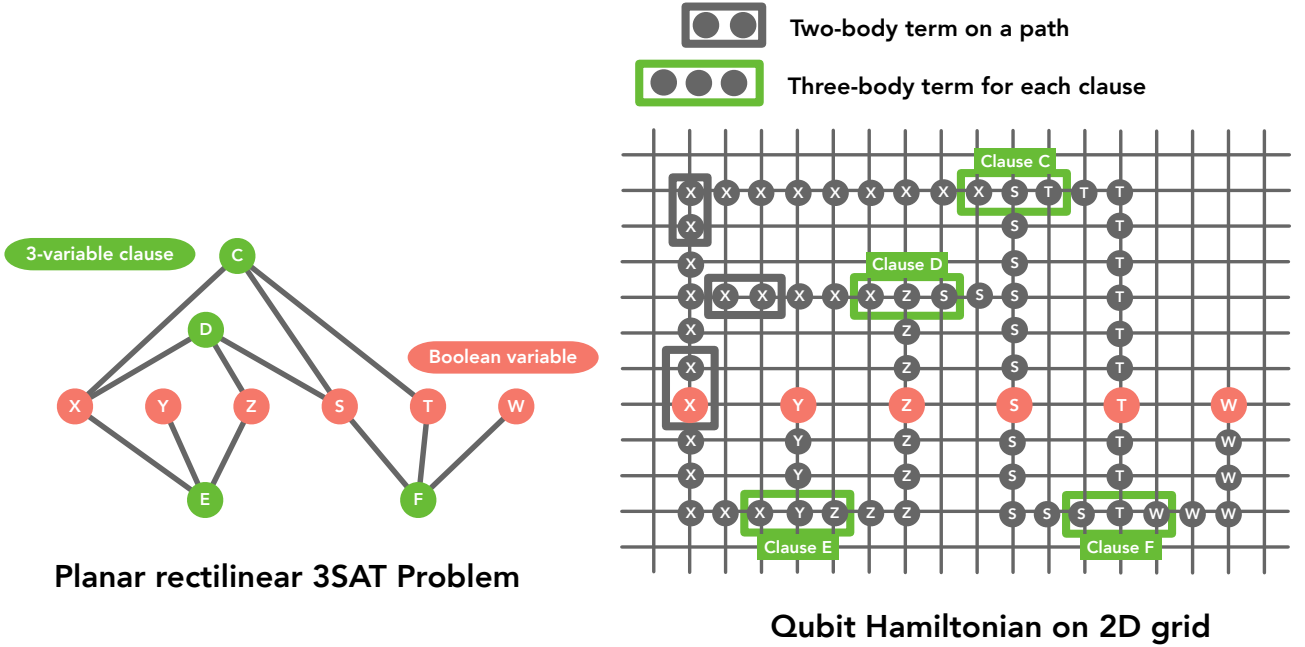}
    \caption{Reduction of planar rectilinear 3SAT (\textsc{left}) to a qubit Hamiltonian on a 2D grid (\textsc{right}). Each pair $(i, j)$ of nearby grid points on a path (originating from variable $X, Y, Z, S, T, W$) contains a two-body local term $-Z_i Z_j$ (illustrated by boxes with gray stroke). Each clause $(C, D, E, F)$ corresponds to a three-body local term that imposes the Boolean constraint, e.g., $X \vee Z \vee S$ would correspond to $-\sum_{x, z, s \in \{0, 1\}} \indicator[x \vee z \vee s = 1] \cdot \ketbra{x}{x} \otimes \ketbra{z}{z} \otimes \ketbra{s}{s}$. Every empty grid point (the irrelevant qubits) contain a single body term $-Z_i$.}
    \label{fig:planarrect3SAT}
\end{figure}

\begin{proposition}[Detailed restatement of Proposition~\ref{thm:hardnonML}; a variant of Lemma 1.4 in \cite{abrahamsen2020sub}] \label{prop:phaseclassification-app}
Consider a randomized polynomial-time classical algorithm $\mathcal{A}(H, i, r)$ whose inputs are the description of a Hamiltonian $H$, an index $i$ that enumerates the qubits in the Hamiltonian, and a random bit string $r$.
Suppose that for any smooth class of Hamiltonians on a two-dimensional grid with a spectral gap $\geq 1$ and a unique ground state,
\begin{equation}
    H(x) = \sum_{a} h_a(x) \,\,\, \mathrm{with} \,\, \rho(x) \,\, \mbox{: the ground state of} \,\, H(x),
\end{equation}
where $x \in [-1, 1]^m$ is a parameter and $h_a(x)$ is a three-qubit geometrically-local observable, and for each one-body Pauli-Z observable $Z_i$, the randomized classical algorithm $\mathcal{A}$ outputs $\mathcal{A}(H, i, r)$ that approximates $\Tr(Z_i \rho(x))$ up to an average error $\E_{x \sim [-1, 1]^m} \left| \E_r \mathcal{A}(H, i, r) - \Tr(Z_i \rho(x)) \right| \leq 1/4.$ Then $\mathrm{RP} = \mathrm{NP}$.
\end{proposition}
\begin{proof}
From standard results in complexity theory \cite{Lichtenstein1982PlanarFA, knuth1992problem, VALIANT198685}, it is 
known that if there is a randomized polynomial-time classical algorithm that can find the solution for any planar rectilinear 3SAT problem with a unique solution with probability at least $1/2$, then $\mathrm{RP} = \mathrm{NP}$. ($\mathrm{RP}$, also known as Randomized Polynomial Time, is the class of decision problems such that there is a polynomial-time randomized classical algorithm that outputs YES with probability at least 1/2 when the correct answer is YES, and outputs NO with probability one when the correct answer is NO. RP is contained in BPP, the class of decision problems that can be solved efficiently by a randomized classical computer.)
The planar rectilinear 3SAT problem is a constrained version of 3SAT, where all the Boolean variables $x_1,\ldots, x_n$ are vertices on the $x$-axis and all the clauses containing three variables are vertices that lie above or below the $x$-axis.
Each clause is connected by an edge to each of the the variables that the clause contains.
The vertices and the edges form a planar graph; see Figure~\ref{fig:planarrect3SAT} (left) for an illustration.

We can embed such a planar graph in a two-dimensional grid with a single qubit on each grid point; see Figure~\ref{fig:planarrect3SAT} (right) for an illustration of the embedding.
First, we distinguish between the variable vertices and the clause vertices in the planar graph.
Variable vertices lie on the $x$-axis of the two-dimensional qubit grid, and clause vertices lie above or below the $x$-axis. 
Edges of the planar graph become embedded paths on the the 2D grid connecting clause vertices to variable vertices. 
Because the original graph is planar, we can ensure that the paths corresponding to each edge on the planar graph do not overlap (except when they terminate at the same variable) by choosing a large enough spacing between the variable vertices on the $x$-axis.
For each path on the 2D grid, we add a $-Z_i Z_j$ term to the Hamiltonian for every pair of nearest neighbors along the path.
The two body $-Z_i Z_j$ term ensures that, in the unique ground state, the qubits on the path must be either all $\ket{0}$'s or all $\ket{1}$'s.
Then, for every clause vertex on the planar graph, we add a three-body geometrically-local term (diagonal in the $Z$-basis) to the Hamiltonian enforcing that in the ground state the endpoints of the three corresponding paths satisfy the Boolean constraint of the corresponding clause.
For example, the Boolean clause $X \vee Z \vee S$ would correspond to the three body local term $-\sum_{x, z, s \in \{0, 1\}} \indicator[x \vee z \vee s = 1] \cdot \ketbra{x}{x} \otimes \ketbra{z}{z} \otimes \ketbra{s}{s}$, where $\indicator[A]$ is $1$ if $A$ is true and $0$ otherwise.
The qubits on paths are called the ``relevant'' qubits, and the rest of the qubits are called ``irrelevant.'' 
We add a $-Z_i$ term to the Hamiltonian for all the irrelevant qubits, fixing these qubits to be $|0\rangle$ in the ground state. 

Moreover, the eigenstates of the Hamiltonian are computational basis states, because all the local terms are diagonal in the $Z$-basis.
We can also see that there are no terms connecting the relevant and irrelevant qubits, hence the ground state space of the constructed Hamiltonian must be the tensor product of the ground state space for the relevant qubits and the ground state space for the irrelevant qubits.
The unique ground state for the irrelevant qubits is the all-zero state $\ket{0} \otimes\dots\otimes \ket{0}$ due to the $-Z_i$ term.
Because the original planar rectilinear 3SAT problem has a unique solution, the ground state for the relevant qubits is also unique.
We denote the ground state by $\ketbra{b}{b}, b \in \{0, 1\}^n,$ where $n$ is the total number of qubits in the two dimensional grid.
In this ground state, all variable vertices are fixed at the values that solve the 3SAT problem.
Furthermore, because all eigenvalues of the Hamiltonian are integers, the spectral gap is at least one.

Let us define $\sum_a h_a$ to be the Hamiltonian constructed from a planar rectilinear 3SAT problem.
Note that $h_a$ is diagonal in the $Z$-basis and acts on at most three geometrically-local qubits.
We define a trivial class of two-dimensional Hamiltonians with a spectral gap $\geq 1$,
\begin{equation}
    H(x) = \sum_{a} h_a(x) = \sum_{a} h_a = H,
\end{equation}
where $x \in [-1, 1]^m$ is the parameter, and $H(x)$ does not depend on $x$.
Let $\rho(x)$ be the unique ground state of $H(x)$.
We have $\rho(x) = \ketbra{b}{b}, b \in \{0, 1\}^n,$ where $b$ encodes the solution to the planar rectilinear 3SAT problem.

We apply the randomized classical algorithm to provide estimates for all the expectation values of Pauli-$Z$ observables in the ground state space $\rho(x)$ of $H(x)$.
Let $\mathcal{A}$ be the randomized classical algorithm.
By the assumption that the randomized classical algorithm could output an estimate of $\Tr(Z_i \rho(x))$ up to an additive error $1/4$ averaged uniformly over $x \in [-1, 1]^m$, we have
\begin{equation}
    \E_{x \sim [-1, 1]^m} \left| \E_r \mathcal{A}(H(x), i, r) - \Tr(Z_i \rho(x)) \right| \leq 1/4, \quad \forall i = 1, \ldots, n.
\end{equation}
Using Jensen's inequality, we have the following bound,
\begin{equation}
    \left| \E_{x \sim [-1, 1]^m} \E_r \mathcal{A}(H(x), i, r) - \E_{x \sim [-1, 1]^m} \Tr(Z_i \rho(x)) \right| \leq 1/4, \quad \forall i = 1, \ldots, n.
\end{equation}
We can see that $\E_{x \sim [-1, 1]^m} \Tr(Z_i \rho(x)) = \bra{b_i} Z_i \ket{b_i},$
where $b_i$ is the $i$-th bit in the $n$-bit string $b$ that encodes the solution to the planar rectilinear 3SAT problem.

We sample random $x$ uniformly from $[-1, 1]^m$ and sample the random string $r$, obtaining the output value $\mathcal{A}(H(x), i, r)$ using the randomized classical algorithm $\mathcal{A}$.
As a result of the above analysis, by sampling $\mathcal{O}(\log(n))$ times and computing the average over the output $\mathcal{A}(H(x), i, r)$, we can obtain an estimate for $\bra{b_i} Z_i \ket{b_i}$ up to an additive error $1/2$ with probability at least $1 - \tfrac{1}{2n}$, where $n$ is the total number of qubits in the 2D grid.
Because $\bra{b_i} Z_i \ket{b_i} \in \{-1, 1\}$, an estimate for $\bra{b_i} Z_i \ket{b_i}$ up to an additive error $1/2$ allows us to obtain $b_i \in \{0, 1\}$.
Using the union bound, with probability at least $1/2$, we can obtain $b_i$, for all $i = 1, \ldots, n$.  This implies that we can obtain the bit string $b$ with probability at least $1/2$.
Hence, we can use the randomized classical algorithm $\mathcal{A}$ to find the unique solution for the planar rectilinear 3SAT problem with probability at least $1/2$. Therefore, RP$=$NP if such an algorithm exists.
\end{proof}

We remark that a similar argument still applies if we replace the constant Hamiltonian $H(x)$ considered above by a suitably chosen class of Hamiltonians $\{H(x) = H : x \in [-1, 1]^m\}$ with nontrivial dependence on $x$. 
For example, we can consider $H(x) = \sum_a h_a(x) = \sum_a \left(U_1(x_1) \otimes \ldots \otimes U_n(x_n) \right) h_a \left(U_1(x_1) \otimes \ldots \otimes U_n(x_n) \right)^\dagger$, where $n$ is the number of qubits in the Hamiltonian, $U_i(x_i) = \exp(- \mathrm{i} (\pi / 4) X_i x_i)$ is a single-qubit rotation, $X_i$ is the Pauli-$X$ matrix on the $i$-th qubit, and $m = n$.
It is not hard to see that $h_a(x)$ still acts on at most three geometrically-local qubits.
Furthermore, one can adapt the proof to show that predicting ground state properties averaged over $x$ for this nonconstant class of Hamiltonians is still hard.

\subsection{Computational hardness for a class of Hamiltonians based on factoring}
\label{sec:factoringH}

Theorem~\ref{thm:mainFourier} and Proposition \ref{prop:phaseclassification-app} together implies that an NP-hard problem could be solved by performing single-qubit measurements on a modest number of copies of the ground state of a two-dimensional local Hamiltonian, and then performing an efficient classical computation with the measurement outcomes as input. We may therefore conclude that, in hard instances, the preparation of the ground state is itself an NP-hard task. Because we do not expect any NP-hard task to be performed efficiently in the physics lab, or in any other physically realizable process, Proposition \ref{prop:phaseclassification-app} does not usefully characterize the computational power of data under realistic conditions. 

In contrast, it is reasonable to expect that simple measurements performed on quantum states that are efficiently prepared by quantum computers, combined with classical processing, suffice for solving computational problems that are beyond the reach of classical processing alone. Indeed, proposals for using variation quantum eigensolvers to study quantum chemistry and materials \cite{peruzzo2014variational, mcclean2016theory} are motivated by this expectation. Theorem~\ref{thm:mainFourier} is of potential practical interest for a class of Hamiltonians $\{H(x)\}$ such that the ground state of $H(x)$ can be prepared efficiently by a feasible quantum process, yet cannot be efficiently prepared classically.

The rest of this subsection outlines a stylized example that illustrates this idea. Leveraging the efficient quantum algorithm for factoring large numbers, and the assumption that factoring is classically hard, we construct a smooth class of local Hamiltonians whose ground states are easy to prepare quantumly, such that expectation values of one-local observables can be learned efficiently from training data, yet are hard to learn by any classical procedure without access to data. 

The first step is to construct two-dimensional Hamiltonians such that computing expectation values of one-local observables in the ground state is equivalent to solving a factoring problem. This can be done by noting a series of well-known facts in complexity theory.
\begin{enumerate}
    \item The following task is expected to be hard for classical computers. Given a $n$-bit number $R$ guaranteed to be a product of two prime numbers $p < q$, find $p, q$. When $R$ is large, all known classical algorithms scale superpolynomially with $n$. Solving this problem suffices to break the RSA encryption \cite{rivest1978method}.
    \item We can 
    represent $p, q$ using at most $2n$ binary variables (bits), and  we can write down a propositional formula for these $2n$ variables, which corresponds to a logical circuit 
    that computes the multiplication of $p, q$ and checks if the product equals $R$.
    The propositional formula can be written without any additional Boolean variable.
    This yields a SAT problem with $2n$ Boolean variables whose unique solution is equal to the two prime numbers $p, q$.
    \item A SAT problem with a unique solution can be efficiently mapped to a 3SAT problem with a unique solution; see \cite{kozen1992design}.
    \item A 3SAT problem with a unique solution can be efficiently mapped to a planar rectilinear 3SAT problem with a unique solution; see \cite{Lichtenstein1982PlanarFA,knuth1992problem}.
    \item A planar rectilinear 3SAT problem with a unique solution can be efficiently mapped to a two-dimensional 3-local Hamiltonian with a spectral gap of one and a unique ground state, such that estimating one-local observables in the ground state of the Hamiltonian to a constant error with a constant probability is sufficient to find the unique solution for the planar rectilinear 3SAT problem; see the proof of Proposition~\ref{prop:phaseclassification-app}.
\end{enumerate}
We now focus on any smooth class of two-dimensional Hamiltonians $H^{\mathrm{RSA}}(x)$ with a constant spectral gap such that there exists $x^{\mathrm{RSA}} \in [-1, 1]^m$ such that $H^{\mathrm{RSA}}(x^{\mathrm{RSA}})$ can be written as a two-dimensional Hamiltonian that is mapped from a factoring problem. We refer to such a class of Hamiltonians as an RSA-based two-dimensional gapped Hamiltonian class.

For any RSA-based Hamiltonian class $H^{\mathrm{RSA}}$, we can efficiently obtain the training data from a quantum experiment. We first prepare the ground state for $H^{\mathrm{RSA}}(x^{\mathrm{RSA}})$ by applying Shor's algorithm. Then we can adiabatically evolve the ground state for $H^{\mathrm{RSA}}(x^{\mathrm{RSA}})$ to obtain the ground state for $H^{\mathrm{RSA}}(x), \forall x \in [-1, 1]$ due to the existence of a constant spectral gap \cite{aharonov2008adiabatic, wan2020fast}.
Hence, according to Theorem~\ref{thm:mainFourier}, for any RSA-based Hamiltonian class, a classical ML algorithm trained from data obtained in quantum experiments can predict efficiently expectation values of one-local observables in the ground state.
In contrast, a classical algorithm that does not learn from training data is unable to efficiently estimate 1-body observables in the ground state, assuming that RSA encryption can not be broken by classical computers.

\section{No observable can classify topological phases}
\label{sec:no-nonlocal-obs}

Recall that ground states of two Hamiltonians are in the same topological phase if there exists a constant-depth geometrically-local quantum circuit that can transform one state to another \cite{zeng2019quantum}. The goal of this section is to establish the following proposition.

\begin{proposition}
Consider two distinct topological phases $A$ and $B$ (one of the phases could be the trivial phase). No observable $O$ exists such that
\begin{equation}
    \Tr(O \rho) > 0, \forall \rho \in \mbox{phase $A$}, \quad\quad \Tr(O \rho) \leq 0, \forall \rho \in \mbox{phase $B$}.
\end{equation}
\end{proposition}
\begin{proof}
We consider depth-$1$ quantum circuits consisting of single-qubit unitaries $U_1, \ldots, U_n$.
We let $\ket{\psi_A}, \ket{\psi_B}$ be the signature quantum state for phase $A$ and $B$.
Suppose there is an observable such that
\begin{equation}
    \Tr(O \rho) > 0, \forall \rho \in \mbox{phase $A$}, \quad\quad \Tr(O \rho) \leq 0, \forall \rho \in \mbox{phase $B$}.
\end{equation}
Then, by definition, we have
\begin{subequations}
\begin{align}
    \bra{\psi_A} (U_1^\dagger \otimes \ldots \otimes U_n^\dagger) O (U_1 \otimes \ldots \otimes U_n)\ket{\psi_A} & > 0, \forall U_1, \ldots, U_n \in U(2), \\
    \bra{\psi_B} (U_1^\dagger \otimes \ldots \otimes U_n^\dagger) O (U_1 \otimes \ldots \otimes U_n)\ket{\psi_B} & \leq 0, \forall U_1, \ldots, U_n \in U(2),
\end{align}
\end{subequations}
However, from Lemma~\ref{lem:noobserU1...Un}, no such observable exists. Hence no observable exists that can be used to classify two topologically ordered phases.
\end{proof}

The key lemma utilized in the above proof is the following.

\begin{lemma} \label{lem:noobserU1...Un}
For any two $n$-qubit states $\ket{\psi_1}, \ket{\psi_2}$, no observable $O$ exists such that
\begin{subequations}
\begin{align}
    \bra{\psi_1} (U_1^\dagger \otimes \ldots \otimes U_n^\dagger) O (U_1 \otimes \ldots \otimes U_n)\ket{\psi_1} & > 0, \forall U_1, \ldots, U_n \in U(2), \label{eq:psi1condition} \\
    \bra{\psi_2} (U_1^\dagger \otimes \ldots \otimes U_n^\dagger) O (U_1 \otimes \ldots \otimes U_n)\ket{\psi_2} & \leq 0, \forall U_1, \ldots, U_n \in U(2), \label{eq:psi2condition}
\end{align}
\end{subequations}
where $U(2)$ is the unitary group of $2\times 2$ unitary matrices.
\end{lemma}
\begin{proof}
We will prove this result by contradiction. Assume the existence of an observable $O$ such that Eq.~\eqref{eq:psi1condition}~and~\eqref{eq:psi2condition} both hold.
Consider $U_1, \ldots, U_n$ to be independent random matrices that follows the Haar measure on the unitary group $U(2)$. Then using the identity for the first order moment of Haar integration,
\begin{equation}
    \E _{U \sim \mathrm{Haar}(U(d))} U X U^\dagger = \Tr(X) \frac{\mathbb{I}}{d},
\end{equation}
we can obtain the following identity,
\begin{equation} \label{eq:depolarizing}
    \E _{U_1, \ldots, U_n \sim \mathrm{Haar}(U(2))} \left[ (U_1 \otimes \ldots \otimes U_n)\ketbra{\psi_1}{\psi_1}(U_1^\dagger \otimes \ldots \otimes U_n^\dagger) \right] = \Tr( \ketbra{\psi_1}{\psi_1} ) \frac{\mathbb{I}}{2^n} = \frac{\mathbb{I}}{2^n}.
\end{equation}
The key property is the compactness of the unitary group $U(2)$.
Consider the following infimum,
\begin{equation}
    o_1 = \inf_{U_1, \ldots, U_n \in U(2)} \bra{\psi_1} (U_1^\dagger \otimes \ldots \otimes U_n^\dagger) O (U_1 \otimes \ldots \otimes U_n)\ket{\psi_1}.
\end{equation}
Because the infimum is always attained by an element in the compact set, $\exists U^{\inf}_1, \ldots, U^{\inf}_n \in U(2)$ such that
\begin{equation}
    o_1 = \bra{\psi_1} ((U^{\inf}_1)^\dagger \otimes \ldots \otimes (U^{\inf}_n)^\dagger) O (U^{\inf}_1 \otimes \ldots \otimes U^{\inf}_n)\ket{\psi_1}.
\end{equation}
Therefore, we have $o_1 > 0$ from Eq.~\eqref{eq:psi1condition}.
Using the property of infimum, we have
\begin{equation}
    \bra{\psi_1} (U_1^\dagger \otimes \ldots \otimes U_n^\dagger) O (U_1 \otimes \ldots \otimes U_n)\ket{\psi_1} \geq o_1, \forall U_1, \ldots, U_n \in U(2),
\end{equation}
we have the following inequality,
\begin{equation}
    \E _{U_1, \ldots, U_n \sim \mathrm{Haar}(U(2))} \bra{\psi_1} (U_1^\dagger \otimes \ldots \otimes U_n^\dagger) O (U_1 \otimes \ldots \otimes U_n)\ket{\psi_1} \geq o_1 > 0.
\end{equation}
By the linearity of expectation and Eq.~\eqref{eq:depolarizing}, we have
\begin{align}
    &\E _{U_1, \ldots, U_n \sim \mathrm{Haar}(U(2))} \bra{\psi_1} (U_1^\dagger \otimes \ldots \otimes U_n^\dagger) O (U_1 \otimes \ldots \otimes U_n)\ket{\psi_1}\\
    &= \Tr\left(O \E _{U_1, \ldots, U_n \sim \mathrm{Haar}(U(2))} \left[ (U_1 \otimes \ldots \otimes U_n)\ketbra{\psi_1}{\psi_1}(U_1^\dagger \otimes \ldots \otimes U_n^\dagger) \right] \right) = \frac{\Tr(O)}{2^n}.
    \nonumber
\end{align}
Together, we have
\begin{equation} \label{eq:Oleqo1}
    \frac{\Tr(O)}{2^n} \geq o_1 > 0.
\end{equation}
The argument for $\ket{\psi_2}$ is slightly simpler. Consider the following supremum,
\begin{equation}
    o_2 = \sup_{U_1, \ldots, U_n \in U(2)} \bra{\psi_2} (U_1^\dagger \otimes \ldots \otimes U_n^\dagger) O (U_1 \otimes \ldots \otimes U_n)\ket{\psi_2}.
\end{equation}
From Eq.~\eqref{eq:psi2condition}, we have $o_2 \leq 0$.
Using the fact that
\begin{equation}
    \bra{\psi_2} (U_1^\dagger \otimes \ldots \otimes U_n^\dagger) O (U_1 \otimes \ldots \otimes U_n)\ket{\psi_2} \leq o_2, \forall U_1, \ldots, U_n \in U(2),
\end{equation}
we have the following inequality,
\begin{equation}
    \E _{U_1, \ldots, U_n \sim \mathrm{Haar}(U(2))} \bra{\psi_2} (U_1^\dagger \otimes \ldots \otimes U_n^\dagger) O (U_1 \otimes \ldots \otimes U_n)\ket{\psi_2} \leq o_2 \leq 0.
\end{equation}
By the linearity of expectation and Eq.~\eqref{eq:depolarizing}, we have
\begin{align}
    &\E _{U_1, \ldots, U_n \sim \mathrm{Haar}(U(2))} \bra{\psi_2} (U_1^\dagger \otimes \ldots \otimes U_n^\dagger) O (U_1 \otimes \ldots \otimes U_n)\ket{\psi_2}\\
    &= \Tr\left(O \E _{U_1, \ldots, U_n \sim \mathrm{Haar}(U(2))} \left[ (U_1 \otimes \ldots \otimes U_n)\ketbra{\psi_2}{\psi_2}(U_1^\dagger \otimes \ldots \otimes U_n^\dagger) \right] \right) = \frac{\Tr(O)}{2^n}.
    \nonumber
\end{align}
Together, we have
\begin{equation} \label{eq:Oleqo2}
    \frac{\Tr(O)}{2^n} \leq o_2 \leq 0.
\end{equation}
From Eq.~\eqref{eq:Oleqo1}~and~\eqref{eq:Oleqo2}, we have derived the following result
\begin{equation}
    \frac{\Tr(O)}{2^n} \leq o_2 \leq 0 < o_1 \leq \frac{\Tr(O)}{2^n},
\end{equation}
which is a contradiction. Therefore, no such observable $O$ exists.
\end{proof}

\section{Proof of efficiency for classifying phases of matter} \label{sec:proofthmPHASEC}

This section contains a detailed proof for another one of our main contributions. Namely, a rigorous performance guarantee for learning to predict quantum phases of matter.

\subsection{Training support vector machines}
\label{sec:trainSVM}

Let us start by reviewing the textbook framework for reasoning about supervised learning tasks: support vector machines (SVMs). 
The underlying idea is simple and intuitive. Suppose that we have $N$ data points $\vct{x}_\ell \in \mathbb{R}^D$ with binary labels $y_\ell \in \left\{ \pm 1 \right\}$ that form two well separated clusters. Then, we may try to separate these training clusters with a linear hyperplane $\mathsf{H}_{\vct{\alpha}} = \left\{ \vct{x} \in \mathbb{R}^D:\; \langle \vct{\alpha},\vct{x} \rangle = 0 \right\} \subset \mathbb{R}^D$, defined using any vector $\vct{\alpha}$ that is perpendicular to all vectors in the hyperplane. Here, we implicitly assume that the hyperplane $\mathsf{H}_{\vct{\alpha}}$ must contain the origin $\vct{0} \in \mathbb{R}^D$. This simplifies exposition and will suffice for our purposes, but also constitutes an actual restriction (linear SVMs typically also allow for affine shifts). Such a hyperplane divides $\mathbb{R}^D$ up into two half-spaces. For linear classification, we want that these half-spaces perfectly capture the labels of training data: $\langle \vct{\alpha},\vct{x}_\ell \rangle > 0$ whenever $y_\ell=+1$ and $\langle \vct{\alpha},\vct{x}_\ell \rangle <0$ whenever $y_\ell= -1$. The hope is that this simple linear classification strategy generalizes to data we haven't yet seen. 
When we get a new data point, we simply check which halfspace it belongs to and assign the label accordingly.
In the training stage, the main question is: how do we find a suitable hyperplane?
Several strategies are known in the literature. One of them is the \emph{soft margin} problem:
\begin{subequations}
\begin{align}
\underset{\vct{\alpha} \in \mathbb{R}^D}{\text{minimize}} & \quad 
\sum_{\ell=1}^N \max \left\{0,1-y_\ell\langle \vct{\alpha},\vct{x}_\ell \rangle \right\} \\
\text{subject to} & \quad \langle \vct{\alpha},\vct{\alpha} \rangle \leq \Lambda^2. 
\end{align}
\end{subequations}
For both label values, a positive product $y_\ell\langle \vct{\alpha},\vct{x}_\ell \rangle$ is theoretically sufficient. However, numerical precision considerations warrant a nonzero separation between the clusters, so the product is optimized to be at least as large as a positive number (here, $1$). Otherwise, a hyperplane defined by $\vct{\alpha}$ does not perfectly classify the data, yielding the training error $\mathrm{E}_{\mathrm{tr}}(\vct{\alpha})=\sum_{\ell=1}^N \max \left\{0,1-y_\ell\langle \vct{\alpha},\vct{x}_\ell \rangle \right\}$. 
The task is to find $\vct{\alpha}_\sharp$ that achieves the smallest training error: $\mathrm{E}_{\mathrm{tr}}(\vct{\alpha}_\sharp) \leq \mathrm{E}_{\mathrm{tr}}(\vct{\alpha})$ for all vectors that obey $\langle \vct{\alpha},\vct{\alpha} \rangle \leq \Lambda^2$.
This is a convex optimization problem that can be solved in 
polynomial time and we refer to Figure~\ref{fig:SVM-illustration} for a visual illustration. 
The most interesting situation occurs if we manage to achieve an optimal objective value of 0. This corresponds to zero training error. In this case, we have found a hyperplane $\mathsf{H}_{\vct{\alpha}_\sharp}$ that perfectly separates training data. What is more, the constraint $\langle \vct{\alpha}_\sharp, \vct{\alpha}_\sharp \rangle \leq \Lambda^2$ ensures that the margin of separation is strictly positive. Let $\hat{\vct{\alpha}} = \vct{\alpha}/\| \vct{\alpha} \|$ be the unit vector that characterizes a hyperplane. Then, zero training error implies $\langle \hat{\vct{\alpha}},\vct{x}_\ell \rangle \geq 1/ \| \vct{\alpha} \| \geq  1/\Lambda$ for all $\vct{x}_\ell$ with $y_\ell=+1$ and $\langle \hat{\vct{\alpha}},\vct{x}_\ell \rangle < - 1/\Lambda$ for all $\vct{x}_\ell$ with $y_\ell=-1$. In turn, the minimal margin amounts to $2/\Lambda$.

\begin{figure}
    \centering
    \includegraphics[height=0.3\textwidth]{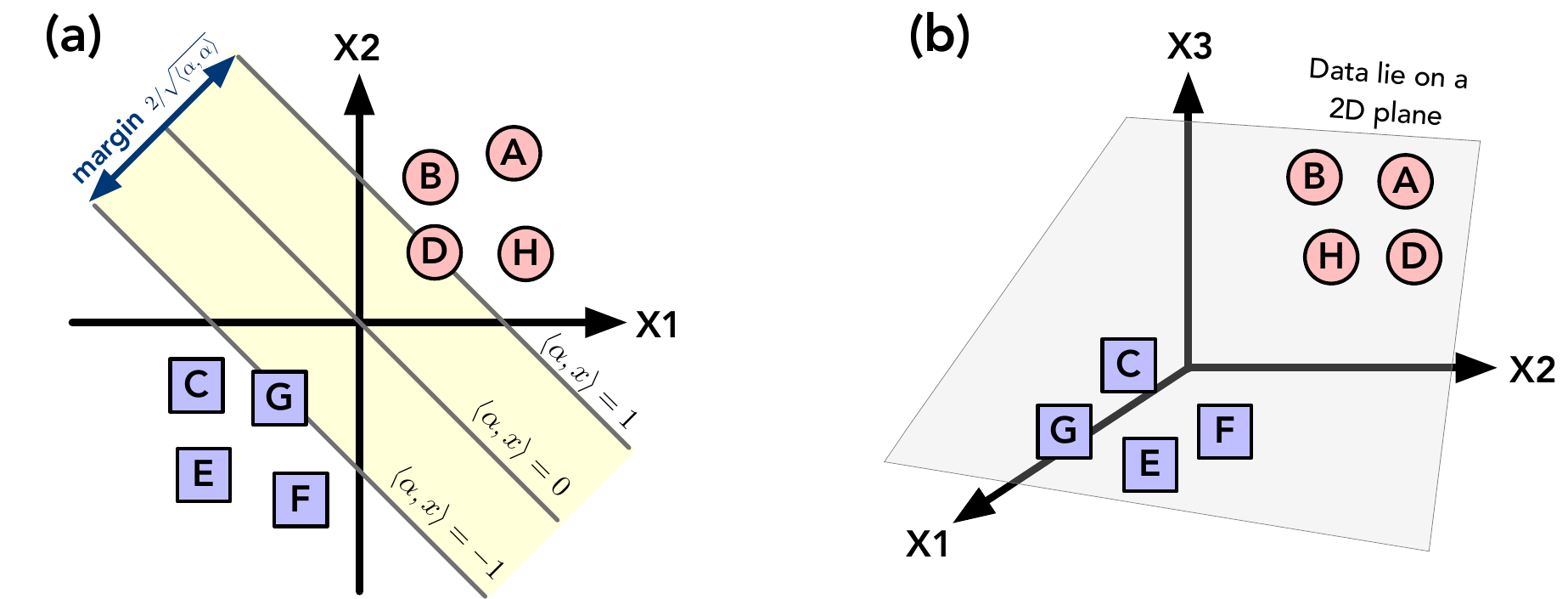}
    \caption{(a) \textsc{Geometric intuition behind support vector machines (SVMs)}. 
     The idea is to separate clusters of labeled data with a linear hyperplane. The separation margin (yellow) is inversely proportional to the length $\sqrt{\langle \alpha,\alpha\rangle}$ of the hyperplane normal vector. 
     During the training stage we try to find a hyperplane that separates points with label +1 (blue) from points with label -1 (red) such that the margin is as large as possible (left). This hyperplane separates the data space into two halfspaces. In order to predict the label of a new data point, we simply check which halfspace it belongs to. (b) \textsc{Geometric intuition behind the representer theorem}. When trying to find a separating hyperplane, the total dimension of the data space does not matter. We can without loss restrict our attention to the smallest subspace that contains all the data points. This is because orthogonal directions don't matter during training and has two implications: (i) the cost of finding a separating hyperplane depends on the training data size $N$, not feature space dimension and (ii) we can express the hyperplane vector as a linear combination of training data points.  
    }
    \label{fig:SVM-illustration}
\end{figure}

However, it should not come as a surprise that such linear classification strategies are often inadequate. Most labeled collections of data simply cannot be separated by a linear hyperplane. However, it has been observed that this drawback can be overcome by first transforming data into a (usually much larger) feature space $\vct{x}_\ell \mapsto \phi (\vct{x}_\ell)$
and trying to find a separating hyperplane there.
This transformation is typically nonlinear and increases the expressiveness of hyperplane classification. Although the separating hyperplane is linear in feature space, it may be highly nonlinear in the original data space. Denote the feature space by $\mathcal{F}$ and suppose that it comes with an inner product $\langle \cdot, \cdot \rangle_{\mathcal{F}}$ and dual space $\mathcal{F}^*$. We can then formally phrase the search for a linear classifier in feature space as
\begin{subequations}
\begin{align}
\underset{\vct{\alpha} \in \mathcal{F}^*}{\text{minimize}} & \quad  
\sum_{\ell=1}^N \max \left\{0,1-y_\ell\langle \vct{\alpha},\phi (\vct{x}_\ell) \rangle_{\mathcal{F}} \right\} \label{eq:SVM-classifier}\\
\text{subject to} & \quad \langle \vct{\alpha},\vct{\alpha} \rangle_{\mathcal{F}} \leq \Lambda^2. 
\end{align}
\end{subequations}
This problem looks more daunting than its linear counterpart, especially because the feature space $\mathcal{F}$ may have an exceedingly large -- perhaps even infinite -- dimension. But we are still interested in identifying a hyperplane that separates a total of $N$ transformed data points $\phi (\vct{x}_1),\ldots,\phi(\vct{x}_N) \in \mathcal{F}$ in a linear fashion: $\langle \vct{\alpha}, \phi (\vct{x}_\ell) \rangle_{\mathcal{F}} >0$ if $y_\ell=+1$ and $\langle \vct{\alpha},\phi (\vct{x}_\ell) \rangle_{\mathcal{F}} <0$ else if $y_\ell = -1$. 
And in order to achieve this, we can without loss restrict ourselves to the $N$-dimensional subspace
$\mathrm{span} \left\{ \phi (\vct{x}_1),\ldots,\phi (\vct{x}_N) \right\} \subset \mathcal{F}$  that is spanned by the data points themselves
(all other directions are orthogonal to \emph{all} data points and do not play a role for classification). For finite dimensional feature spaces $(\mathcal{F},\langle \cdot,\cdot \rangle_{\mathcal{F}})$, this is an intuitive observation that follows from basic orthogonality arguments. 
We refer to Figure~\ref{fig:SVM-illustration} for a visual illustration.
For infinite-dimensional feature spaces it is the content of the celebrated generalized representer theorem~\cite{scholkopf2001generalized}.
More formally, this insight allows us to decompose every (relevant) hyperplane normal vector $\vct{\alpha}$ in the optimization problem~\eqref{eq:SVM-classifier} as $\vct{\alpha} = \sum_{\ell=1}^N \alpha_\ell \phi (\vct{x}_\ell)$. 
Linearity then ensures $\langle \vct{\alpha}, \phi (\vct{x}_{\ell'}) \rangle_{\mathcal{F}} = \sum_{\ell=1}^N \alpha_\ell \langle \phi (\vct{x}_\ell),\phi (\vct{x}_{\ell'}) \rangle_{\mathcal{F}}$ for each $\ell' \in \left\{1,\ldots,N\right\}$ and also $\langle \vct{\alpha},\vct{\alpha} \rangle_{\mathcal{F}} = \sum_{\ell,\ell'=1}^N \alpha_\ell \alpha_{\ell'} \langle \phi (\vct{x}_\ell), \phi (\vct{x}_{\ell'}) \rangle_{\mathcal{F}}$.
These expressions only depend on the elements of a $N \times N$ Gram matrix 
in feature space:
\begin{equation}
\left[ \mtx{K}\right]_{ \ell \ell'}= \langle \phi (\vct{x}_\ell), \phi (\vct{x}_{\ell'}) \rangle_{\mathcal{F}}=: k \left( \vct{x}_\ell,\vct{x}_{\ell'} \right) \quad \text{for $\ell,\ell' \in \left\{1,\ldots,N\right\}$.}
\label{eq:kernel-matrix}
\end{equation}
The expression $k(\vct{x}_\ell,\vct{x}_{\ell'})$ is called the \emph{kernel} associated with the feature map $\phi$ and the matrix $\mtx{K}$ is the \emph{kernel matrix}. Kernels are a measure of similarity between (training) data points that is often easier to compute than performing the underlying feature map $\phi: \mathbb{R}^D \to \mathcal{F}$. But, for linear classification (in feature space), both contain exactly the same amount of information. Indeed, we may re-express the optimization problem~\eqref{eq:SVM-classifier} as
\begin{subequations}
\begin{align}
\underset{\vct{\alpha} \in \mathbb{R}^N}{\text{minimize}} & \quad 
\sum_{\ell=1}^N \max \left\{0,1 - y_{\ell}\vct{\alpha}^T \mtx{K} \vct{e}_{\ell} \right\} \label{eq:SVM-classifier-2}\\
\text{subject to} & \quad \vct{\alpha}^T \mtx{K} \vct{\alpha} \leq \Lambda^2. 
\end{align}
\end{subequations}
We can also collect the classification labels in a diagonal matrix $\mtx{Y}=\mathrm{diag} \left( y_1,\ldots,y_N\right)$ of compatible dimension and linearize the loss function by means of an entry-wise nonnegative slack variable $\vct{\beta} \geq \vct{0}$. Let $\vct{1}=(1,\ldots,1)^T$ denote the vector of ones. Then, problem~\eqref{eq:SVM-classifier-2} is equivalent to solving
\begin{subequations}
\begin{align}
\underset{\vct{\alpha},\vct{\beta} \in \mathbb{R}^N}{\text{minimize}} & \quad \langle \vct{1}, \vct{\beta} \rangle \label{eq:training-problem} \\
\text{subject to} & \quad \vct{\beta} \geq \vct{1} - \mtx{Y} \mtx{K} \vct{\alpha} \\
& \quad \vct{\beta} \geq \vct{0},\; \vct{\alpha}^* \mtx{K} \vct{\alpha} \leq \Lambda^2. 
\end{align}
\end{subequations}
Similar to before, the optimal function value denotes the minimal \emph{training error} $\mathrm{E}_{\mathrm{tr}} (\vct{\alpha}_\sharp)=\langle \vct{1},\vct{\beta}_\sharp \rangle$.
Apart from a single quadratic constraint ($\vct{\alpha}^* \mtx{K} \vct{\alpha} \leq \Lambda^2$), this optimization problem looks like a linear program in $2N$ dimensions. It is a convex instance of a quadratically constrained quadratic program (QCQP) and can be solved in time at most polynomial 
in the training data size $N$ \cite{boyd2004convex}. In practice, one could use existing software packages, such as scikit-learn \cite{scikit-learn} or LIBSVM \cite{CC01a}.
If the time to compute the kernel function $k(\vct{x}_\ell, \vct{x}_{\ell'})$ is $t_{\mathrm{kernel}}$, then the time complexity for training a support vector machine is given by
\begin{align}
    \mathcal{O}(t_{\mathrm{kernel}} N^2  + \mathrm{poly}(N)) &&& \mbox{(training time)}. \label{eq:general-training-time}
\end{align}
Hence, for support vector machines with efficiently computable kernel functions $k(\vct{x}_\ell, \vct{x}_{\ell'})$, 
small training data sizes $N$ directly ensure a short training time.
The polynomial scaling in training data size depends on the type of algorithm. Dedicated solvers for the soft margin problem \cite{joachims1998making, CC01a, hazan2011beating} require (at most) $\mathcal{O} \left( N^3 + \Lambda^2 N/\epsilon^2 \right)$ arithmetic operations to produce a solution $\vct{\alpha}_{\sharp,\epsilon}$ that is $\epsilon$-close to optimal: $\mathrm{E}_{\mathrm{tr}}(\vct{\alpha}_{\sharp,\epsilon}) \leq \mathrm{E}_{\mathrm{tr}}(\vct{\alpha}_\sharp)+\epsilon$. 
For the concrete training problems considered here, such an approximation is good enough and the associated runtime bound simplifies to $\mathcal{O} \left( t_{\mathrm{kernel}} N^2 + N^3 \right)$.
Interior point methods offer an alternative that scale worse in training data size, but much better in the approximation error $\epsilon$, see e.g.\ \cite{boyd2004convex}.

\subsection{Prediction using support vector machines} \label{sec:pred-SVM}

In the last section, we have explained how feature maps and kernels can considerably boost the expressiveness of initially linear classifiers. We have also explained how to use labeled training data of size $N$ to find a separating hyperplane in feature space by solving a quadratic program~\eqref{eq:training-problem} that depends on the kernel matrix~\eqref{eq:kernel-matrix}. Ideally, $\mathrm{E}_{\mathrm{tr}} (\vct{\alpha}_\sharp)=0$ (zero training error) and the optimal solution $\vct{\alpha}_\sharp \in \mathbb{R}^N$ parametrizes a separating hyperplane with minimal margin $2/\Lambda$ in feature space:
\begin{equation}
h_\sharp (\vct{x}_{\ell'}) = \sum_{\ell=1}^N \left[\vct{\alpha}_\sharp\right]_\ell \langle \phi (\vct{x}_\ell), \phi (\vct{x}_{\ell'}) \rangle_{\mathcal{F}} = \sum_{\ell=1}^N \left[\vct{\alpha}_\sharp\right]_\ell k \left( \vct{x}_\ell, \vct{x}_{\ell'} \right) 
\quad
\begin{cases}
> + 1/\Lambda & \text{if $y_{\ell'}=+1$}, \\
< - 1/\Lambda & \text{else if $y_{\ell'}=-1$,}
\end{cases}
\end{equation}
for all (labeled) training data points $(\vct{x}_1,y_1),\ldots,(\vct{x}_N,y_N)$.The sign of this classifier, in turn, correctly reproduces training labels:
\begin{equation}
y_\sharp (\vct{x}_{\ell'}) := \mathrm{sign} \left( h_\sharp (\vct{x}_{\ell'}) \right) = y_{\ell'} \quad \text{for each $\ell \in \left\{1,\ldots,N\right\}$.} \label{eq:prediction-function}
\end{equation}
In the prediction stage, we use this function to assign a label $y_\sharp (\vct{x}) \in \left\{ \pm 1 \right\}$ to a new (and unlabeled) data point $\vct{x}$. 
The cost of evaluating $y_\sharp (\vct{x}_{\ell'})$ is dominated by the cost of evaluating $N$ kernel functions.
If the time to compute the kernel function is $t_{\mathrm{kernel}}$, then the prediction time for a new input vector $\vct{x}$ is bounded by
\begin{align}
    \mathcal{O}(t_{\mathrm{kernel}} N)&&& \mbox{(prediction time)}.
\end{align}
Similar to the training time~\eqref{eq:general-training-time}, a small training data size $N$ translates into a fast prediction time.

The hope is that training with an adequate kernel uncovers latent structure that generalizes beyond training data. 
Typically, larger training data sizes $N$ also increase the chance for learning good classifiers~\eqref{eq:prediction-function}. 
But generalization beyond training data often only makes sense if the new data point $\vct{x}$ is somewhat related to the training data (e.g.\ training a SVM on labeled cat-vs-dog images does not necessarily produce a classifier that can distinguish apples from oranges). 
Extra assumptions that address similarity of training and prediction data are  important when one aims at establishing rigorous bounds on the probability of making a wrong prediction, i.e.  $y_\sharp (\vct{x})= - y(\vct{x})$. A common assumption is that both the training data and new data points are generated independently from the same distribution: $(\vct{x}_1, y_1),\ldots,(\vct{x}_N, y_N), (\vct{x}, y) \sim  \mathcal{D}$.
The data distribution $\mathcal{D}$ is a joint distribution over both the input vector $\vct{x}$ and the label $y$.
Such an assumption encompasses the intuition that the label $y$ is correlated with the input vector $\vct{x}$, but is not necessarily a function of $\vct{x}$. 
Flexibility of this form is useful for describing situations where the data points $\vct{x}$ are corrupted by noise. This is often the case in quantum mechanics due to the inherent randomness in quantum measurements.
The underlying data distribution should be taken into account when reasoning about false predictions, motivating the probability 
\begin{align}
\mathrm{Pr}_{(\vct{x}, y) \sim \mathcal{D}} \left[ y_\sharp (\vct{x}) \neq y \right] \in & ~[0,1] & \text{(average-case prediction error)}\label{eq:average-prediction}
\end{align}
as a good figure of merit. Noting that there are in general many approaches to bounding the prediction error, 
we present a user-friendly theorem that bounds the average-case prediction error in terms of the training error $\mathrm{E}_{\mathrm{tr}} (\vct{\alpha}_\sharp)$ and training data size $N$.

\begin{theorem}[Prediction error for support vector machines] \label{thm:svm-prediction}
Fix a data distribution $(\vct{x},y) \sim \mathcal{D}$, a kernel function $k(\cdot,\cdot)$, a minimal margin $2/\Lambda$ and a training data size $N$. Assume $k(\vct{x},\vct{x}) \leq R^2$ almost surely.
Then, with probability (at least) $1-\delta$,
\begin{equation}
\mathrm{Pr}_{(\vct{x}, y) \sim \mathcal{D}} \left[ y_\sharp (\vct{x}) \neq y \right] \leq \frac{1}{N} \mathrm{E}_{\mathrm{tr}} (\vct{\alpha}_\sharp) + 7 (\Lambda R +1) \sqrt{\frac{\log (2/\delta)}{N}},
\end{equation}
where $y_\sharp (\vct{x})$ is the classifier~\eqref{eq:prediction-function} obtained from solving the training problem~\eqref{eq:training-problem} on independently sampled training data $\left(\vct{x}_1,y_1\right),\ldots,\left(\vct{x}_N,y_N \right) \sim \mathcal{D}$, and $\mathrm{E}_{\mathrm{tr}}$ denotes the associated training error.
\end{theorem}

This rigorous statement bounds the average prediction error in terms of the training error plus an error term that decays as $1/\sqrt{N}$ in training data size. The core assumption is that training and prediction data is sampled in an independent and identically distributed (\textit{iid}) fashion. The proof is based on specializing a standard result from high dimensional probability theory to the task at hand.

\begin{theorem}[Theorem 3.3 in \cite{mohri2018foundations}] \label{thm:rademacher}
Fix a probability distribution $\mathcal{D}$ over elements in a set $\mathsf{X}$, a family of functions $\mathcal{G}$ from $\mathsf{X}$ to the interval $[0,\gamma_{\max}]$, as well as $\delta \in (0,1)$ and $N \in \mathbb{N}$. Then, with probability $1-\delta$, the following bound is valid for \emph{all} functions $g \in \mathcal{G}$ simultaneously:
\begin{equation}
    \E_{x \sim \mathcal{D}} [g(x)] \leq \frac{1}{N} \sum_{\ell = 1}^N g(x_\ell) + 3 \gamma_{\max} \sqrt{\frac{\log(2 / \delta)}{2 N}} + \frac{2}{\sqrt{N}} \E_{\varepsilon_1,\ldots,\varepsilon_N} \left[ \sup_{g \in \mathcal{G}} \frac{1}{\sqrt{N}} \sum_{\ell = 1}^N \varepsilon_\ell g(x_\ell) \right] . \label{eq:rademacher}
\end{equation}
Here, $x_1,\ldots,x_N \overset{\textit{iid}}{\sim}\mathcal{D}$ are sampled from $\mathsf{X}$ and $\varepsilon_1,\ldots,\varepsilon_N \overset{\textit{iid}}{\sim} \left\{ \pm 1 \right\}$ are Rademacher random variables (the failure probability $\leq \delta$ addresses these random selections).
\end{theorem}


The right hand side of this upper bound contains three qualitatively different contributions. The first term describes an empirical average over $N$ independent samples. It approximates the true expectation value by Monte Carlo sampling, and can underestimate the true average.
As $N$ increases, the approximation accuracy becomes better and, simultaneously, the probability of sampling a poor approximation diminishes exponentially.
This is precisely the content of the second term. Larger sampling rates $N$ suppress it and also allow for insisting on ever smaller failure probabilities $\delta$. However, these two terms are still not enough for an upper bound because we would like to have a bound for \emph{all functions} $g \in \mathcal{G}$. This is where the third term comes into play. It contains the empirical width, a statistical summary parameter for the extent of the function set $\mathcal{G}$, see e.g.\ \cite{vershynin2018high}. Suppose, for instance, that $\mathcal{G}=\left\{g \right\}$ contains only a single function. Then, we can ignore the supremum (over a single element) and the contribution vanishes entirely (Rademacher random variables have zero expectation). 
The empirical width parameter can, however, grow with the size of the function set $g \in \mathcal{G}$. 

In the context of bounding the performance of support vector machines, the domain variable $x$ becomes $(\vct{x},y)$, and the function family consists of the training error $g_{\vct{\alpha}}$ from Eq.~(\ref{eq:SVM-classifier}), indexed by $\vct{\alpha}$. The third term in Theorem~\ref{thm:rademacher} can then be bounded by the largest norm of the feature vectors.

\begin{lemma} \label{lem:svm-width}
Fix a feature map $\phi:\mathbb{R}^D \times \left\{ \pm 1 \right\} \to \mathcal{F}$ and define
$g_{\vct{\alpha}} (\vct{x},y) = \max \left\{0,1-y \langle \vct{\alpha}, \phi (\vct{x}) \rangle_{\mathcal{F}} \right\}$ for $\vct{\alpha} \in \mathcal{F}^*$. Then, 
\begin{align}
\E_{\varepsilon_1,\ldots,\varepsilon_N} \left[ \sup_{\langle \vct{\alpha},\vct{\alpha} \rangle_{\mathcal{F}} \leq \Lambda^2} \tfrac{1}{\sqrt{N}} \sum_{\ell=1}^N \varepsilon_\ell g_{\vct{\alpha}} \left(\vct{x}_\ell,y_\ell\right) \right] \leq \Lambda  \max_{1 \leq \ell \leq N} \sqrt{ \langle \phi (\vct{x}_\ell),\phi (\vct{x}_\ell) \rangle_{\mathcal{F}}}
\end{align}
for any collection $\left(\vct{x}_1,y_1\right),\ldots,\left( \vct{x}_N,y_N\right) \in \mathbb{R}^D \times \left\{ \pm 1 \right\}$.
\end{lemma}

\begin{proof}
Let us abbreviate the expectation over all $N$ Rademacher random variables by $\E_{\varepsilon}$. Note that the empirical width is invariant under a constant shift of the hinge loss function: $\max \left\{0,1-z \right\} \mapsto \max \left\{0,1-z\right\}-1$.
In turn, the shifted loss function $L (z)= \max \left\{0,1-z \right\}-1$ describes a contraction, i.e.\ $L (0) = 0$ and $\left| L (z_1) - L (z_2) \right| \leq \left|z_1 -z_2 \right|$ for all $z_1,z_2 \in \mathbb{R}$.
Such contractions can only decrease the empirical width. More precisely,  the Rademacher comparison principle \cite[Eq.~(4.20)]{ledoux2013probability} asserts
\begin{subequations}
\begin{align}
\E_{\varepsilon}\left[\sup_{\langle\vct{\alpha},\vct{\alpha}\rangle_{\mathcal{F}}\leq\Lambda^{2}}\tfrac{1}{\sqrt{N}}\sum_{\ell=1}^{N}\varepsilon_{\ell}g_{\vct{\alpha}}(\vct{x}_{\ell},y_{\ell})\right]&=\E_{\varepsilon}\left[\sup_{\langle\vct{\alpha},\vct{\alpha}\rangle_{\mathcal{F}}\leq\Lambda^{2}}\tfrac{1}{\sqrt{N}}\sum_{\ell=1}^{N}\varepsilon_{\ell}\left(\max\left\{ 0,1-y_{\ell}\langle\vct{\alpha},\phi(\vct{x}_{\ell})\rangle_{\mathcal{F}}\right\} -1\right)\right]\\&\leq\E_{\varepsilon}\left[\sup_{\langle\vct{\alpha},\vct{\alpha}\rangle_{\mathcal{F}}\leq\Lambda^{2}}\tfrac{1}{\sqrt{N}}\sum_{\ell=1}^{N}\varepsilon_{\ell}y_{\ell}\langle\vct{\alpha},\phi(\vct{x}_{\ell})\rangle_{\mathcal{F}}\right]\\&=\E_{\varepsilon}\left[\sup_{\langle\vct{\alpha},\vct{\alpha}\rangle_{\mathcal{F}}\leq\Lambda^{2}}\langle\vct{\alpha},h_{\varepsilon}\rangle_{\mathcal{F}}\right].
\end{align}
\end{subequations}
In the last step, we have introduced the short-hand notation $h_{\varepsilon} = \tfrac{1}{\sqrt{N}} \sum_{\ell=1}^N \varepsilon_\ell y_\ell \phi (\vct{x}_\ell) \in \mathcal{F}$.
Applying a Cauchy-Schwarz inequality in feature space allows us to separate the supremum from the Rademacher randomness:
\begin{align}
\E_\varepsilon \left[ \sup_{\langle \vct{\alpha}, \vct{\alpha} \rangle_{\mathcal{F}} \leq \Lambda^2} \langle \vct{\alpha}, h_\varepsilon \rangle_{\mathcal{F}} \right] 
\leq  \sup_{\langle \vct{\alpha}, \vct{\alpha} \rangle_{\mathcal{F}} \leq \Lambda^2} \sqrt{\langle \vct{\alpha}, \vct{\alpha} \rangle_{\mathcal{F}}}
\;\E_\varepsilon \left[ \sqrt{\langle h_\varepsilon, h_\varepsilon \rangle_{\mathcal{F}}} \right]
\leq  \Lambda \sqrt{ \E_\varepsilon \langle h_\varepsilon, h_\varepsilon \rangle_{\mathcal{F}}}.
\end{align}
The last inequality is Jensen's.
We complete the argument 
using $\E_\varepsilon \left[ \varepsilon_\ell \varepsilon_{\ell'} \right]=\delta_{\ell,\ell'}$ and $y_\ell^2=1$:
\begin{equation}
\E _\varepsilon \langle h_\varepsilon, h_\varepsilon \rangle_{\mathcal{F}}
= \frac{1}{N}\sum_{\ell,\ell'=1}^N \E _{\varepsilon} \left[ \varepsilon_{\ell} \varepsilon_{\ell'} \right]  y_\ell y_{\ell'} \langle \phi (\vct{x}_\ell), \phi (\vct{x}_{\ell'}) \rangle_{\mathcal{F}}
= \frac{1}{N}\sum_{\ell=1}^N \langle \phi (\vct{x}_\ell), \phi (\vct{x}_\ell) \rangle_{\mathcal{F}} \leq \max_{1 \leq \ell \leq N} \langle \phi (\vct{x}_\ell),\phi (\vct{x}_\ell) \rangle_{\mathcal{F}}.
\end{equation}
\end{proof}

We are now ready to prove the general connection between average prediction (\ref{eq:average-prediction}) and training error.

\begin{proof}[Proof of Theorem~\ref{thm:svm-prediction}]
We consider functions $y_{\vct{\alpha}} (\vct{x}) = \mathrm{sign} \left( \langle {\vct{\alpha}}, \phi \left( \vct{x} \right) \rangle_{\mathcal{F}} \right) \in \left\{\pm 1 \right\}$, such that ${\vct{\alpha}} \in \mathcal{F}^*$ obeys $\langle {\vct{\alpha}}, {\vct{\alpha}} \rangle_{\mathcal{F}} \leq \Lambda^2$. This family of functions includes all classifiers that are feasible points in the training stage~\eqref{eq:training-problem} of our support vector machine. 
For ${\vct{\alpha}}$ fixed, but otherwise arbitrary, we want to compare the corresponding classifier $y_{\vct{\alpha}} \left( \vct{x} \right)$ to the true data label $y \in \left\{ \pm 1 \right\}$. 
Elementary reformulations then allow us to re-express the failure probability as
\begin{equation}
    \mathrm{Pr}_{(\vct{x}, y) \sim \mathcal{D}} \left[ y_{{\vct{\alpha}}} (\vct{x}) \neq y \right]
= \mathrm{Pr}_{(\vct{x}, y) \sim \mathcal{D}} \left[ \mathrm{sign} \left( \langle {\vct{\alpha}}, \phi(\vct{x}) \rangle_{\mathcal{F}}\right) \neq y \right]
= \mathrm{Pr}_{(\vct{x}, y) \sim \mathcal{D}} \left[ y \langle {\vct{\alpha}}, \phi(\vct{x}) \rangle_{\mathcal{F}} < 0 \right],
\end{equation}
because the sign is negative if and only if the number itself is. 
Next, we rewrite this probability as the expectation value of the associated indicator function and use $\mathbf{1} \left\{ z \leq 0 \right\} \leq \max \left\{0,1-z \right\}$ for all $z \in \mathbb{R}$ to obtain
\begin{equation}
\mathrm{Pr}_{(\vct{x}, y) \sim \mathcal{D}} \left[ y \langle {\vct{\alpha}}, \phi(\vct{x}) \rangle_{\mathcal{F}} < 0 \right]
= \E _{(\vct{x}, y) \sim \mathcal{D}} \left[ \vct{1} \left\{y \langle {\vct{\alpha}}, \phi(\vct{x}_\ell) \rangle_{\mathcal{F}} < 0  \right\} \right]
\leq \E _{(\vct{x}, y) \sim \mathcal{D}} \left[\max \left\{0,1- y \langle {\vct{\alpha}}, \phi (\vct{x}) \rangle_{\mathcal{F}} \right\} \right].
\end{equation}
This upper bound is the expected value of a certain hinge loss function
\begin{equation}
g_{{\vct{\alpha}}} \left( \vct{x},y \right) = \max \left\{0,1-y\langle {\vct{\alpha}}, \phi (\vct{x}) \rangle_{\mathcal{F}} \right\}\quad \text{with} \quad \langle {\vct{\alpha}},{\vct{\alpha}} \rangle_{\mathcal{F}} \leq \Lambda^2.
\end{equation}
The function is a specific element of an entire family, namely
\begin{equation}
    \mathcal{G} =\left\{ g_{\vct{\alpha}} (\cdot,\cdot): \langle {\vct{\alpha}},{\vct{\alpha}}\rangle_{\mathcal{F}} \leq \Lambda^2\right\}: \mathbb{R}^D \times \left\{ \pm 1 \right\} \to \left[0,\infty \right).
\end{equation}
The associated  function values are always nonnegative and bounded. Indeed, the Cauchy-Schwarz inequality in feature space asserts
\begin{align}
g_{\vct{\alpha}} (\vct{x},y) \leq & \left| y \langle {\vct{\alpha}}, \phi (\vct{x}) \rangle_{\mathcal{F}} \right| +1
\leq \sqrt{\langle {\vct{\alpha}},{\vct{\alpha}} \rangle_{\mathcal{F}} \langle \phi (\vct{x}),\phi (\vct{x}) \rangle_{\mathcal{F}}} +1
\leq \Lambda \sqrt{k (\vct{x},\vct{x})} + 1 \leq \Lambda R+1 =: \gamma_{\max}.
\end{align}
We are now in a position to use Theorem~\ref{thm:rademacher}. 
With probability (at least) $1-\delta$,
\begin{align}
\E _{(\vct{x}, y) \sim \mathcal{D}} \left[ g_{\vct{\alpha}} (\vct{x},y) \right]
\leq & \frac{1}{N} \sum_{\ell=1}^N g_{\vct{\alpha}} (\vct{x}_\ell,y_\ell) + 3 \gamma_{\max} \sqrt{\frac{\log (2/\delta)}{2N}}
+ \frac{2}{\sqrt{N}}\E _{\varepsilon}\left[ \sup_{\langle {\vct{\alpha}},{\vct{\alpha}} \rangle_{\mathcal{F}} \leq \Lambda^2} \frac{1}{\sqrt{N}}\sum_{\ell=1}^N \varepsilon_\ell g_{\vct{\alpha}} (\vct{x}_\ell,y_\ell) \right],
\end{align}
is true for \emph{all} dual vectors ${\vct{\alpha}} \in \mathcal{F}^*$ that obey $\langle {\vct{\alpha}},{\vct{\alpha}} \rangle_{\mathcal{F}} \leq \Lambda^2$. Here, $\left(\vct{x}_1,y_1\right),\ldots,\left(\vct{x}_N,y_N\right)\sim \mathcal{D}$ is a randomly sampled (but fixed) collection of labeled data points. We now use $\sqrt{k(\vct{x}_\ell,\vct{x}_\ell)} \leq R$ almost surely to apply Lemma~\ref{lem:svm-width} and control the empirical width term:
\begin{subequations}
\begin{align}
\E _{(\vct{x},y)\sim\mathcal{D}}\left[g_{\vct{\alpha}}(\vct{x},y)\right]&\leq\frac{1}{N}\sum_{\ell=1}^{N}g_{\vct{\alpha}}(\vct{x}_{\ell},y_{\ell})+3(\Lambda R+1)\sqrt{\frac{\log(2/\delta)}{2N}}+\frac{2\Lambda R}{\sqrt{N}}\\&\leq\frac{1}{N}\sum_{\ell=1}^{N}g_{\vct{\alpha}}(\vct{x}_{\ell},y_{\ell})+7(\Lambda R+1)\sqrt{\frac{\log(2/\delta)}{N}}.
\end{align}
\end{subequations}
With probability (at least) $1-\delta$, this bound is valid for all hyperplane vectors ${\vct{\alpha}} \in \mathcal{F}$.
The tightest bound is achieved for minimizing the right hand side. This is precisely what training a support vector machine does, as the first term is precisely the training error that is minimized in the training stage~\eqref{eq:training-problem}.
The optimal solution ${\vct{\alpha}}^\sharp$ to this problem simultaneously produces the actual classifier $y_\sharp (\vct{x})$ on the left hand side and the (minimal) training error on the right hand side.
\end{proof}

\subsection{Kernel functions for classical shadows} \label{sub:shadow-kernel}


We have reviewed the classical shadow formalism in Appendix~\ref{sec:classical-shadows}.
For randomized single-qubit Pauli measurements,  a classical shadow approximates a $n$-qubit state $\rho$ by means of $T$ elementary tensor products.
Each shadow raw data corresponds to a two-dimensional array
\begin{align}
S_T (\rho) = S_T (\rho) =\left\{ |s_i^{(t)} \rangle :\; i \in \{1,\ldots,n\}, t \in \{1,\ldots,T\}\right\} 
\in \left\{ |0 \rangle, |1 \rangle, |+ \rangle, |- \rangle, |\mathrm{i}+\rangle,|\mathrm{i}-\rangle \right\}^{n \times T}
\end{align}
and is combined into an approximator of the state as
\begin{align}
\sigma_T (\rho) = \frac{1}{T} \sum_{t=1}^T \left( 3|s_1^{(t)} \rangle \! \langle s_1^{(t)}|-\mathbb{I} \right) \otimes \cdots \otimes \left( 3 |s_n^{(t)} \rangle \! \langle s_n^{(t)}| -\mathbb{I} \right) 
= \frac{1}{T}\sum_{t=1}^T \sigma_1^{(t)} \otimes \cdots \otimes \sigma_n^{(t)},
\end{align}
where we have introduced the short-hand notation $\sigma_i^{(t)}=3|s_i^{(t)} \rangle \! \langle s_i^{(t)}|-\mathbb{I}$.
For these quantum state representations, we fix parameters $\tau, \gamma >0$ and introduce a suggestive, yet finite-dimensional feature map. 
For large, but finite, integers $D,R >0$ we define
\begin{equation}
\phi^{\text{(finite)}}\left( S_T (\rho) \right)
= \bigoplus_{d=0}^D \sqrt{\frac{\tau^d}{d!}} \Big( \bigoplus_{r=0}^R \sqrt{\frac{1}{r!} \left( \frac{\gamma}{n}\right)^r} \bigoplus_{i_1=1}^r \cdots \bigoplus_{i_r=1}^r \frac{1}{T}\sum_{t=1}^T \mathrm{vec} \left( \sigma_{i_1}^{(t)} \right) \otimes \cdots \otimes \mathrm{vec} \left( \sigma_{i_r}^{(t)} \right) 
\Big)^{\otimes d}, \label{eq:feature-space-linearization}
\end{equation}
Here, $\mathrm{vec}(\cdot)$ denotes an appropriate vectorization operation that maps the real-valued vector space $\mathbb{H}_2$ of Hermitian $2 \times 2$ matrices to $\mathbb{R}^4$ such that the Hilbert-Schmidt inner product is preserved: $\langle \mathrm{vec}(A),\mathrm{vec}(B) \rangle = \mathrm{tr}(AB)$. 

This feature map embeds classical shadows in a very large-dimensional, real-valued feature space $\mathcal{F}^{\text{(finite)}}$. This feature space arises from taking direct sums and tensor products of $\mathrm{vec} (\mathbb{H}_2) \simeq \mathbb{R}^4$. We can extend the standard inner product $\langle \cdot,\cdot \rangle$ on $\mathbb{R}^4$ to this feature space by setting $\langle x_1 \oplus x_2, y_1 \oplus y_2 \rangle = \langle x_1, y_1 \rangle + \langle x_2,y_2 \rangle $ (direct sums), as well as $\langle x_1 \otimes x_2, y_1 \otimes y_2 \rangle = \langle x_1, y_1 \rangle \langle x_2, y_2 \rangle$ (tensor products) and extend these definitions linearly. Doing so equips the feature space $\mathcal{F}^{\text{(finite)}}$ with a well-defined inner product $\langle \cdot, \cdot \rangle_{\mathcal{F}^{\text{(finite)}}}$.
The inner product and feature map induce a kernel function on pairs of classical shadows of equal size $T$:
\begin{subequations}
\begin{align}
& k^{\text{(finite)}}\left( S_T (\rho_1), \tilde{S}_T (\rho_2) \right)
= \left\langle \phi^{\text{(finite)}} \big( S_T (\rho_1) \big), \phi^{\text{(finite)}} \big( \tilde{S}_T (\rho_2) \big) \right\rangle_{\mathcal{F}^{\text{(finite)}}}   \\
=& \sum_{d=0}^D \frac{\tau^d}{d!} \left( \sum_{r=0}^R \frac{1}{r!} \left( \frac{\gamma}{n}\right)^r \sum_{i_1=1}^n \cdots \sum_{i_r=1}^n \frac{1}{T^2} \sum_{t,t'=1}^T \left\langle \mathrm{vec} \left( \sigma_{i_1}^{(t)} \right), \mathrm{vec} \left( \tilde{\sigma}_{i_1}^{(t')}\right) \right\rangle \cdots \left\langle \mathrm{vec} \left( \sigma_{i_r}^{(t)} \right), \mathrm{vec} \left( \tilde{\sigma}_{i_r}^{(t')}\right) \right\rangle \right)^d  \\
=& \sum_{d=0}^D \frac{\tau^d}{d!} \left( \sum_{r=0}^R \frac{1}{r!} \left( \frac{\gamma}{n}\right)^r \sum_{i_1=1}^n \cdots \sum_{i_r=1}^n \mathrm{tr} \left( \left( \frac{1}{T}\sum_{t=1}^T \sigma_{i_1}^{(t)} \otimes \cdots \otimes \sigma_{i_r}^{(t)} \right) \right) 
\left( \frac{1}{T}\sum_{t'=1}^T \tilde{\sigma}_{i_1}^{(t')} \otimes \cdots \otimes \tilde{\sigma}_{i_r}^{(t')} \right) \right)^d \\
=& \sum_{d=0}^D \frac{1}{d!} \left( \frac{\tau}{T^2} \sum_{t,t'=1}^T \sum_{r=0}^R \frac{1}{r!} \left( \frac{\gamma}{n}\right)^r \sum_{i=1}^n \cdots \sum_{i_r=1}^r \mathrm{tr} \left( \sigma_{i_1}^{(t)} \tilde{\sigma}_{i_1}^{(t')} \right) \cdots \mathrm{tr} \left( \sigma_{i_r}^{(t)} \tilde{\sigma}_{i_r}^{(t')}\right) \right)^d \\
=& \sum_{d=0}^D \frac{1}{d!} \left( \frac{\tau}{T^2} \sum_{t,t'=1}^T \sum_{r=0}^R \frac{1}{r!} \left( \frac{\gamma}{n} \sum_{i=1}^n  \mathrm{tr} \left( \sigma_i^{(t)} \tilde{\sigma}_i^{(t')} \right)\right)^r \right)^d.
\end{align}
\end{subequations}
This kernel function sill looks somewhat complicated, but it simplifies considerably if we first take $R \to \infty$ and then $D \to \infty$:
\begin{subequations}
\begin{align}
k^{\mathrm{(shadow)}} \left(S_T (\rho_1),\tilde{S}_T (\rho_2) \right) := &
\lim_{D \to \infty} \lim_{R \to \infty} k^{\text{(finite)}}  \left(S_T (\rho_1),\tilde{S}_T (\rho_2) \right)  \label{eq:shadow-kernel} \\
=& 
\lim_{D \to \infty} \sum_{d=0}^D \frac{1}{d!} \left( \frac{\tau}{T^2}\sum_{t,t'=1}^T  \lim_{R \to \infty} \sum_{r=0}^R \frac{1}{r!} \left( \frac{\gamma}{n}\sum_{i=1}^n \mathrm{tr} \left( \sigma_i^{(t)} \tilde{\sigma}_i^{(t')} \right)\right)^r \right)^d \\
=& \exp \left( \frac{\tau}{T^2} \sum_{t,t'=1}^T \exp \left( \frac{\gamma}{n}\sum_{i=1}^n \mathrm{tr} \left( \sigma_i^{(t)} \tilde{\sigma}_i^{(t')} \right) \right) \right)  
\end{align}
\end{subequations}
We call this kernel function a \emph{shadow kernel}. In contrast to its finite approximations, this kernel function can be computed very efficiently. Trace inner products between single-qubit shadow constituents assume one out of 3 values only:
\begin{align}
\mathrm{tr} \left( \sigma_i^{(t)} \tilde{\sigma}_i^{(t)} \right) =& \mathrm{tr} \left( (3 |s_i^{(t)}\rangle \! \langle s_i^{(t)}| -\mathbb{I} ) (3 |\tilde{s}_i^{(t)}\rangle \! \langle \tilde{s}_i^{(t)}|-\mathbb{I} ) \right)
= 9 \left| \langle s_i^{(t)}| \tilde{s}_i^{(t)} \rangle \right|^2 -4 \in \left\{-4,1/2,5 \right\}.
\label{eq:shadow-inner-products}
\end{align}
And we need to compute exactly $nT^2$ of them to unambiguously characterize the shadow kernel~\eqref{eq:shadow-kernel}. The total cost for evaluating shadow kernels also amounts to
\begin{align}
\mathcal{O} \left( n T^2 \right) && \text{(shadow kernel evaluation cost)} \label{eq:shadow-kernel-evaluation-cost}
\end{align}
arithmetic operations. As long as $T$ is not too large, this is extremely efficient, given that we combine classical approximations of $n$-qubit quantum states $\rho_1,\rho_2$ which way well have $(4^n-1)$ degrees of freedom. Eq.~\eqref{eq:shadow-inner-products} also ensures that shadow kernels remain bounded functions:
\begin{equation}
0 \leq k^{\mathrm{(shadow)}} \left( S_T (\rho_1),\tilde{S}_T (\rho_2) \right)
\leq 
\exp \left( \tau \exp \left( 5 \gamma \right) \right), \label{eq:kernel-upper-bound}
\end{equation}
because exponential functions are nonnegative and monotonic.

While easy to evaluate and conceptually appealing, the shadow kernel does have its downsides. By construction, the associated feature space is not finite-dimensional anymore. This can complicate a thorough analysis of support vector machines substantially. In particular, it is a priori not clear if powerful results, like Theorem~\ref{thm:svm-prediction}, cover the shadow kernel as well. 
Fortunately, we can bypass such mathematical subtleties by approximating $k^{\mathrm{(shadow)}}(\cdot,\cdot)$ with $k^{\text{(finite)}}(\cdot,\cdot)$, where $D$ and $R$ are large, but finite, numbers. 
This incurs an additional approximation error, but allows us to formulate theoretical prediction and training guarantees exclusively for finite-dimensional feature spaces.
What is more, elementary approximation results from calculus ensure that we can make this additional approximation error arbitrarily small by making the cutoffs sufficiently large. 
Taylor's approximation theorem, for instance, shows that $D=\mathrm{e}^2 \tau \exp (5 \gamma) + \log (1/\eta) -1$, as well as $R=5 \mathrm{e}^2 \gamma + \tau \exp (5 \gamma) + \log (\tau/\eta)-1$ 
ensure 
\begin{equation}
\left|k^{\mathrm{(shadow)}} \left( S_T (\rho_1),\tilde{S}_T (\rho_2)\right)-k^{\text{(finite)}}\left( S_T (\rho_1),\tilde{S}_T (\rho_2) \right) \right| \leq 2 \eta
\end{equation}
for all pairs of classical shadows with compatible size $T$.
Properly tuning $\gamma$ and $\tau$ would yield better prediction performance in practice. Nevertheless, for simplicity, we will assume $\gamma = \tau = 1$ in the following theoretical analysis.

Finite-dimensional feature space approximations also allow us to highlight the expressiveness behind the shadow kernel~\eqref{eq:shadow-kernel}. 
It describes (the limit of) a feature map that extracts \emph{all} tensor powers of \emph{all}
subsystem operators $X_A =\mathrm{tr}_{\neg A} (X) \in \mathbb{H}_2^{\otimes |A|}$, where $A \subset \left[n\right]=\left\{1,\ldots,n\right\}$.
In particular, any function that can be written as a finite power series, of degree at most $d_p$, in reduced subsystem operators, of size at most $r$, becomes a \emph{linear} function in feature space, represented by the dual vector $\vct{\alpha}_f$:
\begin{subequations}
\begin{align}
f\left( S_T (\rho) \right)&= \sum_{d=0}^{d_p} \frac{1}{d!} \sum_{A_1 \ldots A_d \subset \{1,\ldots,n\}, |A_i| \leq r} \mathrm{tr} \left( O_{A_1,\ldots,A_d} \mathrm{tr}_{\neg A_1} \left( \sigma_T (\rho) \right) \otimes \cdots \otimes \mathrm{tr}_{\neg A_d} \left( \sigma_T (\rho) \right) \right) \\
&= \langle \vct{\alpha}_f, \phi^{\text{(finite)}} (S_T (\rho))\rangle_{\mathcal{F}^{\text{(finite)}}}, \label{eq:feature-linearization}
\end{align}
\end{subequations}
provided that $d_p \leq D, r \leq R$.
The (extended) Euclidean norm of $\vct{\alpha}_f$ is also bounded. Use Eq.~\eqref{eq:feature-space-linearization} (with tuning parameters $\gamma,\tau = 1$) to compute
\begin{subequations}
\label{eq:separator-size}
\begin{align}
\langle \vct{\alpha}_f, \vct{\alpha}_f \rangle_{\mathcal{F}^{\text{(finite)}}} 
& \leq  \sum_{d=0}^{d_p}  \frac{\left(r! n^r\right)^d}{d!}
\sum_{A_1,\ldots,A_d \subset \{1,\ldots,n\}, |A_i| \leq r} \mathrm{tr} \left( O_{A_1,\ldots,A_d}^2 \right) \\
& \leq \sum_{d=0}^{d_p}  \frac{\left(r! n^r\right)^d}{d!} \sum_{A_1,\ldots,A_d \subset \{1,\ldots,n\}, |A_i| \leq r}  2^{rd} \| O_{A_1,\ldots,A_d} \|_\infty^2 \\
& \leq \left(2nr\right)^{r {d_p}} \max_{\substack{d \leq d_p, A_1, \ldots, A_d\\ \subset \{1, \ldots, n\}, |A_i| \leq r}} \norm{O_{A_1, \ldots, A_d}}_\infty \,\,\, \sum_{d=0}^{d_p} \frac{1}{d!}
\sum_{A_1,\ldots,A_d \subset \{1,\ldots,n\}, |A_i| \leq r} \| O_{A_1,\ldots,A_d} \|_\infty \\ 
& \leq \left(2nr\right)^{r {d_p}} d_p^{d_p} \left( \sum_{d=0}^{d_p} \frac{1}{d!}
\sum_{A_1,\ldots,A_d \subset \{1,\ldots,n\}, |A_i| \leq r} \| O_{A_1,\ldots,A_d} \|_\infty\right)^2
\end{align}
\end{subequations}
Here, we have used the fundamental Schatten-$p$ norm relation $\|X \|_2 \leq \sqrt{\mathrm{dim}(X)} \|X \|_\infty$, as well as the assumption that each $O_{A_1,\ldots,A_d}$ is supported on a total tensor product space with dimension $2^{rd}$ (a tensor product of  $d$ subsystems comprised of at most $r$ qubits each).
The second to last inequality follow from using $\sum_i x_i^2 \leq \max_i |x_i| \sum_i |x_i|$, and Stirling's formula.
The final simplifications uses Stirling's formula again as well as the fact that $\sum_i |x_i| \geq \max_i |x_i|$.

\subsection{Physical assumptions about classifying quantum phases of matter} \label{sec:phases-assumption}

We want to learn how to classify two phases of $n$-qubit states: either $\rho$ belongs to phase $A$ ($y(\rho)=+1$) or $\rho$ belongs to phase $B$ ($y(\rho) = -1$). We assume that we have access to labeled classical shadows: $\left\{ \big(S_T (\rho_\ell), y(\rho_\ell) \big):\; \ell \in \{1,\ldots,N\} \right\}$, where each $S_T (\rho_\ell)$  is classical shadow data obtained from performing $T$ randomized single-qubit measurements on independent copies of $\rho_\ell$.
We can use this raw data to form classical representations $\sigma_T (\rho_\ell)$  of the underlying quantum state $\rho_\ell$, see Eq.~\eqref{eq:sigma-T-shadow}.
The number $T$ determines the resolution of these approximations. Note that $\sigma_T (\rho_\ell) \approx \rho_\ell$ can only become exact for $T \geq \exp \left( \Omega (n) \right)$ \cite{guta2020fast,haah2017sample}. This would be far too costly for experimental implementations and efficient data processing. For instance, recall from Eq.~\eqref{eq:shadow-kernel-evaluation-cost} that a single shadow kernel evaluation scales quadratically in $T$. In this section, we show  that we can choose much coarser resolutions if the underlying phase can be classified by a nice analytic function on reduced density matrices.

\begin{assumption}[well-conditioned phase separation] \label{as:phase-separation}
Consider two phases among $n$-qubit states. For $\epsilon >0$, we assume that there exists a 
function $f$ on reduced $r$-body density matrices $\rho_{A} = \mathrm{tr}_{\neg A} (\rho)$ that can distinguish the two phases in question. In particular,
\begin{subequations}
\label{eq:classification-assumption}
\begin{align}
f(\rho) &= f\left( \left\{\rho_A:\; A \subset\{1,\ldots,n\},|A| \leq r\right\} \right) \quad \text{satisfies} \\
f(\rho) &\quad \begin{cases} >  +1&  \text{for all $\rho$ that belong to phase $A$ ($y(\rho)=+1$)}, \\
<  -1 &\text{for all $\rho$ that belong to phase $B$ ($y(\rho)=-1$)}.
\end{cases}
\end{align}
\end{subequations}
Moreover, we assume that $f(\rho)$ can be approximated by a truncated power series
\begin{equation}
    f^{(d_p)}(\rho) = \sum_{d=0}^{{d_p}} \frac{1}{d!} \sum_{A_1,\ldots,A_d\subset\{1,\ldots,n\},|A_i| \leq r} \mathrm{tr} \left( O_{A_1,\ldots,A_d} \rho_{A_1} \otimes \cdots \otimes \rho_{A_d}\right), \label{eq:polynomial-approximation}
\end{equation}
up to constant accuracy: $\left|f(\rho) - f^{(d_p)}(\rho)\right| \leq 0.25$ for all $n$-qubit quantum states $\rho$.
We refer to $d_p$ as the truncation degree and define the normalization constant 
\begin{equation}
    C = \sum_{d=0}^{d_p} \frac{1}{d!}  \Big(\sum_{A_1,\ldots,A_d \subset\{1,\ldots,n\},|A_i| \leq r} \| O_{A_1,\ldots,A_d}\|_\infty \Big).
\end{equation}
\end{assumption}

We don't need to know the normalization constant exactly. An upper bound is fully adequate for the theoretical analysis presented in this section.

Morally, the second part of Assumption~\ref{as:phase-separation} requires that the phase classication function can be well-approximated by a degree-$d_p$-polynomial in reduced density matrices. 
The actual formulation is general enough to encompass most physically relevant functions. 
Let us illustrate this by means of three popular examples. 

\paragraph{Subsystem purity:} Fix a subsystem $A \subset \left\{1,\ldots,n\right\}$ comprised of $|A|=r$ qubits and let $\rho_A = \mathrm{tr}_{\neg A}(\rho)$ be the associated $r$-body density matrix. The subsystem purity $f(\rho)=\mathrm{tr}(\rho_A^2)$ is a quadratic polynomial in this reduced density matrix. We can rewrite this as $f^{(2)}(\rho) = \mathrm{tr}(S_A \rho_A \otimes \rho_A)$, where $S_A$ denotes the swap operator between two copies of the subsystem $A$. This reformulation is also an \emph{exact} approximation of $f(\rho)$ with degree $d_p=2$ and normalization constant $C=\frac{1}{2!} \|S_A \|_\infty = \frac{1}{2}$.
These arguments readily extend to averages of multiple subsystem purities.

\paragraph{Subsystem R\'enyi entropy:} Let us consider the subsystem R\'enyi entropy of order two $H_2 (\rho_A) = - \log \left( \mathrm{tr}(\rho_A^2)\right)$ (the argument will generalize straightforwardly to higher order entropies). This function is closely related to the subsystem purity, but also features a logarithm. And, although the logarithm is \emph{not} a polynomial, $-\log (1-x)$ can be accurately approximated by the truncated Mercator series. A crude, but sufficient, bound ensures
\begin{align}
l^{(d_p)}(x) = \sum_{d=1}^{d_p} \frac{1}{d}x^d
\quad \text{obeys} \quad \left| l^{(d_p)}(x)- \log (1-x)\right| \leq x^{d_p}  \log \left(1/(1-x)\right)
\quad \text{for $x \in (-1,1)$.}
\end{align}
We can now approximate $H_2 (\rho_A)=-\log \left( 1- (1-\mathrm{tr}(\rho_A^2))\right)$ by $l^{(d_p)}(1-\mathrm{tr}(\rho_A^2))$. Subsystem purities necessarily obey $\mathrm{tr}(\rho_A^2) \geq 2^{-|A|}=2^{-r}$.
This allows us to conclude
\begin{align}
\left| l^{(d_p)}\left(1- \mathrm{tr}(\rho_A^2) \right) - H_2 (\rho_A) \right| 
\leq \left(1-2^{-r}\right)^{d_p} r \log (2)
\leq \log (2) r \exp \left( - d_p/2^r\right)
\end{align}
which drops beneath $0.25$ if we set $d_p = \log ( 4\log(2) r) 2^r=\mathcal{O} \left( \log (r) 2^r \right)$. This degree scales exponentially in the subsystem size $r$, but is independent of total dimension. 
We can also use $1=\mathrm{tr}(\rho_A)^2=\mathrm{tr} \left( \mathbb{I}_A^{\otimes 2} \rho_A^{\otimes 2} \right)$
and $\mathrm{tr}(X) \mathrm{tr}(Y)=\mathrm{tr}(X \otimes Y)$
to bring this polynomial approximation onto the form advertised in Eq.~\eqref{eq:polynomial-approximation}. Indeed,
\begin{subequations}
\begin{align}
l^{(d_p)}\left( 1- \mathrm{tr}(\rho_A^2) \right) =& l^{(d_p)} \left( \mathrm{tr} \left( (\mathbb{I}_A^{\otimes 2}- S_A) \rho_A^{\otimes 2}\right) \right)
= \sum_{d=1}^{d_p} \frac{1}{d!} \mathrm{tr} \left( (d-1)! \left( \mathbb{I}_A^{\otimes 2} - S_A \right)^{\otimes d} \rho_A^{\otimes 2d} \right) \quad \text{and} \\
C =& \sum_{d=1}^{d_p} \frac{1}{d!} \left\| (d-1)! (\mathbb{I}_A^{\otimes 2} - S_A )^{\otimes d} \right\|_\infty = \sum_{d=1}^{d_p} \frac{1}{d} \left\| \mathbb{I}_A^{\otimes 2}-S_A \right\|_\infty^d= \sum_{d=1}^{d_p} \frac{1}{d}
\approx \log (d_p).
\end{align}
\end{subequations}
This analysis readily extends to higher order R\'enyi entropies, as well as averages over multiple subsystems.

\paragraph{Entanglement entropy:} This is where things start to get somewhat interesting, because the entanglement (von Neumann) entropy $H(\rho_A) = -\mathrm{tr} \left( \rho_A \log (\rho_A) \right) \in [0,r \log (2)]$ of a $r$-body subsystem is notoriously difficult to accurately approximate with a polynomial \cite{Fawzi2019}. Fortunately, Assumption~\ref{as:phase-separation} does not require an accurate approximation -- a constant error of size $1/4$ is fine. To achieve this goal, we make the following polynomial ansatz in the reduced density matrix $\rho_A$:
\begin{align}
H^{(d_p)}(\rho_A) = - \Tr\left( \left(\rho_A - \mathbb{I}_A\right) + \sum_{k=2}^{d_p} \frac{\left(\mathbb{I}_A - \rho_A\right)^k}{k(k-1)} \right)
\end{align}
Let $\lambda_i$ denote the eigenvalues of a subsystem density matrix $\rho_A$ and note that there are $2^r$ eigenvalues in $\rho_A$. We can rewrite the entanglement entropy and the polynomial ansatz as
\begin{subequations}
\begin{align}
H(\rho_A) &= - \sum_{i=1}^{2^r} \lambda_i \log (\lambda_i) \quad \text{and} \\
H^{(d_p)}(\rho_A) & = - \sum_{i=1}^{2^r} \left( (\lambda_i - 1) + \sum_{k=2}^{d_p} \frac{\left(1 - \lambda_i \right)^k}{k(k-1)} \right),
\end{align}
\end{subequations}
respectively.
Using Taylor's theorem in the interval $[0, 1]$, we have
\begin{equation}\label{eq:xlogx-taylor}
    x \log (x) = (x - 1) + \left( \sum_{k=2}^{\infty} \frac{(1-x)^k}{k(k-1)} \right).
\end{equation}
Note that at $x = 0$, $x\log x = 0$ and the infinite sum comprising the second term on the right hand side also converges to $1$. 
This ensures that the above equality is valid for the closed interval $[0,1]$.
We shall also use the following identity
\begin{equation} \label{eq:partialsum}
    \sum_{k=2}^{n} \frac{1}{k(k-1)} = 1 - \frac{1}{n},
\end{equation}
which remains valid even in the limit $n \to \infty$.
We can combine Eq.~\eqref{eq:xlogx-taylor}~and~\eqref{eq:partialsum} to obtain an approximation error for our polynomial ansatz function.
For all $x \in [0, 1]$, we have
\begin{equation}
    \left| x \log (x) - \left( (x - 1) + \left( \sum_{k=2}^{d_p} \frac{(1-x)^k}{k(k-1)} \right) \right)\right| \leq \sum_{k=d_p + 1}^{\infty} \frac{(1-x)^k}{k(k-1)} \leq \sum_{k=d_p + 1}^{\infty} \frac{1}{k(k-1)} = \frac{1}{d_p}.
\end{equation}
This allows us to bound the approximation error for each individual eigenvalue $\lambda_i \in [0,1]$ of $\rho_A$. There are in total $2^r$ eigenvalues and a triangle inequality asserts
\begin{align}
 |H(\rho_A) - H^{(d_p)}(\rho_A)| 
 \leq \sum_{i=1}^{2^r} \left| \lambda_i \log (\lambda_i)- \left( (\lambda_i-1) + \left( \sum_{k=2}^{d_p} \frac{(1-\lambda_i)^k}{k(k-1)}\right) \right) \right| \leq \frac{2^d}{d_p}.
\end{align}
By choosing $d_p = 2^{r+2}$, we can approximate the entanglement entropy in $r$-body subsystem by a polynomial function.
As long as the subsystem size $r$ is a constant independent of total system size $n$, the polynomial approximation degree $d_p$ is also a constant. And it is not hard to check that the same is true for the normalization constant $C$.
This analysis readily extends to averages of multiple entanglement entropies.

\subsection{Training with shadow kernels} \label{sec:train-shadowkernel}

We are now ready to dive into the main results of this section: converting Assumption~\ref{as:phase-separation} into a statement about classical shadows and their expressiveness when it comes to training a support vector machine.
Our measure of similarity is the \emph{shadow kernel}~\eqref{eq:shadow-kernel} evaluated on classical shadows. The kernel matrix is
\begin{equation}
\left[\mtx{K}\right]_{\ell \ell'}= k^{\mathrm{(shadow)}} \left( S_T(\rho_\ell),S_T (\rho_{\ell'}) \right) \quad \text{for $\ell, \ell' \in \{1,\ldots,N\}$},
\end{equation}
and implicitly specifies the feature map, as well as the nonlinear geometry with respect to which we want to find classifiers for phases.
We begin by approximating the true classifier, given as a nonlinear function $f(\rho)$ in Assumption~\ref{as:phase-separation}, by a finite power series $f^{(d_p)}(\rho)$ with degree-$d_p$.
We will then use $f^{(d_p)}(\rho)$ as an approximate phase classifier.
Recalling Eq.~\eqref{eq:feature-linearization}, a finite power series $f^{(d_p)}(S_T (\rho))$ is linear in feature space, with its corresponding dual vector $\vct{\alpha}_f$ defining a candidate hyperplane for separating the two phases.
To complete the connection to the support vector machines from Section~\ref{sec:trainSVM}, we must ensure that $f^{(d_p)}(S_T (\rho))$ does not differ substantially from the approximate phase classifier $f^{(d_p)}(\rho)$ from Assumption~\ref{as:phase-separation}. This is the content of the following auxiliary statement.

\begin{lemma} \label{lem:function-approximation-shadows}
Suppose that Assumption~\ref{as:phase-separation} is valid for a function on reduced $r$-body density matrices with the two constants $C \geq 1$ and ${d_p} \in \mathbb{N}$.
For any $0 < \epsilon < 1$,  classical shadows of size 
\begin{equation}
T= (32/3)d_p^2 C^2 12^r \left( r \left( \log (n) + \log (12) \right) + \log (1/\delta) \right)/\epsilon^2
\end{equation}
suffice to $\epsilon$-approximate $f^{(d_p)}(\rho)$ with high probability. 
In particular, for any density matrix $\rho \in \mathbb{H}_2^{\otimes n}$,
\begin{equation}
\left| f^{(d_p)}(S_T (\rho)) - f^{(d_p)}(\rho) \right| \leq \epsilon    
\end{equation}
with probability at least $1-\delta$ (over the randomized measurement settings and outcomes producing $S_T (\rho)$).
\end{lemma}

A proof can be found at the end of this subsection. With high probability, this statement ensures that existence of a well-conditioned phase separation implies the existence of a separating hyperplane in shadow feature space. 
This, in turn, is enough to ensure that the SVM training stage can be executed perfectly: solving the training problem~\eqref{eq:training-problem} efficiently yields a separating hyperplane parametrization $\vct{\alpha}_\sharp$ that (1) lies in the subspace $\mathbb{R}^N$ of $\mathcal{F}^{\text{(shadow)}}$ spanned by the $N$ training vectors, and (2) performs at least as well as $\vct{\alpha}_f$. Since we are guaranteed that $\vct{\alpha}_f$ separates training data perfectly and achieves zero training error, $\vct{\alpha}_\sharp$ must be at least as good: $\mathrm{E}_{\mathrm{tr}}(\vct{\alpha}_\sharp) \leq \mathrm{E}_{\mathrm{tr}} (\vct{\alpha}_f) =0$ with high probability.
The main result of this section formalizes this observation.

\begin{proposition} \label{prop:zero-training-error}
Suppose that Assumption~\ref{as:phase-separation} is valid for some function on reduced $r$-body density matrices with normalization constant $C$ and truncation degree $d_p$.
Then, for $\delta \in (0,1)$, a (joint) classical shadow size
$T = (512/3)d_p^2 C^2 12^r \left( r \left( \log (n) + \log (12) \right) + \log (N/\delta) \right)$
ensures that we can achieve zero training error when solving~\eqref{eq:training-problem} with squared margin constant
$\Lambda^2 = 4 \left( 2rn \right)^{r d_p} d_p^{d_p} C^2$.
\end{proposition}

The extra constraint $\Lambda^2 \geq \langle \vct{\alpha}_f, \vct{\alpha}_f \rangle_{\mathcal{F}^{\text{(finite)}}}$ ensures that the ideal separating hyperplane is a feasible point of the training problem~\eqref{eq:training-problem}.

\begin{proof}[Proof of Proposition~\ref{prop:zero-training-error}]

We establish the claim not for the shadow kernel itself ($k^{\text{(shadow)}}(\cdot,\cdot)$), but for large finite-dimensional approximations ($k^{\text{(finite)}}(\cdot,\cdot)$) thereof.
We begin by utilizing Eq.~\eqref{eq:polynomial-approximation} that approximates the nonlinear function $f(\rho)$ by a finite power series $f^{(d_p)}(\rho)$ with the approximation error,
\begin{equation}
    |f(\rho) - f^{(d_p)}(\rho)| \leq 0.25.
\end{equation}
For each $\ell \in \{1,\ldots,N\}$, we invoke Lemma~\ref{lem:function-approximation-shadows} using the truncated Taylor series to conclude
\begin{equation}
\mathrm{Pr} \left[ | f^{(d_p)}(\rho_\ell) - f^{(d_p)} \left( S_T (\rho_\ell) \right) | \geq 0.25\right] \leq \delta/N,
\end{equation}
provided that
$T= (512/3)d_p^2 C^2 12^r \left( r \left( \log (n) + \log (12) \right) + \log (N/\delta) \right)$.
Triangle inequality and a union bound allows us to combine these approximation guarantees into a single statement:
\begin{equation}
\max_{1 \leq \ell \leq N} \left| f (\rho_\ell) - f^{(d_p)} \left( S_T (\rho_\ell) \right) \right| \leq 0.5 \quad \text{with probability (at least) $1-\delta$.}
\label{eq:training-error-aux1}
\end{equation}
Let us
condition on this desirable event and also assume hat the cutoff values of our finite kernel approximation are large enough, i.e.\ $D \geq d_p$, $R \geq r$). Then, the function
\begin{equation}
    2f^{(d_p)}(S_T (\rho_\ell)) = \langle \vct{\alpha}_f, \phi^{\text{(finite)}}(S_T (\rho_\ell)) \rangle
\end{equation} 
describes a linear function in feature space $\mathcal{F}^{\text{(finite)}}$ that is guaranteed to achieve zero training error. Indeed, combine Eq.~\eqref{eq:classification-assumption} and Eq.~\eqref{eq:training-error-aux1} to ensure $\left| 2f^{(d_p)} (S_T (\rho_\ell) ) \right| \geq 2(\left| f(\rho_\ell) \right| - \left| f(\rho_\ell) - f^{(d_p)}(S_T (\rho_\ell)) \right|) \geq 2(1 -0.5) =1$ and, moreover, $\mathrm{sign} \left( f^{(d_p)}(S_T (\rho_\ell))\right) = \mathrm{sign} \left( f (\rho_\ell) \right) = y (\rho_\ell) \in \left\{ \pm 1 \right\}$
for all $\ell \in \{1,\ldots,N\}$. In turn,
\begin{subequations}
\begin{align}
\sum_{\ell=1}^N \max \left\{0,1 - y (\rho_\ell) \langle \vct{\alpha}_f, \phi^{\text{(finite)}} \left( S_T (\rho_\ell) \right) \rangle \right\} 
&= \sum_{\ell=1}^N \max \left\{ 0, 1- \mathrm{sign} \left( f^{(d_p)} (S_T (\rho_\ell))\right) 2 f^{(d_p)} \left( S_T (\rho_\ell) \right) \right\} \\
&= \sum_{\ell=1}^N \max \left\{0,1-\left| f^{(d_p)}\left( S_T (\rho_\ell) \right) \right| \right\} =0.
\end{align}
\end{subequations}
Since zero is the smallest possible training error,
the minimizer of the original training problem~\eqref{eq:training-problem} must also achieve zero, provided that $\vct{\alpha}_f$ is actually a feasible point of this optimization. We can, however, ensure this by choosing the squared margin constant large enough. Eq.~\eqref{eq:separator-size} and Assumption~\ref{as:phase-separation} ensures
\begin{equation}
\langle \vct{\alpha}_f, \vct{\alpha}_f \rangle_{\mathcal{F}^{\text{(finite)}}} \leq 4 \left( 2rn \right)^{r d_p} d_p^{d_p} C^2.
\end{equation}
Choosing a squared margin size $\Lambda^2$ that exceeds this bound ensures that $\vct{\alpha}_f$ is indeed a feasible point of the training problem~\eqref{eq:training-problem} and the claim follows.
\end{proof}

We conclude our discussion on training with shadow kernels by providing a rigorous proof of the auxiliary statement.

\begin{proof}[Proof of Lemma~\ref{lem:function-approximation-shadows}]

It suffices to analyze implications of Lemma~\ref{lem:subsystem-approximation}: for $\eta,\delta \in (0,1)$
\begin{equation}
T \geq (8/3) 12^r \left( r \left( \log (n) + \log (12) \right) + \log (1/\delta) \right)/\eta^2 \; \Rightarrow \; \max_{A \subset \{1,\ldots,n\}, |A|\leq r}\left\| \mathrm{tr}_{\neg A} \left( \sigma_T (\rho) \right) - \mathrm{tr}_{\neg A} (\rho) \right\|_1 \leq \eta \label{eq:reduced-densities-claim}
\end{equation}
with probability at least $1-\delta$. Here, $\|\cdot \|_1$ denotes the trace norm. 
Abbreviate $\mathrm{tr}_{\neg A_i} \left( \sigma_T (\rho) \right)$ and $\mathrm{tr}_{\neg A_i} (\rho)$ as $\sigma_{A_i}$ and $\rho_{A_i}$, respectively.
A combination of triangle inequalities and Matrix Hoelder ($\mathrm{tr}(XY) \leq \|X \|_\infty \|Y \|_1$) asserts
\begin{subequations}
\label{eq:approximation-aux1}
\begin{align}
& \left| f^{(d_p)} (\rho)- f^{(d_p)} \left( S_T (\rho) \right)   \right|\\
& \leq \sum_{d=0}^{d_p} \frac{1}{d!} \sum_{A_1,\ldots,A_d \subset \{1,\ldots,n\}, |A_i| \leq r} \left| \mathrm{tr} \left( O_{A_1,\ldots,A_r} \left( \rho_{A_1} \otimes \cdots \otimes \rho_{A_d}- \sigma_{A_1} \otimes \cdots \otimes \sigma_{A_d}   \right) \right)\right|  \\
& \leq \sum_{d=0}^{d_p} \frac{1}{d!} \sum_{A_1,\ldots,A_d \subset\{1,\ldots,n\},|A_i| \leq r} \left\| O_{A_1,\ldots,A_d} \right\|_\infty 
\left\|   \rho_{A_1} \otimes \cdots \otimes \rho_{A_d} - \sigma_{A_1} \otimes \cdots \otimes \sigma_{A_d}\right\|_1.
\end{align}
\end{subequations}
Next, we fix a trace norm contribution and 
use a telescoping trick ($A_1 \otimes A_2 - B_1 \otimes B_2 = (A_1 - B_1) \otimes A_2 + B_1 \otimes (A_2 - B_2)$), as well as a triangle inequality and $\|\rho_{A_i} \|_1 = \mathrm{tr}(\rho_{A_i}) =1$ to infer
\begin{subequations}
\begin{align}
& \left\|   \rho_{A_1} \otimes \cdots \otimes \rho_{A_d} - \sigma_{A_1} \otimes \cdots \otimes \sigma_{A_d}\right\|_1 \\
&= \left\| \left( \rho_{A_1} - \sigma_{A_1} \right) \otimes \rho_{A_2}\otimes \cdots \otimes \rho_{A_d}
+ \sigma_{A_1} \otimes \left( \rho_{A_2} \otimes \cdots \otimes \rho_{A_d} - \sigma_{A_2} \otimes \cdots \otimes \sigma_{A_d} \right) \right\|_1 \\
&\leq \| \rho_{A_1} - \sigma_{A_1}\|_1 \| \rho_{A_2} \|_1 \cdots \| \rho_{A_d} \|_1
+ \| \sigma_{A_1} \|_1 \left\| \rho_{A_2} \otimes \cdots \otimes \rho_{A_d} - \sigma_{A_1}\otimes \cdots \otimes \sigma_{A_d} \right\|_1 \\
&\leq \| \rho_{A_1}- \sigma_{A_1} \|_1 + \left( 1 + \| \rho_{A_1}-\sigma_{A_1} \|_1 \right) \left\| \rho_{A_2} \otimes \cdots \otimes \rho_{A_d} - \sigma_{A_1}\otimes \cdots \otimes \sigma_{A_d} \right\|_1 \\
&\leq \eta + (1+\eta) \| \rho_{A_2} \otimes \cdots \otimes \rho_{A_d} - \sigma_{A_1} \otimes \sigma_{A_d}\|_1.
\end{align}
\end{subequations}
The last line follows from Rel.~\eqref{eq:reduced-densities-claim}. Iterating this simplification procedure ensures
\begin{align}
\left\|   \rho_{A_1} \otimes \cdots \otimes \rho_{A_d} - \sigma_{A_1} \otimes \cdots \otimes \sigma_{A_d}\right\|_1
\leq \eta  \sum_{k=0}^{d-1}(1+\eta)^k
= (1+\eta)^d -1.
\end{align}
According to Rel.~\eqref{eq:reduced-densities-claim}, such an upper bound is valid for every trace norm contribution in Eq.~\eqref{eq:approximation-aux1}.
This allows us to obtain
\begin{subequations}
\begin{align}
    \left| f^{(d_p)} (\rho)- f^{(d_p)} \left( S_T (\rho) \right)   \right| &\leq \sum_{d=0}^{d_p} \frac{1}{d!} \sum_{A_1,\ldots,A_d \subset\{1,\ldots,n\},|A_i| \leq r} \left\| O_{A_1,\ldots,A_d} \right\|_\infty \left[ (1+\eta)^d - 1 \right]\\
    &\leq \left[ (1+\eta)^{d_p} - 1 \right] \sum_{d=0}^{d_p} \frac{1}{d!} \sum_{A_1,\ldots,A_d \subset\{1,\ldots,n\},|A_i| \leq r} \left\| O_{A_1,\ldots,A_d} \right\|_\infty\\
    &= C \left[ (1+\eta)^{d_p} - 1 \right].
\end{align}
\end{subequations}
Here, we have used Assumption~\ref{as:phase-separation}.
Finally, by choosing $\eta = \epsilon / (2C d_p)$, we can see that
\begin{equation}
     \left| f^{(d_p)} (\rho)- f^{(d_p)} \left( S_T (\rho) \right)   \right|  \leq C\left[ \left(1+\frac{\epsilon}{2C d_p}\right)^{d_p} - 1\right] \leq C [\exp(\epsilon / 2C) - 1] \leq \epsilon.
\end{equation}
The second inequality follows from $(1+x / n)^n \leq \exp(x), \forall |x| \leq n, n \geq 1$.
The third inequality utilizes $\exp(x) \leq 1+2x, \forall x \in [0, 1]$.
The claim of Lemma~\ref{lem:function-approximation-shadows} now follows from inserting this specific choice of $\eta$ into Rel.~\eqref{eq:reduced-densities-claim}.
\end{proof}

\subsection{Prediction based on shadow kernels} \label{sec:predict-shadowkernel}

We now have all pieces in place to prove strong bounds on the prediction error of a SVM based on shadow kernels. The main result of this section will be a consequence of Theorem~\ref{thm:svm-prediction}.
For fixed parameters $\tau,\gamma = 1$, the shadow kernel~\eqref{eq:shadow-kernel} (and finite approximations thereof) is always bounded when applied to classical shadows. Eq.~\eqref{eq:kernel-upper-bound} (under $\tau = \gamma = 1$) asserts
\begin{align}
k^{\mathrm{(shadow)}} \left( S_T (\rho_1), \tilde{S}_T (\rho_2) \right) \leq \exp \left( \exp \left( 5 \right) \right)
\end{align}
for any $T$ and quantum states $\rho_1,\rho_2$. This bound readily extends to finite dimensional approximations $k^{\text{(finite)}}(\cdot, \cdot)$. 
Next, we need to specify a distribution. We assume that $\tilde{\mathcal{D}}$ is a distribution over $n$-qubit quantum states $\rho$ that either belong to phase $A$ or phase $B$. We sample quantum states $\rho_\ell \sim \tilde{\mathcal{D}}$ accordingly, but are not permitted to process them directly.
Instead, we obtain a (randomly generated) classical shadow of size $T$. Denote the raw data by $S_T (\rho_\ell)$ which allows us to produce a state approximation $\sigma_T (\rho_\ell)$.
 We do, however, require that we have direct access to the label $y(\rho_\ell) \in \left\{ \pm 1 \right\}$ associated with the phase of $\rho_\ell$.
This produces a joint distribution over input data $S_T (\rho_\ell)$ and the label $y(\rho_\ell)$ which we call $\mathcal{D}$.
In summary, we assume that training data and new data are generated independently from this data distribution: $\left(S_T (\rho_1),y(\rho_1)\right),\ldots,\left( S_T (\rho_N),y(\rho_N)\right), \left( S_T (\rho),y\right) \sim \mathcal{D}$.
We are now ready to combine Theorem~\ref{thm:svm-prediction} (the prediction error is bounded by the training error) and Proposition~\ref{prop:zero-training-error} (the training error vanishes if a good phase classifier exists) 
to obtain a powerful result about generalization.

\begin{corollary} \label{cor:shadow-prediction}
Fix $\delta,\epsilon \in (0,1)$ and
suppose there exists an analytic function on reduced $r$-body density matrices that can distinguish phases: $f (\rho) >1$ if $\rho \in \text{phase $A$}$ and $f(\rho) < -1$ else if $\rho \in \text{phase $B$}$. Let $C$ be the normalization constant and $d_p$ be the truncation degree given in Assumption~\ref{as:phase-separation}.
Suppose that we obtain identically distributed training data $\left(S_T (\rho_1),y(\rho_1)\right),\ldots,\left( S_T (\rho_N),y(\rho_N)\right) \sim \mathcal{D}$ such that
\begin{subequations}
\begin{align}
T &\geq (512/3)d_p^2 C^2 12^r \left( r \left( \log (n) + \log (12) \right) + \log (N / \delta) \right) 
\quad \text{and} \\
N &\geq 256 \left( 2rn \right)^{r d_p} d_p^{d_p} C^2 \exp ( \exp (5))  \log (4/\delta)/\epsilon^2. 
\end{align}
\end{subequations}
Then, solving the training problem~\eqref{eq:training-problem} for the shadow kernel with squared margin constant $\Lambda^2 = 4 \left( 2rn \right)^{r d_p} d_p^{d_p} C^2$ will produce a hyperplane $\vct{\alpha}_\sharp \in \mathbb{R}^N$ in shadow feature space that achieves zero training error with probability (at least) $1-\delta/2$.
Conditioned on perfect training, the resulting classifier 
\begin{equation}
    y_\sharp \left( S_T (\rho) \right) = \mathrm{sign} \big( \sum_{\ell=1}^N \left[ \vct{\alpha}_\sharp \right]_\ell k^{\mathrm{(shadow)}} (S_T (\rho_\ell),S_T (\rho) ) \big) \in \left\{ \pm 1 \right\}
\end{equation} achieves, with probability (at least) $1-\delta/2$,
\begin{equation}
\mathrm{Pr}_{(S_T (\rho),y(\rho))} \left[ y_\sharp \left( S_T (\rho) \right) \neq y(\rho) \right] \leq \epsilon~.
\end{equation}
\end{corollary}

The total probability of success is (at least) $1-\delta$ and follows from a union bound over either desirable event failing.
Theorem~\ref{cor:shadow-prediction} is contingent on four core assumptions:
\begin{enumerate}
\item It must be possible to distinguish phases $A$ and $B$ by evaluating a well-conditioned analytical function on reduced $r$-body density matrices. The coefficients in the power series of the analytical function should also be bounded, but explicit knowledge is \emph{not} necessary. This is the content of Assumption~\ref{as:phase-separation}.

\item We use classical shadow raw data to read-in training data ($\rho_\ell \mapsto S_T (\rho_\ell)$) and process new states in the prediction phase ($\rho \mapsto S_T (\rho)$). We assume that each classical shadow arises from $T$ randomized single-qubit Pauli measurements on independent state copies. 
The larger $T$, the more accurate these representations become. Theorem~\ref{cor:shadow-prediction} requires $T \geq (512/3)d_p^2 C^2 12^r \left( r \left( \log (n) + \log (12) \right) + \log (N/\delta) \right)  = \mathcal{O} \left( r 12^r d_p^2 C^2 \log (nN/\delta) ) \right)$.
If $r, C, d_p$ are constants, this resolution only scales polylogarithmically in system size $n$ because $N$ scales polynomially in $n$; see the next bullet point.

\item The training data size must not be too small either. We need to have a training data size $N$ of order at least $\left( 2rn \right)^{r d_p} d_p^{d_p} C^2 \exp ( \exp (5))  \log (4/\delta)/\epsilon^2$. As long as $r, C, d_p$ are constants (independent of system size $n$), this requirement simplifies to
$N = \mathcal{O} \left( n^{r d_p} \log (1/\delta)/\epsilon^2 \right)$. Hence, the number scales polynomially in system size $n$.

\item The squared margin constant also scales polynomially with system size $n$: $\Lambda^2 = 4 \left( 2rn \right)^{rd_p} d_p^{d_p} C^2 = \mathcal{O}\left( n^{r d_p} \right)$ if $r, C, d_p =\mathrm{const}$. This is equivalent to demanding that the minimal margin $2/\Lambda$ scales inverse polynomially in system size $n$.

\end{enumerate}

Corollary~\ref{cor:shadow-prediction} does not only bound a hypothetical training error. The required shadow size $T$ and training data size $N$ both scale favorably in the number of qubits $n$. This also ensures that the numerical costs behind this procedure remain tractable for a wide range of system sizes.
The costs associated with storage (classical shadows are sums of $T$ elementary tensor products), training (can be reduced to a QCQP in $N$ dimensions per Section~\ref{sec:trainSVM}) and prediction (execute Formula~\eqref{eq:prediction-function}) all scale polynomially in system size $n$, shadow size $T$, and training data size $N$. 

\begin{proof}[Proof of Corollary~\ref{cor:shadow-prediction}]

Again, we establish the claim for large, but finite-dimensional, approximations to the shadow kernel ($1 \leq d_p \ll D < \infty$ and $1 \leq r \ll R <\infty$).
Fix $\delta \in (0,1)$ (probability of failure) and $\epsilon \in (0,1)$ (bound on average prediction error).
Consider the data distribution $\left( S_T (\rho),y(\rho) \right) \sim \mathcal{D}$, the kernel $k^{\text{(finite)}} (\cdot,\cdot)$ -- which obeys $k^{\text{(finite)}} \left( S_T (\rho), S_T (\rho) \right) \leq \exp \left( \exp (5) \right)$ -- and a squared margin constant $\Lambda^2$ to be specified later. 
Assume $\Lambda^2 \exp ( \exp (5 ) )\geq 1$ for simplicity (the other case is similar).
Then, for training data size $N$, Theorem~\ref{thm:svm-prediction} asserts
\begin{equation}
\mathrm{Pr}_{(S_T (\rho),y(\rho))\sim \mathcal{D}} \left[ y_\sharp \left( S_T (\rho) \right) \neq y (\rho) \right] 
\leq \frac{1}{N}\mathrm{E}_{\mathrm{tr}} (\vct{\alpha}_\sharp) + 8\sqrt{ \Lambda^2 \exp ( \exp (5) )\frac{\log (4/\delta)}{N}},
\end{equation}
with probability (at least) $1-\delta/2$. Choosing $N$ large enough allows us to suppress the second contribution beneath the desired approximation error bound:
\begin{equation}
N \geq 64  \Lambda^2\exp ( \exp (5 ))  \log (4/\delta)/\epsilon^2
\quad \Rightarrow \quad \mathrm{Pr}_{(S_T (\rho),y(\rho))\sim \mathcal{D}} \left[ y_\sharp \left( S_T (\rho) \right) \neq y (\rho) \right]
\leq \frac{1}{N}\mathrm{E}_{\mathrm{tr}}(\vct{\alpha}_\sharp) + \epsilon, \label{eq:prediction-aux1}
\end{equation}
with probability (at least) $1-\delta/2$. Here, $\mathrm{E}_{\mathrm{tr}} (\vct{\alpha}_\sharp)$ is the training error obtained from solving problem~\eqref{eq:training-problem} for $N$ independently sampled training data points $\left(S_T (\rho_1),y(\rho_1)\right),\ldots,\left( S_T (\rho_N),y(\rho_N)\right) \sim \mathcal{D}$. Proposition~\ref{prop:zero-training-error} asserts that this training error can vanish with high probability, provided that a well-conditioned analytical function on reduced $r$-body density matrices exists that can distinguish the phases (see Assumption~\ref{as:phase-separation}). 
The classical shadow size $T$ and the squared margin constant $\Lambda^2$ depend on the number of body $r$, the normalization constant $C$, and the truncation degree $d_p$ of this classifier:
\begin{equation}
\left.
\begin{array}{ccl}
T &\geq & (512/3)d_p^2 C^2 12^r \left( r \left( \log (n) + \log (12) \right) + \log (N/\delta) \right)  \\
\Lambda &\geq & 4 \left( 2rn \right)^{r d_p} d_p^{d_p} C^2
\end{array}
\right\} \quad \Rightarrow \mathrm{E}_{\mathrm{tr}}(\vct{\alpha}_\sharp)= 0
\end{equation}
with probability (at least) $1-\delta/2$. The claim now follows from inserting this squared margin size into the expression~\eqref{eq:prediction-aux1} for training data size.
\end{proof}

\section{Classifying SPT phases with $O(2)$ symmetry using a few-body observable}
\label{sec:tasaki}

\subsection{Symmetry-protected topological phases} \label{sec:tasaki1}

We consider a scenario similar to that of Section~\ref{sub:smoothness},
namely, a family of Hamiltonians $H(x)$ parameterized by $x$. We
additionally enforce that $H(x)$ be invariant under certain symmetry
transformations, which can include tensor products of on-site rotations,
``spatial'' transformations permuting the sites, or antiunitary
maps characterizing time-reversal. These additional symmetry constraints
allow for a fine-grained characterization of $H(x)$ into various
symmetry-protected topological (SPT) phases. Removing said constraints
reduces this characterization to the coarser one involving purely
topological phases. Similar to the coarser characterization, ground
states of $H(x)$ remain in a particular SPT when the parameters $x$
are varied continuously, as long as the spectral gap of the Hamiltonian
remains finite. In other words, the gap has to close at some $x$
in order for the ground states to transition into another phase.
When there is a constant spectral gap, it is expected that an operator acting on a local region larger than some constant size independent of the full system size $n$ can classify different SPT phases.
The existence of a classifying function of local density matrices has been rigorously established for a handful of cases:  $U(1)$-symmetric systems in two dimensions (either noninteracting fermionic \cite{Kitaev2006,zhang2017quantum} or interacting \cite{Hastings2015,Kapustin2020,Bachmann2020}), and certain spin-$1$ chains in one dimension \cite{Bachmann2014a,tasakiPRL2018,Tasaki2020}.

SPT phases of one-dimensional spin chains with unique ground states,
symmetric under tensor-product unitaries forming a symmetry
group $G$, are in one-to-one correspondence with the various projective representations realized
by $G$ \cite{Chen2011}. Projective representations are
those in which the group's multiplication table is decorated with
phases in a way that is consistent with associativity \cite{Arovas}. A genuine (i.e.,
linear) representation corresponds to the unique trivial projective
representation.

Consider, for example, spin chains symmetric under $G=SO(3)$. This
group admits two distinct classes of projective representations:
one class corresponds to integer spin, and one corresponds to half-integer
spin. Thus, there are two different phases for such chains --- the
trivial phase and the ``Haldane phase'' \cite{Haldane1983a,Chen2011}.

Relaxing the symmetry group down to its $O(2)$ subgroup maintains
the two-phase classification, because $O(2)$ also admits two projective representations
\cite{Chen2013}. In fact, one can relax the symmetry all the way
down to the simplest dihedral subgroup $Z_{2}\times Z_{2}$ \cite{Gu2009,Pollmann2010}; such a classification is similar to that of the model in Appendix~\ref{sec:numdetail-phases}. We investigate systems admitting the larger $O(2)$ symmetry below, noting that the work we rely on \cite{tasakiPRL2018,Tasaki2020} also studies symmetry groups that include spatial inversion and time reversal.

\subsection{$O(2)$-symmetric qutrit spin chains} \label{sec:tasaki2}

The representative states for each of the two $O(2)$-symmetric phases for qutrit spin chains are the product state, representing the trivial phase, and the valence-bond-solid (VBS) state \cite{Affleck1988}, admitting a projective representation of the symmetry \cite{Tasaki2020} and thus representing the Haldane phase. It has long been known that the expectation value of a nonlocal ``twist'' operator $O_L$ \cite{Totsuka1995,nakamura2002a} distinguishes these two representative states: $\text{sign}(\langle O_L \rangle)$ is $+1$ for the product state, and $-1$ for the VBS state. We will see later that, by continuity arguments, this sign will stay constant for other states within the same phase.

In order to work efficiently, our phase classification algorithms require a \textit{local} operator whose expectation value (a) has the same sign as that of $O_L$; and (b) is above or below a margin (here, $1/2$), in order to determine the required accuracy of the classical shadows.
Recently, criterion (a) was explicitly demonstrated by Tasaki \cite{tasakiPRL2018,Tasaki2020} using a local version $O_\ell$ of the twist, see Eq.~\eqref{eq:twist} below. We collect relevant parts of his results to prove both criteria in the theorem below.
Due to the existence of a local operator for classifying the SPT phases, our ML algorithms are guaranteed to predict the SPT phases accurately based on the proof given in Appendix~\ref{sec:proofthmPHASEC}.

\begin{theorem}
Consider the triple $\left\{ H(x),\ket{\psi(x)} ,\Delta(x)\right\} $ containing
$(2L+2)$-site spin-one chains with periodic boundary conditions
\begin{equation}
H(x)=\sum_{j=-(L-r)}^{L-r+1}h_{j}(x)+h_{-L}(x)+h_{L+1}(x)
\end{equation}
that admit corresponding unique ground states $\ket{\psi(x)}$ and spectral
gaps $\Delta(x) \geq \gamma = \Omega(1)$,
bounded interaction strength $\norm{h_j(x)}_\infty \leq R = \mathcal{O}(1)$,
and whose terms $h_{j}(x)$ are supported on sites $k$ such that $|j-k|\leq r = \mathcal{O}(1)$.
Assume that $H(x)$ is $O(2)$-symmetric, with the symmetry group
generated by\begin{enumerate}
\item a collective $z$-axis rotation by any angle, and
\item an $x$-axis rotation by $\pi$.
\end{enumerate}
There exists a few-body observable $A$, such that for all $x$, we have
\begin{subequations}
\label{eq:toprove}
\begin{align}
    \mathrm{sign}(\bra{\psi(x)} A \ket{\psi(x)}) & = \mathrm{sign}\left( \bra{\psi(x)} O_L \ket{\psi(x)}\right), \quad \text{as well as}
    \label{eq:toprove-2}\\ |\bra{\psi(x)} A \ket{\psi(x)}| & \geq 1/2    \label{eq:toprove-1}~.
\end{align}
\end{subequations}
\end{theorem}

\begin{proof}
We use spin-one operators $S^{(\alpha)}$ with $\alpha\in \{x,y,z\}$ that
have eigenvalues $\{0,\pm 1 \}$ and satisfy angular-momentum commutation
relations $[S^{(x)},S^{(y)}]=\mathrm{i}S^{(z)}$. Eigenstates of $S^{(z)}$
are denoted by $|\sigma\rangle$ with $\sigma\in\{0,\pm 1\}$. A rotation
around axis $\alpha$ is a unitary operator generated by the corresponding
$S^{(\alpha)}$. The two symmetry group generators are, for $\theta\in[0,2\pi)$,
\begin{equation}
U\left(\theta\right)=\bigotimes_{j=-L}^{L+1}e^{-\mathrm{i}\theta S_{j}^{(z)}}\,\,\,\,\,\,\,\,\,\,\,\,\,\,\,\,\,\,\,\,\,\,\,\,\,\,\,\,\,\,\,\,\,\,\,\,\,\,\,\text{and}\,\,\,\,\,\,\,\,\,\,\,\,\,\,\,\,\,\,\,\,\,\,\,\,\,\,\,\,\,\,\,\,\,\,\,\,\,\,\,V=\bigotimes_{j=-L}^{L+1}e^{-\mathrm{i}\pi S_{j}^{(x)}}\,.\label{eq:rotinv}
\end{equation}
By assumption, both symmetries commute with each Hamiltonian term
$h_{j}$; we will explicitly use both to prove the theorem. We will also need superimposed versions $S^{(\pm)}=S^{(x)}\pm \mathrm{i}S^{(y)}$,
which satisfy
\begin{equation}
e^{\mathrm{i}\phi S^{(z)}}S^{(\pm)}e^{-\mathrm{i}\phi S^{(z)}}=S^{(\pm)}e^{\pm \mathrm{i}\phi}\,.\label{eq:spin-rotation}
\end{equation}

The family of unitary twist operators \cite{Affleck1986a}, acting on an interval of $2 \ell$ spins centered at the origin, is
\begin{align}
O_{\ell}=\bigotimes_{k\,,\,\left|k-\half\right|\leq\ell+\half}\exp\left(-\mathrm{i}2\pi\frac{k+\ell}{2\ell+1}S_{k}^{(z)}\right)\,.\label{eq:twist}
\end{align}
Each site's rotation is by a multiple of $2\pi/(2\ell+1)$ that is
proportional to the site index, forming the namesake twist pattern. The $\ell=L$ case reduces to the aforementioned nonlocal twist operator $O_L$, while $\ell \ll L$ are its local versions.

Suppressing $x$ dependence, the key property is that the twisted
ground state $O_{\ell}|\psi\rangle$ has energy close to that of the ground state.
In particular, there exists $C_0, C_1 > 0$, such that for all $\ell \geq C_0$, Lemma~\ref{lem:vanenergy} below yields
\begin{equation}
\langle\psi|O_{\ell}HO_{\ell}^{\dagger}|\psi\rangle-\langle\psi|H|\psi\rangle\leq\frac{C_1}{\ell}\,.\label{eq:energies}
\end{equation}
The ground state is unique by our assumption of a gap, so the twisted ground state must then become proportional to the ground state as $\ell\to\infty$. In other
words, the magnitude of their overlap must be close to one as long as $\ell \geq C_0$,
\begin{equation}
\left|\langle\psi|O_{\ell}|\psi\rangle\right|^{2}\geq1-\frac{C_1}{\Delta \ell}\,;\label{eq:overlap}
\end{equation}
see Lemma~\ref{lem:highoverlap} 
below.
The phase of this overlap is either $0$ or $\pi$ because the $\pi$-rotation
$V$ leaves the ground state invariant:
\begin{equation}
\langle\psi|O_{\ell}|\psi\rangle=\langle\psi|V^{\dagger}O_{\ell}V|\psi\rangle=\langle\psi|O_{\ell}^{\dagger}|\psi\rangle=\overline{\langle\psi|O_{\ell}|\psi\rangle}\in\mathbb{R}\,.
\end{equation}
Hence, the few-body Hermitian observable
$A = (O_\ell + O^\dagger_\ell) / 2$ with $\ell = \max( 4 \gamma / (3 C_1), C_0)$ satisfies
\begin{equation}
    |\langle\psi| A |\psi\rangle| = |\langle\psi| O_\ell |\psi\rangle| \geq \sqrt{1 - \frac{C_1}{\Delta \ell}} \geq \frac{1}{2}~,\label{eq:sign}
\end{equation}
proving Eq.~(\ref{eq:toprove-1}). Note that the required value of $\ell$ depends on the gap, and thus also on $x$.

To prove Eq.~(\ref{eq:toprove-2}), we need to show that the sign of the twist's expectation value remains the same for any $\ell \geq \max( 4 \gamma / (3 C_1), C_0)$. To do this, first notice that, when $\ell$ is relaxed to be a nonnegative real, the twist (\ref{eq:twist}) is \textit{continuous} in $\ell$. (This can be verified, e.g., by studying the twist's eigenvalues.) 
Continuity implies that the expectation value cannot change sign; otherwise, it would have to cross zero, thus violating Eq.~(\ref{eq:sign}). Therefore, the sign remains the same, confirming Eq.~(\ref{eq:toprove-2}). Similarly, by continuity in $\ell$ and $x$, the expectation value maintains its sign within each phase.
\end{proof}

The above argument is contingent on two auxiliary statements, which we now prove.

\begin{lemma}[Vanishing~energy~difference \cite{tasakiPRL2018}; Eq.~\eqref{eq:energies}] \label{lem:vanenergy}
For constants $C_0, C_1$, as long as $\ell \geq C_0$, we have
\begin{equation}
\langle\psi|O_{\ell}HO_{\ell}^{\dagger}|\psi\rangle-\langle\psi|H|\psi\rangle\leq\frac{C_1}{\ell}.
\end{equation}
\end{lemma}
\begin{proof}
Using the variational principle (which says that the difference in energy between any state and the ground state is nonnegative), plugging in $O_{\ell}$ and $H$, applying
$\langle\psi|O|\psi\rangle\leq\left\Vert O\right\Vert _{\infty}$,
and distributing the norm over the sum yields\begin{subequations}
\begin{align}
\langle\psi|O_{\ell}HO_{\ell}^{\dagger}|\psi\rangle-\langle\psi|H|\psi\rangle & \leq\langle\psi|\left(O_{\ell}HO_{\ell}^{\dagger}+O_{\ell}^{\dagger}HO_{\ell}-2H\right)|\psi\rangle\\
 & =\sum_{j=-\left(\ell+r\right)}^{\ell+r+1}\langle\psi|\left(O_{\ell}h_{j}O_{\ell}^{\dagger}+O_{\ell}^{\dagger}h_{j}O_{\ell}-2h_{j}\right)|\psi\rangle\\
 & \leq\sum_{j=-\left(\ell+r\right)}^{\ell+r+1}\left\Vert O_{\ell}h_{j}O_{\ell}^{\dagger}+O_{\ell}^{\dagger}h_{j}O_{\ell}-2h_{j}\right\Vert _{\infty}\label{eq:sumj}
\end{align}
\end{subequations}Next, we use the finite support and rotational
invariance of $h_{j}$ from Eq.~(\ref{eq:rotinv}) to rotate the
twist $O_{\ell}$,\begin{subequations}
\begin{align}
O_{\ell}h_{j}O_{\ell}^{\dagger} & =O_{\ell}U\left(\theta_j\right)h_{j}U^{\dagger}\left(\theta_j\right)O_{\ell}^{\dagger}\\
 & =\left(\bigotimes_{\left|k-j\right|\leq r}e^{-\mathrm{i}\left(\frac{2\pi}{2\ell+1}\left[k+\ell\right]+\theta_j\right)S_{k}^{(z)}}\right)h_{j}\left(\bigotimes_{\left|k-j\right|\leq r}e^{\mathrm{i}\left(\frac{2\pi}{2\ell+1}\left[k+\ell\right]+\theta_j\right)S_{k}^{(z)}}\right)\,,
\end{align}
\end{subequations}where we pick $\theta_j=-\frac{2\pi}{2\ell+1}\left(j+\ell\right)$
for each $j$. That way, the twist does not affect site $j$, with
\begin{align}
O_{\ell}h_{j}O_{\ell}^{\dagger} & =e^{\mathrm{i}\frac{2\pi}{2\ell+1}M_{j}}h_{j}e^{-\mathrm{i}\frac{2\pi}{2\ell+1}M_{j}}\,,\,\,\,\,\,\,\,\,\,\,\,\,\,\,\,\,\,\,\,\text{and}\,\,\,\,\,\,\,\,\,\,\,\,\,\,\,\,\,\,\,M_{j}=\sum_{|k-j|\leq r}\left(j-k\right)S_{k}^{(z)}\,.\label{eq:M}
\end{align}

We now expand $h_{j}$ as a polynomial in $\{S_{k}^{(z)},S_{k}^{(\pm)}\}$. This can be done because products of powers of these operators form a matrix basis for any operator on the chain. For a single site, the set $\{S^{(z)}S^{(\pm)},(S^{(+)})^{2}\}$, along with their complex conjugates and some powers of $S^{(z)}$, form the basis of nine matrix units for all $3\times3$ operators on the site. Tensor products of these operators therefore form a matrix-unit basis for all sites. The conjugation property (\ref{eq:spin-rotation}) and Eq.~(\ref{eq:M}) imply that each term in the expansion of $h_j$, upon conjugation by $O_\ell$, will be imparted with a phase that is some multiple $\mu$ of $2\pi/(2\ell+1)$. Combining all terms with the same phase into $h_{j,\mu}$, we have
\begin{equation}
e^{\mathrm{i}\frac{2\pi}{2\ell+1}M_{j}}h_{j,\mu}e^{-\mathrm{i}\frac{2\pi}{2\ell+1}M_{j}}=h_{j,\mu}e^{\mathrm{i}\frac{2\pi}{2\ell+1}\mu}\,.
\end{equation}
Moreover, $|\mu|\leq2\mu_{\text{max}}$, where $\mu_{\text{max}}=\sum_{\left|k-j\right|\leq r}\left|j-k\right|=r\left(r+1\right)$
is the largest eigenvalue of $M_{j}$. Plugging this in and expanding
the resulting cosine yields\begin{subequations}
\begin{align}
\left\Vert O_{\ell}h_{j}O_{\ell}^{\dagger}+O_{\ell}^{\dagger}h_{j}O_{\ell}-2h_{j}\right\Vert _{\infty} & =2\left\Vert \sum_{\left|\mu\right|\leq2r\left(r+1\right)}\left[\cos\left(\frac{2\pi}{2\ell+1}\mu\right)-1\right]h_{j,\mu}\right\Vert _{\infty}\\
 & \leq \left(\frac{2\pi}{2\ell+1}\right)^{2} \sum_{\left|\mu\right|\leq2r\left(r+1\right)}\mu^{2} \left\Vert h_{j,\mu} \right\Vert _{\infty}\,.
\end{align}
\end{subequations}
Since the spin operators form a matrix-unit basis, each $h_{j,\mu}$ is simply $h_j$ with some entries removed. Therefore, the norm of $h_{j,\mu}$ is bounded by $R$. Applying that and performing the remaining sum (\ref{eq:sumj}) over
$j$ yields
\begin{align}
\langle\psi|O_{\ell}HO_{\ell}^{\dagger}|\psi\rangle-\langle\psi|H|\psi\rangle & \leq\frac{\ell+r+1}{\left(2\ell+1\right)^{2}}4\pi^{2}R\left(\sum_{\left|\mu\right|\leq2r\left(r+1\right)}\mu^{2}\right)\,.
\end{align}
Thus, for $\ell \geq C_0$, the difference in energies between
the ground state and twisted ground state will be bounded by $C_1/\ell$,
where $C_0, C_1$ are two constants depending on the interaction range $r$ and norm bound $R$ of the Hamiltonian terms.
\end{proof}

\begin{lemma}[High overlap \cite{Tasaki2020}; Eq.~\eqref{eq:overlap}] \label{lem:highoverlap}
For constants $C_0, C_1$, as long as $\ell \geq C_0$, we have
\begin{equation}
    \left|\langle\psi|O_{\ell}|\psi\rangle\right|^{2}\geq1-\frac{C_1}{\Delta \ell}\,.
\end{equation}
\end{lemma}
\begin{proof}
All eigenvalues
of $H$ are bounded below by the sum of the ground state energy $E_{\text{gnd}}=\langle\psi|H|\psi\rangle$
and spectral gap $\Delta$,
\begin{align}
H & \geq E_{\text{gnd}}|\psi\rangle\!\langle\psi|+\left(E_{\text{gnd}}+\Delta\right)\left(\mathbb{I} -|\psi\rangle\!\langle\psi|\right)=E_{\text{gnd}} \mathbb{I} +\Delta\left(\mathbb{I} -|\psi\rangle\!\langle\psi|\right)\,.
\end{align}
Conjugating by $O_{\ell}$ and evaluating the result in the ground
state yields
\begin{align}
\langle\psi|O_{\ell}HO_{\ell}^{\dagger}|\psi\rangle & \geq E_{\text{gnd}}+\Delta\left(1-\left|\langle\psi|O_{\ell}|\psi\rangle\right|^{2}\right)\,.
\end{align}
Rearranging this and plugging in Lemma~\ref{lem:vanenergy} yields
the desired result.
\end{proof}

\bibliography{ref,ref_vva,ref_gt}
\bibliographystyle{abbrv}

\end{document}